%% file: ys-202605-008-no-go-theorem-for-spin-boson.tex
\documentclass[11pt,a4paper]{article}

\usepackage{fontspec}
\usepackage[titletoc,title,header]{appendix}

\usepackage{amsmath}
\usepackage{amssymb}
\usepackage{amsthm}
\usepackage{mathtools}
\usepackage{xparse}    % for \NewDocumentCommand

\usepackage[margin=1in]{geometry}

\usepackage{graphicx}
\usepackage{xcolor}

\usepackage{hyperref}
\hypersetup{
  colorlinks=true,
  linkcolor=blue,
  citecolor=blue,
  urlcolor=blue,
  pdfauthor={Yoshitsugu Sekine},
  pdftitle={No-Go Theorem for Quasiparticle BEC in the Spin-Boson Model}
}

\usepackage[]{natbib}
\theoremstyle{plain}
\newtheorem{thm}{Theorem}[section]
\newtheorem{lem}[thm]{Lemma}
\newtheorem{prop}[thm]{Proposition}
\newtheorem{cor}[thm]{Corollary}

\theoremstyle{definition}
\newtheorem{defn}[thm]{Definition}
\newtheorem{ex}[thm]{Example}

\theoremstyle{remark}
\newtheorem{rem}[thm]{Remark}

\input{mycommands.tex}

\title{No-Go Theorem for Quasiparticle BEC in the Spin-Boson Model}

\author{%
Yoshitsugu Sekine\\{\small\texttt{4429sekine@gmail.com}}%
}

\date{\today}

\begin{document}

\maketitle

\begin{abstract}
We analyze the possibility of Bose-Einstein condensation (BEC) at finite temperature in the spin-boson model within the frameworks of functional integral representations and the resolvent algebra. Because a sesquilinear form arising from the zero mode appears, analogous to the case of the free Bose gas, a BEC-type component is also formally present in the spin-boson model. However, according to \cite{VIYukalov001}, quasiparticles do not undergo BEC, so we have to exclude this possibility. A no-go theorem for BEC is obtained, but we find that the spin-boson model is not a model that follows the criterion in \cite{VIYukalov001}.

\end{abstract}

\setcounter{tocdepth}{3}
\tableofcontents

\section{Introduction}\label{expedition0012078}

The spin-boson model is a basic model describing the linear interaction between a two-level quantum system and a Bose field, and it has an important role in the analysis of non-relativistic quantum field theory, dissipative two-state systems, infrared problems, and quantum phase transitions. For a massless Bose field, infinitely many low-energy quanta can contribute with finite energy, so infrared divergence directly governs the structure of ground states and KMS states. While this point is shared with the Nelson model, the spin-boson model exhibits structures different from a simple linear perturbation of a free field because internal degrees of freedom and spin-flip symmetry enter the analysis of infrared singularities and phase transitions.

For the finite-temperature spin-boson model, Fannes, Nachtergaele, and Verbeure\cite{FannesNachtergaeleVerbeure1} analyzed the one-dimensional Bose field in the framework of \(\oacstar\)-algebras and proved uniqueness of the KMS state at all temperatures and the absence of spontaneous breaking of spin-reflection symmetry. Spohn\cite{HerbertSpohn001} also defined the ground state from the zero-temperature limit of finite-temperature KMS states and, in the Ohmic case, analyzed the critical coupling, the order parameter, and the divergence of the boson number by exploiting the correspondence with the one-dimensional long-range Ising model. This line of work developed into the problem of existence and non-existence of ground states for infrared-singular spin-boson models, leading to the non-perturbative existence result of Hasler--Hinrichs--Siebert\cite{HaslerHinrichsSiebert001} and to Betz--Hinrichs--Kraft--Polzer's analysis of ground-state non-existence at large coupling and the Ising-type phase transition\cite{BetzHinrichsKraftPolzer001}.

These studies mainly focus on ground states, infrared divergence, spin degrees of freedom, and Ising-type correlations. The focus of the present paper is the zero mode of the Bose field in finite-temperature KMS states. In the free Bose gas, the infinite-volume limit together with the chemical-potential limit \(\smchemicalpotential \uparrow 0\) produces a singular sesquilinear form corresponding to the zero mode, which is tied to non-uniqueness of KMS states, the center of the representation algebra, ergodic decomposition, and off-diagonal long-range order. The resolvent algebra\cite{BuchholzGrundling2} is a natural tool for treating this type of singular direction as a \(\oacstar\)-algebraic formulation of the Bose field, and its ideal structure carries information separating physically admissible directions from directions to be removed.

This paper starts from the finite-temperature functional integral representation. The free spin part is represented by a two-state continuous-time Markov process with periodic boundary conditions, namely the spin loop measure. The free Bose field is represented by a \(\sminvtemperature\)-periodic Gaussian field. Since the interaction is a linear functional of the Bose field along a spin path, after Gaussian integration over the Bose field the quadratic part of the characteristic functional retains the same covariance as the free Bose field, while the spin degrees of freedom remain as a weight and a phase factor. This computation shows that the same zero-mode sesquilinear form \(\opform{q}_{\txtbsn,0,\sminvtemperature}\) as in the free Bose gas appears formally also in the spin-boson model.

However, the formal zero-mode factor does not immediately imply physical BEC. In the spin-boson model particle number is not conserved, and the picture in which particle number is controlled by a chemical potential as in the free Bose gas cannot be applied directly. Moreover, because of infrared divergence, the physical test-function space is restricted to the intersection with the domain of some functional \(\mathsf{m}\). As a result, one must distinguish whether a direction that appears to be a zero mode remains as a regular direction on the algebra of physical quantities, or disappears because of the absence of off-diagonal long-range order. We formulate this distinction in terms of field two-point functions, order parameters, and the ideal structure of the resolvent algebra.

The main conclusions are as follows. For the KMS state after removal of the infrared cutoff, the two-point off-diagonal long-range order detects only the zero-mode form \(\opform{q}_{\txtbsn,0,\sminvtemperature}\). If the physical test-function space can detect the zero mode, then absence of off-diagonal long-range order, vanishing of the condensate density \(\smnumberdensity_{\txtbsn,0}(\sminvtemperature)\), the order-parameter criterion, the self-consistent quasiparticle equilibrium condition for the centered field, and the zero-mode disappearance condition are equivalent. Hence in a phase where off-diagonal long-range order vanishes, the formally appearing BEC zero-mode factor becomes trivial in the characteristic functional and does not remain as an independent physical condensation direction. Note that the self-consistent quasiparticle equilibrium-state condition for the centered field corresponds to the Hamiltonian design criterion in \cite{VIYukalov001}; however, we find that the spin-boson model is not a model that follows this criterion.

Finally, we organize this conclusion in terms of the resolvent algebra. Infrared-divergent directions lie outside the physical test-function space and are removed by taking the quotient by the infrared ideal. On this physical quotient algebra, the KMS state is regular and the GNS representation is faithful. On the other hand, BEC directions are auxiliary objects recorded as directions inside the physical test-function space on which \(\opform{q}_{\txtbsn,0,\sminvtemperature}\) is positive. If the zero-mode disappearance condition holds, this set of directions reduces to the empty set and the corresponding BEC ideal also becomes trivial. This is the meaning of the no-go theorem for BEC in the finite-temperature spin-boson model proved in this paper.

\section{Main Results}\label{main-results}

\subsection{Space Setting and Hamiltonian}\label{space-setting-and-hamiltonian}

Let the basic complex Hilbert space and its real subspace be \[\sphilb{H}_{\txtbsn}
=
\fun{\lp^{2}}{\fldreal^{d}},
\quad
\sphilb{H}_{\txtbsn,\txtreal}
=
\fun{\lp_{\txtreal}^{2}}{\fldreal^{d}}
=
\fun{\lp^{2}}{\fldreal^{d};\fldreal}.\] To simplify the discussion, or when discussing BEC, we may restrict the dimension to \(d=3\); in that case this will always be stated explicitly. We set the symplectic form to be \(\sigma(f,g)=\opimag \bkt{f}{g}_{\sphilb{H}_{\txtbsn}}\). As the real symplectic space \((X,\sigma)\), we take \(X\) to be the realification of the complex Hilbert space \(\sphilb{H}_{\txtbsn}\), and define the inner product on this real Hilbert space by \(\opreal \bkt{f}{g}_{\sphilb{H}_{\txtbsn}}\).

For a positive real number \(s>0\), let the one-particle Hamiltonian (dispersion relation) defined on momentum space \(\greuctr{k}{d}\) be \(\omega(k)=\omega_s(k)=\abs{k}^{s}\). For a non-positive chemical potential \(\smchemicalpotential\leq 0\), let the non-negative self-adjoint operator be \[K_{\sminvtemperature,\smchemicalpotential}
=
\coth \frac{\sminvtemperature \rbk{\omega - \smchemicalpotential}}{2}.\] For the associated non-degenerate non-negative symmetric sesquilinear form \(\opform{q}_{\txtbsn,\txtnonzero,\sminvtemperature,\smchemicalpotential}\), define the associated inner-product space and its completion by \[\sphilb{D}_{\txtbsn,\sminvtemperature,\smchemicalpotential}
=
\pairbk{\fun{\opformdomain}{\opform{q}_{\txtbsn,\txtnonzero,\sminvtemperature,\smchemicalpotential}},\opform{q}_{\txtbsn,\txtnonzero,\sminvtemperature,\smchemicalpotential}},
\quad
\sphilb{H}_{\txtbsn,\sminvtemperature,\smchemicalpotential}
=
\gtclos{\sphilb{D}_{\txtbsn,\sminvtemperature,\smchemicalpotential}}^{\opform{q}_{\txtbsn,\txtnonzero,\sminvtemperature,\smchemicalpotential}}.\]

In general, define the bosonic Fock space over the Hilbert space \(\sphilb{H}_{\txtbsn}\) by \[\fun{\spfock_{\txtbsn}}{\sphilb{H}_{\txtbsn}}
=
\bigoplus_{n=0}^{\infty}
\bigotimes_{\txtsym}^{n}
\sphilb{H}_{\txtbsn},\] and for any \(f\in \sphilb{H}_{\txtbsn}\) let \(\opfockcran_{\txtfock}(f)\) denote the creation and annihilation operators on the bosonic Fock space. The Segal field operator is defined by \[\opfocksegal_{\txtfock}(f)
=
\frac{1}{\sqrt{2}}
\rbk{\opfockcr_{\txtfock}(f) + \opfockan_{\txtfock}(f)},\] and the Weyl operator on Fock space is defined by \(\opfockweyl_{\txtfock}(f)=\napiernum^{\imunit \opfocksegal_{\txtfock}(f)}\). Next, for a non-expansive operator \(V\), define the second quantized operator of the second kind \(\fun{\opfocksndqnt_{\txtbsn}}{V}\) by \[\fun{\opfocksndqnt_{\txtbsn}}{V}
=
\bigoplus_{n=0}^{\infty} \bigotimes_{\txtsym}^{n} V,\] and define the second quantized operator of the first kind \(\physham_{\txtbsn,\txtfr}=\fun{\opfocksndqntdiff_{\txtbsn}}{\omega}\) by \[\physham_{\txtbsn,\txtfr}
=
-\imunit
\fnrestr{\opod{t} \opfocksndqnt (\napiernum^{\imunit t \omega})}{t=0}.\] The ground state of the second quantized operator of the first kind is the Fock vacuum \(\opfockvac_{\txtbsn}\).

The Hilbert space of the total system for the spin-boson model is \[\spfock_{\txtspinboson}
=
\fldcmp^{2} \otimes \spfock_{\txtbsn}, \quad
\spfock_{\txtbsn}
=
\spfock_{\txtbsn}(\sphilb{H}_{\txtbsn})
=
\bigoplus_{n = 0}^{\infty} \bigotimes_{\txtsym}^{n} \sphilb{H}_{\txtbsn},
\quad
\sphilb{H}_{\txtbsn}
=
\fun{\lp^{2}}{\fldreal^{d}},\] where \(\bigoplus\) denotes the direct sum, \(\bigotimes_{\txtsym}^{n}\) denotes the symmetric tensor product, and the zeroth tensor power is \(\bigotimes_{\txtsym}^{0} \sphilb{H}_{\txtbsn}= \fldcmp\).

Following the textbook\cite{LorincziHiroshimaBetz3}, we define the Hamiltonian as follows. Let the coupling constant be \(\physcplconst\geq 0\) and let \((\qmpaulispin_{i})\) denote the Pauli matrices. The Hamiltonian of the non-interacting system (free Hamiltonian) is \[\physham_{0}
=
-\physcplconst \qmpaulispin_{x} \otimes 1
+1 \otimes \physham_{\txtbsn,\txtfr}.\] For an infrared-cutoff source \(\varrho_{\kappa}\in \fun{\dsttempered_{\txtreal}}{\fldreal^{d}}\), set \(\mathsf{m}_{\kappa}= \omega^{-\frac{3}{2}} \varrho_{\kappa}\) and define the spin-boson Hamiltonian by \[\physham_{\txtspinboson,\kappa}
=
\physham_{0}
+\qmpaulispin_{z} \otimes \fun{\phi_{\txtfock}}{\omega \mathsf{m}_{\kappa}}.\] By the Kato--Rellich theorem, this is a self-adjoint operator bounded from below. In the case \(\kappa=0\) where the cutoff is removed, we simply omit \(\kappa\) from the notation. Since the spin-boson model is a low-energy effective theory, the condition corresponding to an ultraviolet cutoff is introduced here implicitly as a constraint on the source function.

For the source, we further impose non-negativity as a tempered distribution and \(\varrho\in \dom \omega^{-\onehalf} \cap \dom \inv{\omega}\), and assume the infrared singularity condition \(\varrho\notin \dom \omega^{-\frac{3}{2}}\). We impose \(\varrho_{\kappa}\in \dom{\omega^{-\frac{1}{2}}}\) and \(\varrho_{\kappa}\in \dom{\omega^{-1}}\) for every \(\kappa\geq 0\), and also impose convergence \(\varrho_{\kappa}\to \varrho\) as tempered distributions as \(\kappa\to 0\). In particular, the infrared cutoff is defined by the cutoff with the defining function \(\varrho_{\kappa}= \varrho \fndef{\kappa \leq \abs{k}}\). We further assume that the Fourier transform of \(\varrho\) is an ordinary function, regard \(\mathsf{m}\) as an ordinary function on the Fourier-transform side as well, and denote by the same symbols the functionals induced by these objects. In particular, the domain of the functional \(\mathsf{m}\) on \(\fun{\lp^{2}}{\fldreal^{d}}\) is denoted by \(\dom \mathsf{m}\). In general, the cutoff domain \(\dom \mathsf{m}_{\kappa}\) is the whole space, while the cutoff-free domain \(\dom \mathsf{m}\) is a subspace of \(\fun{\lp^{2}}{\fldreal^{d}}\).

\subsection{Weyl Algebra}\label{weyl-algebra}

Let \(\sphilb{H}\) be a complex Hilbert space, and define the bilinear map \(\sigma\) as the symplectic form by \[\sigma
\colon \sphilb{H} \times \sphilb{H}
\to
\fldreal;
\quad
\sigma(f,g)
=
\opimag \bkt{f}{g}_{\sphilb{H}}.\] For any \(f,g\in \sphilb{H}\), define the \(\oacstar\)-algebra \[\oaweyl
=
\oaweyl(\sphilb{H}, \sigma)
=
\oacstar
\set{\opfockweyl(f)}{f \in \sphilb{H}}\] generated by \(\opfockweyl(f)\) satisfying the Weyl relations \begin{equation}
\begin{aligned}
\faadj{\opfockweyl(f)}
&=
\opfockweyl(-f), \\
\opfockweyl(f) \opfockweyl(g)
&=
\napiernum^{-\frac{\imunit}{2} \opimag \bkt{f}{g}_{\sphilb{H}}}
\opfockweyl(f+g)
\end{aligned}
\end{equation} We also use the Weyl algebra \(\oaweyl(\sphilb{D})\) for a subspace \(\sphilb{D}\subset \sphilb{H}\). If a representation \(\pairbk{\sphilb{H}_{\repn},\repn}\) realizes \(t\mapsto \repn(\opfockweyl(tf))\) as strongly continuous for every \(f\in\sphilb{D}\), then it is called a regular representation on \(\sphilb{D}\). A state is also called regular if its GNS representation is regular.

\subsection{Resolvent Algebra}\label{expedition0012083}

Following \cite{DetlevBuchholz001}, we recall the definition and basic properties of the resolvent algebra. Let \((X,\sigma)\) be a symplectic space. Let \(\oaresolventalgebra_0\) be the universal unital \(\ast\)-algebra generated by the set \(\set{\oaresolvent(\lambda,f)}{\lambda \in \fldmultiplicativegroup{\fldreal}, f \in \sphilb{H}}\), and assume in particular the following resolvent relations: \begin{align}
\oaresolvent(\lambda,0)
&=
-\frac{\imunit}{\lambda} \idone, \\ % (1)
\faadj{\oaresolvent(\lambda,f)}
&=
\oaresolvent(-\lambda,f), \\ % (2)
\nu \oaresolvent(\nu \lambda, \nu f)
&=
\oaresolvent(\lambda, f), \\ % (3)
\oaresolvent(\lambda,f) - \oaresolvent(\mu,f)
&=
\imunit
(\mu - \lambda)
\oaresolvent(\lambda,f) \cdot \oaresolvent(\mu,f) \\ % (4)
&=
\imunit
(\mu - \lambda)
\oaresolvent(\mu,f) \cdot \oaresolvent(\lambda,f), \\ % (4)
\commutator{\oaresolvent(\lambda,f)}{\oaresolvent(\mu,g)}
&=
\imunit
\sigma(f,g)
\oaresolvent(\lambda,f)
\oaresolvent(\mu,g)^2
\oaresolvent(\lambda,f), \label{expedition0012052} \\ % (5)
\oaresolvent(\lambda,f)
\oaresolvent(\mu,g)
&=
\oaresolvent(\lambda+\mu, f+g)
\cdot
\rbkleft{\oaresolvent(\lambda,f)} \\
&\quad\rbkright{+
\oaresolvent(\mu,g)
+\imunit \sigma(f,g) \oaresolvent(\lambda,f)^2 \oaresolvent(\mu,g)} % (6)
\end{align} In particular, condition \eqref{expedition0012052} implies that \(\oaresolvent(\lambda,f)\) and \(\oaresolvent(\mu,f)\) commute when the entries \(f\) are equal.

The \(\ast\)-algebra obtained by introducing an appropriate norm on \(\oaresolventalgebra_0\) and completing it is called the abstract resolvent algebra, or simply the resolvent algebra. In particular, \cite[P.2730, Theorem 3.6 (iii)]{BuchholzGrundling2} gives \(\norm{\oaresolvent(\lambda,f)}= \frac{1}{\abs{\lambda}}\). Moreover, if a representation \(\oarepn\) satisfies \(\Ker \oarepn(\oaresolvent(1,f))=\setone{0}\) for every \(f\in S\), then it is called a regular representation on \(S\). A state is also called regular if its GNS representation is regular.

\begin{prop}[\cite{BuchholzGrundling2}]\label{expedition0011838}
Let $\pairbk{X,\sigma}$ be a symplectic space of arbitrary dimension, and let $S\subset X$ be a non-degenerate finite-dimensional subspace.
\begin{enumerate}
\item
The norms of the full resolvent algebra $\oaresolventalgebra(X,\sigma)$ and the subalgebra $\oaresolventalgebra(X,\sigma)$ coincide on the $\ast$-subalgebra
$$\oastaralgebra
\set{\oaresolvent(\lambda,f)}
{f \in S, \lambda \in \fldreal \setminus \setone{0}}.$$ In particular, $\oaresolventalgebra(S,\sigma)\subset \oaresolventalgebra(X,\sigma)$ holds.

\item
The full resolvent algebra is the inductive limit of the net $\fml{\oaresolventalgebra(S,\sigma)}{X \subset S}$ over non-degenerate finite-dimensional subspaces $S\subset X$.

\item
Every regular representation of the full resolvent algebra $\oaresolventalgebra(X,\sigma)$ is faithful.
\end{enumerate}
In particular, the center of the full resolvent algebra is trivial.
\end{prop}

\subsection{Finite-Temperature Setting}\label{expedition0011278}

Here again we use the setting adopted in \cite{YoshitsuguSekine004}. For inverse temperature \(\sminvtemperature>0\), define the time interval corresponding to the periodicity of the KMS state by \[S_{\sminvtemperature}
=
\closedinterval{-\frac{\sminvtemperature}{2}}{\frac{\sminvtemperature}{2}}.\]

Let \(\smnumberdensity_{\txtbsn,0}(\sminvtemperature)\) denote the density of the condensate at inverse temperature \(\sminvtemperature>0\). Define the non-closed non-negative symmetric bilinear form corresponding to the condensed component by \[\opform{q}_{\txtbsn,0,\sminvtemperature}(f)
=
2 (2 \pi)^d \smnumberdensity_{\txtbsn,0}(\sminvtemperature)
\abs{\faftr{f}(0)}^{2},
\quad
\opformdomain(\opform{q}_{\txtbsn,0,\sminvtemperature})
=
\fun{\lp^{1}}{\fldreal^{d}}
\cap
\fun{\lp^{2}}{\fldreal^{d}}.\] Define the subspace \(\sphilb{D}_{\txtbsn,0,\sminvtemperature,\smchemicalpotential}\) by \[\sphilb{D}_{\txtbsn,0,\sminvtemperature,\smchemicalpotential}
=
\opformdomain(\opform{q}_{\txtbsn,0,\sminvtemperature})
\cap
\sphilb{H}_{\txtbsn,\sminvtemperature,\smchemicalpotential}.\] When the chemical potential is \(0\), set \(\sphilb{D}_{\txtbsn,0,\sminvtemperature}=\sphilb{D}_{\txtbsn,0,\sminvtemperature,0}\), and define the physical one-particle space by \[\sphilb{D}_{\txtbsn,\txtphys,\sminvtemperature}
=
\dom \mathsf{m}
\cap
\sphilb{D}_{\txtbsn,0,\sminvtemperature}.\] Furthermore, for arbitrary \(f,g\in \sphilb{D}_{\txtbsn,0,\sminvtemperature}\), define the sesquilinear form including the BEC component by \[\fun{\opform{q}_{\txtbsn,\txtbec,\sminvtemperature}}{f,g}
=
\fun{\opform{q}_{\txtbsn,0,\sminvtemperature}}{f,g}
+\fun{\opform{q}_{\txtbsn,\txtnonzero,\sminvtemperature}}{f,g}.\] In one-variable notation, \[\fun{\opform{q}_{\txtbsn,\txtbec,\sminvtemperature}}{f}
=
\fun{\opform{q}_{\txtbsn,0,\sminvtemperature}}{f}
+
\fun{\opform{q}_{\txtbsn,\txtnonzero,\sminvtemperature}}{f}.\] In this paper we mainly deal with \(\oaweyl(\sphilb{D}_{\txtbsn,\txtphys,\sminvtemperature})\), \(\oaresolventalgebra(\sphilb{D}_{\txtbsn,\txtphys,\sminvtemperature},\sigma)\), and the algebra of all physical quantities obtained by tensoring it with the spin matrix algebra \(M_2\).

\subsection{Main Results}\label{main-results-1}

Let \(\oastate[\psi_{\txtspinboson,\sminvtemperature}]\) denote the KMS state of the infinite system after removal of the infrared cutoff.

\begin{thm}[Zero-mode decomposition of KMS states]
Assume the thermodynamic limit in which the BEC zero mode appears for the free Bose gas. Then, for arbitrary $f\in\sphilb{D}_{\txtbsn,\txtphys,\sminvtemperature}$ and $s\in\fldreal$, the centered physical field $\opfocksegal_{\txtphys,\sminvtemperature}(f)$ satisfies
$$\begin{aligned}
&\fun{\psi_{\txtspinboson,\sminvtemperature}}
{\napiernum^{\imunit s\opfocksegal_{\txtphys,\sminvtemperature}(f)}}
\\
&=
\fnexp{-\frac{s^2}{4}\fun{\opform{q}_{\txtbsn,0,\sminvtemperature}}{f}}
\fnexp{-\frac{s^2}{4}\fun{\opform{q}_{\txtbsn,\txtnonzero,\sminvtemperature}}{f}}
\sqfun{\prbexp_{\msrprb_{\txtspin,\txtspinboson,\sminvtemperature}}}
{\fnexp{-\imunit s\widetilde{\mathsf{Y}}_{\sminvtemperature,f}}}
\end{aligned}$$
Here $\widetilde{\mathsf{Y}}_{\sminvtemperature,f}$ is a centered real random variable determined by the spin path. In particular, the quadratic form corresponding to BEC is the form $\opform{q}_{\txtbsn,0,\sminvtemperature}$ inherited from the free Bose gas, and the spin-boson interaction does not generate a new zero-mode form.
\end{thm}

\begin{thm}[BEC no-go theorem via off-diagonal long-range order]
Assume that the physical test-function space $\sphilb{D}_{\txtbsn,\txtphys,\sminvtemperature}$ can detect the zero mode in the sense of Definition \ref{expedition0012134}. Then, for arbitrary $f,g\in \sphilb{D}_{\txtbsn,\txtphys,\sminvtemperature}$,
$$\lim_{\abs{x}\to\infty}
\fun{\oastate[\psi_{\txtspinboson,\sminvtemperature}]}
{\opfocksegal(j_0f)\opfocksegal(j_0\tau_xg)}
=
\onehalf\fun{\opform{q}_{\txtbsn,0,\sminvtemperature}}{f,g}$$
holds. Moreover, the following conditions are equivalent.
\begin{enumerate}
\item
For all $f\in \sphilb{D}_{\txtbsn,\txtphys,\sminvtemperature}$, $\fun{\opform{q}_{\txtbsn,0,\sminvtemperature}}{f}=0$ holds.

\item
For all $f,g\in \sphilb{D}_{\txtbsn,\txtphys,\sminvtemperature}$, the two-point off-diagonal long-range order is $0$.

\item
The condition $\smnumberdensity_{\txtbsn,0}(\sminvtemperature)=0$ holds.

\item
The centered physical field defines a self-consistent quasiparticle equilibrium state.

\item
The order-parameter criterion $\lim_{L\to\infty}\mathsf{o}_{\txtspinboson,\sminvtemperature,L}^{(1)}=1$ holds.
\end{enumerate}
Thus, the absence of off-diagonal long-range order is equivalent to the disappearance of the BEC zero mode on the physical test-function space.
\end{thm}

\begin{thm}[Infrared quotient and BEC ideal in the resolvent algebra]
Describe the physical quantities of the Bose field by the resolvent algebra. For the smallest closed two-sided ideal $\oaideal{J}_{\txtspinboson,\txtirsingular}$ that kills the infrared directions,
$$\setquot{\oa{A}_{\txtspinboson,0,\sminvtemperature}}
{\oaideal{J}_{\txtspinboson,\txtirsingular}}
\eqalgisom
\oa{A}_{\txtspinboson,\txtphys,\sminvtemperature}$$
holds. On the physical quotient algebra, the KMS state is regular and its GNS representation is faithful.

If the zero-mode disappearance condition holds, then the set of BEC directions inside the physical test-function space is empty, and the corresponding BEC ideals satisfy $\oaideal{J}_{\txtbsn,\txtbec}= \setone{0}$ and $\oaideal{J}_{\txtspinboson,\txtbec}= \setone{0}$.
\end{thm}

\section{Functional Integral Representation of the Spin Component at Finite Temperature}\label{functional-integral-representation-of-the-spin-component-at-finite-temperature}

To construct the KMS state of the free Hamiltonian and then construct the KMS state of the full system as its perturbation, we first discuss the KMS state of the free Hamiltonian. Since the free Hamiltonian is the sum of the spin and the free field, it suffices to consider the functional integral representation of each component; here we discuss the spin component. The discussion is modeled on the zero-temperature version in the textbook \cite{LorincziHiroshimaBetz3}.

\begin{defn}
Using the eigenvalues of the $z$ component $\qmpaulispin_z$ of the Pauli matrices, define the spin state space by $\ringratint_{2}
=
\setone{\pm 1}$. In particular, $\qmpaulispin
\in \ringratint_{2}$ may be regarded as an eigenvalue of $\qmpaulispin_z$. Let the spin Hamiltonian be $\physham_{\txtspin}
= - \physcplconst \qmpaulispin_{x}$, and identify the spin space $\sphilb{H}_{\txtspin}$ with
$$\sphilb{H}_{\txtspin}
=
\fldcmp^{2}
=
\fun{\lp^{2}}{\ringratint_{2}}$$
equipped with counting measure. Under this identification, the Pauli matrices act as
$$\funrbk{\qmpaulispin_{z} f}{\qmpaulispin}
= \qmpaulispin f(\qmpaulispin),
\quad
\funrbk{\qmpaulispin_{x} f}{\qmpaulispin}
=
f(-\qmpaulispin).$$
\end{defn}

\begin{defn}
On the probability space $\pairbk{\mathcal{S}_{\txtspin},
\mblfmlborel_{\txtspin},
\msrprb_{\txtpoisson}}$, let $\fml{\prbpoissonprocess_t}{t \geq 0}$ be a Poisson process of intensity $\physcplconst$. Thus $\prbpoissonprocess_0
= 0$, and for arbitrary $t
\geq 0$ and $n
\in \ringratint_{\txtnonneg}$,
$$\msrprb_{\txtpoisson}(\prbpoissonprocess_{t} = n)
=
\napiernum^{-\physcplconst t}
\frac{(\physcplconst t)^n}{n!}$$
holds; the process has independent and stationary increments. In particular, $t
\mapsto
\prbpoissonprocess_t$ is a càdlàg jump process, and the spin flips at each jump. For $\qmpaulispin
\in \ringratint_{2}$, define the spin process by
$$\prbprocess_{\qmpaulispin}^{\txtspin}
=
\fml{\prbprocess_{\qmpaulispin,t}^{\txtspin}}{t \geq 0}
=
\fml{\qmpaulispin (-1)^{\prbpoissonprocess_{t}}}
{t \geq 0}.$$ When the initial spin $\qmpaulispin$ is not displayed explicitly, write simply $\prbprocess^{\txtspin}
=
\fml{\prbprocess_{t}^{\txtspin}}{t \geq 0}$.
\end{defn}

\begin{prop}
The spin process $\fml{\prbprocess_{\qmpaulispin,t}^{\txtspin}}
{t \geq 0}$ is a Markov process. In particular, its generator $L_{\physcplconst}$ is
$$(L_{\physcplconst} f)(\qmpaulispin)
=
\physcplconst \rbk{f(-\qmpaulispin) - f(\qmpaulispin)}
=
\physcplconst (\qmpaulispin_{x} - 1)
f(\qmpaulispin),$$
the Markov semigroup is $\funrbk{P_{\physcplconst,t} f}{\qmpaulispin}
=
\sqfun{\prbexp_{\msrprb_{\txtpoisson}}}
{\fun{f}{\prbprocess_{\qmpaulispin,t}^{\txtspin}}}$, and the transition probability is
$$p_t(\qmpaulispin_{1},\qmpaulispin_{2})
=
\fun{\msrprb_{\txtpoisson}}
{\prbprocess_{\qmpaulispin_{1},t}^{\txtspin} = \qmpaulispin_2}
=
\onehalf
\rbk{1 + \qmpaulispin_{1} \qmpaulispin_{2} \napiernum^{-2 \physcplconst t}},
\quad
\begin{dcases}
p_t(\qmpaulispin,\qmpaulispin)
&=
\onehalf \rbk{1 + \napiernum^{-2 \physcplconst t}}, \\
p_t(\qmpaulispin,-\qmpaulispin)
&=
\onehalf \rbk{1 - \napiernum^{-2 \physcplconst t}}
\end{dcases}.$$
Thus this is a two-state continuous-time Markov chain.
\end{prop}

\begin{proof}
At each jump the sign is reversed, and the increments of the Poisson process are independent and stationary; hence $\fml{\prbprocess_{\qmpaulispin,t}^{\txtspin}}
{t \geq 0}$ is a time-homogeneous Markov process. Indeed, for any bounded function $f$,
$$\funrbk{P_{\physcplconst,t} f}{\qmpaulispin}
=
\sqfun{\prbexp_{\msrprb_{\txtpoisson}}}
{\fun{f}{\qmpaulispin (-1)^{\prbpoissonprocess_t}}}
=
\sum_{n=0}^{\infty}
f(\qmpaulispin (-1)^n)
\napiernum^{-\physcplconst t}
\frac{(\physcplconst t)^n}{n!}$$
holds. Therefore, by the probabilities of an even and an odd number of jumps,
$$\sum_{m=0}^{\infty}
\napiernum^{-\physcplconst t}
\frac{(\physcplconst t)^{2m}}{(2m)!}
=
\onehalf \rbk{1 + \napiernum^{-2 \physcplconst t}},
\quad
\sum_{m=0}^{\infty}
\napiernum^{-\physcplconst t}
\frac{(\physcplconst t)^{2m+1}}{(2m+1)!}
=
\onehalf \rbk{1 - \napiernum^{-2 \physcplconst t}},$$
we obtain the displayed formula for the transition probability $p_t(\qmpaulispin_1,\qmpaulispin_2)$. Differentiating at $t \downarrow 0$ gives
$$\fnrestr{\frac{1}{t}\rbk{P_{\physcplconst,t}f - f}}{t = 0}
=
\physcplconst \rbk{f(-\qmpaulispin) - f(\qmpaulispin)}
=
\physcplconst (\qmpaulispin_{x} - 1)f(\qmpaulispin),$$
so the generator is $L_{\physcplconst}$.
\end{proof}

\begin{defn}
Let $D(S_{\sminvtemperature};
\ringratint_{2})$ be the space of càdlàg paths on the periodic time interval $S_{\sminvtemperature}$ with values in $\ringratint_{2}$. Let $\msrlaw_{\txtspin,\sminvtemperature,\qmpaulispin}$ be the law of the shifted process $\fml{\prbprocess_{\qmpaulispin,t+\frac{\sminvtemperature}{2}}^{\txtspin}}
{t \in S_{\sminvtemperature}}$ starting from $\qmpaulispin$ at time $-\frac{\sminvtemperature}{2}$, and define the spin loop measure at finite temperature by
$$\fun{\nu_{\txtspin,\sminvtemperature}}{A}
=
\napiernum^{\physcplconst \sminvtemperature}
\sum_{\qmpaulispin \in \ringratint_{2}}
\fun{\msrlaw_{\txtspin,\sminvtemperature,\qmpaulispin}}
{A \cap \setone{\prbprocess_{\qmpaulispin,\sminvtemperature}^{\txtspin} = \qmpaulispin}}.$$
\end{defn}

The spin loop measure is not a probability measure but a finite positive measure. The event \(\setone{\prbprocess_{\qmpaulispin,\sminvtemperature}^{\txtspin} = \qmpaulispin}\) imposes the periodic condition that identifies the two endpoints of the interval, and corresponds to the finite-temperature trace.

For notation for time-ordered jump times, write \(\nfoldvar{t}{n}
=
(t_1,\cdots,t_n)
\in I_{\sminvtemperature,<}^{n}\) for \(n\) times satisfying \[-\frac{\sminvtemperature}{2}
< t_1
< \cdots
< t_n
< \frac{\sminvtemperature}{2}.\] Let \(I_{\sminvtemperature,\leq}^{n}\) denote the corresponding set with \(<\) replaced by \(\leq\), and define the path \(\prbprocess_{(\qmpaulispin,\nfoldvar{t}{n})}^{\txtspin}\) with initial value \(\qmpaulispin\) by \[\prbprocess_{(\qmpaulispin,\nfoldvar{t}{n}),t}^{\txtspin}
=
\qmpaulispin (-1)^{\abscard{\set{j}{t_j \leq t}}}.\]

\begin{prop}
For any bounded measurable function $F$ on $D(S_{\sminvtemperature};\ringratint_{2})$,
$$\int_{D(S_{\sminvtemperature};\ringratint_{2})}
\fun{F}{\prbprocess^{\txtspin}}
\opdmsr{\fun{\nu_{\txtspin,\sminvtemperature}}{\prbprocess^{\txtspin}}}
=
\sum_{\qmpaulispin \in \ringratint_{2}}
\sum_{m=0}^{\infty}
\physcplconst^{2m}
\int_{I_{\sminvtemperature,<}^{2m}}
\fun{F}{\prbprocess_{(\qmpaulispin,\nfoldvar{t}{2m})}^{\txtspin}}
\opdmsr{t_{1}}
\cdots
\opdmsr{t_{2m}}$$
holds, and for the total mass one obtains
$$\fun{\msr{\nu_{\txtspin,\sminvtemperature}}}
{D(S_{\sminvtemperature};\ringratint_{2})}
=
\sqfun{\trace_{\fun{\lp^{2}}{\ringratint_{2}}}}{\napiernum^{-\sminvtemperature \physham_{\txtspin}}}
=
2 \cosh (\physcplconst \sminvtemperature).$$
In particular, the normalized thermal spin-path probability measure is
$$\msrprb_{\txtspin,\sminvtemperature}
=
\frac{1}{2 \cosh (\physcplconst \sminvtemperature)}
\msr{\nu_{\txtspin,\sminvtemperature}}.$$
\end{prop}

\begin{proof}
We express the spin loop measure in terms of jump times. The periodicity condition $\prbprocess_{(\qmpaulispin,\nfoldvar{t}{n}),-\frac{\sminvtemperature}{2}}^{\txtspin}
= \prbprocess_{(\qmpaulispin,\nfoldvar{t}{n}),\frac{\sminvtemperature}{2}}^{\txtspin}$ is equivalent to $n$ being even. Thus the number of jumps must be even. For a Poisson process of intensity $\physcplconst$, the density of $n$ jump times is $\napiernum^{- \physcplconst \sminvtemperature}
\physcplconst^{n}
\opdmsr{t_1}
\cdots
\opdmsr{t_n}$, while the measure $\nu_{\txtspin,\sminvtemperature}$ multiplies the whole expression by $\napiernum^{\physcplconst \sminvtemperature}$, so the exponential factor of the Poisson process cancels. Hence for any bounded measurable function $F$,
$$\int_{D(S_{\sminvtemperature};\ringratint_{2})}
\fun{F}{\prbprocess^{\txtspin}}
\opdmsr{\fun{\nu_{\txtspin,\sminvtemperature}}{\prbprocess^{\txtspin}}}
=
\sum_{\qmpaulispin \in \ringratint_{2}}
\sum_{m=0}^{\infty}
\physcplconst^{2m}
\int_{I_{\sminvtemperature,<}^{2m}}
\fun{F}{\prbprocess_{(\qmpaulispin,\nfoldvar{t}{2m})}^{\txtspin}}
\opdmsr{t_{1}}
\cdots
\opdmsr{t_{2m}}$$
holds.

We compute the total mass of the measure $\msr{\nu_{\txtspin,\sminvtemperature}(\prbprocess^{\txtspin})}$. By definition,
$$\fun{\msr{\nu_{\txtspin,\sminvtemperature}}}
{D(S_{\sminvtemperature};\ringratint_{2})}
=
\napiernum^{\physcplconst \sminvtemperature}
\sum_{\qmpaulispin \in \ringratint_{2}}
\fun{\msrlaw_{\txtspin,\sminvtemperature,\qmpaulispin}}
{\prbprocess_{\qmpaulispin,\sminvtemperature}^{\txtspin} = \qmpaulispin}.$$
Since the time difference is $\sminvtemperature$,
$$\fun{\msrlaw_{\txtspin,\sminvtemperature,\qmpaulispin}}
{\prbprocess_{\qmpaulispin,\sminvtemperature}^{\txtspin} = \qmpaulispin}
=
\sum_{m=0}^{\infty}
\napiernum^{-\physcplconst \sminvtemperature}
\frac{(\physcplconst \sminvtemperature)^{2m}}{(2m!)}
=
\napiernum^{-\physcplconst \sminvtemperature}
\cosh (\physcplconst \sminvtemperature),$$
and therefore
$$\fun{\msr{\nu_{\txtspin,\sminvtemperature}}}
{D(S_{\sminvtemperature};\ringratint_{2})}
=
\napiernum^{\physcplconst \sminvtemperature}
\cdot
2
\napiernum^{-\physcplconst \sminvtemperature}
\cosh (\physcplconst \sminvtemperature)
=
2 \cosh (\physcplconst \sminvtemperature).$$
On the other hand,
$$\sqfun{\trace_{\fun{\lp^{2}}{\ringratint_{2}}}}{\napiernum^{\sminvtemperature \physcplconst \qmpaulispin_{x}}}
=
\napiernum^{\physcplconst \sminvtemperature}
+\napiernum^{-\physcplconst \sminvtemperature}
=
2 \cosh (\physcplconst \sminvtemperature),$$
and hence
$$\fun{\msr{\nu_{\txtspin,\sminvtemperature}}}
{D(S_{\sminvtemperature};\ringratint_{2})}
=
\sqfun{\trace_{\fun{\lp^{2}}{\ringratint_{2}}}}{\napiernum^{-\sminvtemperature \physham_{\txtspin}}}$$
also holds.
\end{proof}

\begin{prop}
For a time-ordered sequence $\nfoldvar{s}{n}
\in I_{\sminvtemperature,<}^{n}$, let $f_{1},\cdots,f_n
\colon \ringratint_{2}
\to \fldcmp$ be bounded functions, and denote the corresponding multiplication operators by the same symbols $f_j$. Then
\begin{equation}
\begin{aligned}
&\int_{D(S_{\sminvtemperature};\ringratint_{2})}
\prod_{j=1}^{n}
\fun{f_j}{\prbprocess_{s_j}^{\txtspin}}
\opdmsr{\msrprb_{\txtspin,\sminvtemperature}}
\\ %%%%%%%%%%%%%%%%
&=
\frac{1}{2 \cosh (\physcplconst \sminvtemperature)}
\sqfun{\trace_{\fun{\lp^{2}}{\ringratint_{2}}}}
{\napiernum^{\physcplconst \qmpaulispin_{x} \rbk{s_1 + \frac{\sminvtemperature}{2}}}
f_1
\napiernum^{\physcplconst \qmpaulispin_{x} \rbk{s_2-s_1}}
f_2
\cdots
\napiernum^{\physcplconst \qmpaulispin_{x} \rbk{s_n-s_{n-1}}}
f_{n}
\napiernum^{\physcplconst \qmpaulispin_{x} \rbk{\frac{\sminvtemperature}{2} - s_{n}}}}.
\end{aligned}
\end{equation}
holds. In terms of the transition probability $p_t$, this can be written as
\begin{equation}
\begin{aligned}
&\int_{D(S_{\sminvtemperature};\ringratint_{2})}
\prod_{j=1}^{n}
\fun{f_j}{\prbprocess_{s_j}^{\txtspin}}
\opdmsr{\msrprb_{\txtspin,\sminvtemperature}}
\\ %%%%%%%%%%%%%%%%
&=
\frac{\napiernum^{\physcplconst \sminvtemperature}}{2 \cosh (\physcplconst \sminvtemperature)}
\\ %%%%%%%%%%%%%%%%
&\quad\times
\sum_{\qmpaulispin_{0},\cdots,\qmpaulispin_{n} \in \ringratint_{2}}
p_{s_1 + \frac{\sminvtemperature}{2}}(\qmpaulispin_{0},\qmpaulispin_{1})
f_1(\qmpaulispin_{1})
p_{s_{2} - s_{1}}(\qmpaulispin_{1},\qmpaulispin_{2})
\cdots
p_{s_{n}-s_{n-1}}(\qmpaulispin_{n-1},\qmpaulispin_{n})
f_{n}(\qmpaulispin_{n})
p_{\frac{\sminvtemperature}{2} - s_{n}}(\qmpaulispin_{n},\qmpaulispin_{0}).
\end{aligned}
\end{equation}
\end{prop}

\begin{proof}
By the definition of the thermal spin-path probability measure $\msrprb_{\txtspin,\sminvtemperature}
=
\frac{1}{2 \cosh (\physcplconst \sminvtemperature)}
\msr{\nu_{\txtspin,\sminvtemperature}}$, the expectation is
$$\frac{\napiernum^{\physcplconst \sminvtemperature}}{2 \cosh (\physcplconst \sminvtemperature)}
\sum_{\qmpaulispin_0 \in \ringratint_2}
\int_{D(S_{\sminvtemperature};\ringratint_2)}
\prod_{j=1}^n
\fun{f_j}{\prbprocess_{s_j}^{\txtspin}}
\cdot
\fun{\fndef{\setone{\prbprocess_{\qmpaulispin_0,\sminvtemperature}^{\txtspin}=\qmpaulispin_0}}}
{\prbprocess_{s_j}^{\txtspin}}
\opdmsr{\fun{\msrlaw_{\txtspin,\sminvtemperature,\qmpaulispin_0}}
{\prbprocess_{s_j}^{\txtspin}}}.$$
Applying the Markov property, according to the definition of the time-ordered sequence, on each interval $\rbk{-\frac{\sminvtemperature}{2},s_1}, \rbk{s_1,s_2}, \cdots,
\rbk{s_n,\frac{\sminvtemperature}{2}}$, we obtain
\begin{equation}
\begin{aligned}
&\int_{D(S_{\sminvtemperature};\ringratint_2)}
\prod_{j=1}^n
\fun{f_j}{\prbprocess_{s_j}^{\txtspin}}
\cdot
\fun{\fndef{\setone{\prbprocess_{\qmpaulispin_0,\sminvtemperature}^{\txtspin}=\qmpaulispin_0}}}
{{\prbprocess_{s_j}^{\txtspin}}}
\opdmsr{\fun{\msrlaw_{\txtspin,\sminvtemperature,\qmpaulispin_0}}
{\prbprocess_{s_j}^{\txtspin}}}
\\ %%%%%%%%%%%%%%%%
&=
\sum_{\qmpaulispin_1,\cdots,\qmpaulispin_n \in \ringratint_2}
p_{s_1+\frac{\sminvtemperature}{2}}(\qmpaulispin_0,\qmpaulispin_1)
f_1(\qmpaulispin_1)
p_{s_2-s_1}(\qmpaulispin_1,\qmpaulispin_2)
\cdots
p_{s_n-s_{n-1}}(\qmpaulispin_{n-1},\qmpaulispin_n)
f_n(\qmpaulispin_n)
p_{\frac{\sminvtemperature}{2}-s_n}(\qmpaulispin_n,\qmpaulispin_0).
\end{aligned}
\end{equation}
Summing further over $\qmpaulispin_0$ gives the representation by the probability kernel.

Next, regarding $f_j$ as multiplication operators on $\fun{\lp^2}{\ringratint_2}$, the identity
$$\funrbk{\napiernum^{t\physcplconst\qmpaulispin_x}g}{\qmpaulispin}
=
\sum_{\eta\in\ringratint_2} p_t(\qmpaulispin,\eta) g(\eta)$$
rewrites the preceding probability-kernel formula directly as
$$\sqfun{\trace_{\fun{\lp^{2}}{\ringratint_{2}}}}
{\napiernum^{\physcplconst \qmpaulispin_x \rbk{s_1+\frac{\sminvtemperature}{2}}}
f_1
\napiernum^{\physcplconst \qmpaulispin_x \rbk{s_2-s_1}}
f_2 \cdots
\napiernum^{\physcplconst \qmpaulispin_x \rbk{s_n-s_{n-1}}}
f_n
\napiernum^{\physcplconst \qmpaulispin_x \rbk{\frac{\sminvtemperature}{2}-s_n}}}$$
as claimed.
\end{proof}

\begin{prop}
Define translations on the circle $S_{\sminvtemperature}$ by $\rbk{u_{a} \prbprocess^{\txtspin}}_{t}
= \prbprocess_{t+a}^{\txtspin}$. Then the invariance
$$\msrprb_{\txtspin,\sminvtemperature}
\circ \inv{u_{a}}
=
\msrprb_{\txtspin,\sminvtemperature}$$
holds. Similarly, let $r$ be the time-reversal map on the circle $S_{\sminvtemperature}$, and define time reversal on paths by $(R \prbprocess^{\txtspin})_{t}
= \prbprocess_{-t}^{\txtspin}$. Then the reflection invariance
$$\msrprb_{\txtspin,\sminvtemperature}
\circ \inv{R}
=
\msrprb_{\txtspin,\sminvtemperature}$$
holds.
\end{prop}

\begin{proof}
In the jump-time representation, translation invariance follows from the rotation invariance of Lebesgue measure on $\setone{t_1,\cdots,t_{2m}}
\subset S_{\sminvtemperature}$; in the operator trace representation, it follows from cyclicity of the trace. Thus this is not merely a bridge on an interval but a loop measure on the circle $S_{\sminvtemperature}$. In the jump-time representation, the sequence
$$-\frac{\sminvtemperature}{2}
< t_1
\cdots
< t_{2m}
\leq \frac{\sminvtemperature}{2}$$
is mapped by time reversal to
$$-\frac{\sminvtemperature}{2}
\leq -t_{2m}
\cdots
< -t_1
< \frac{\sminvtemperature}{2}.$$
The weight $\physcplconst^{2m} \opdmsr{t_{1}} \cdots \opdmsr{t_{2m}}$ is invariant, and the sum over the initial spin $\sum_{\qmpaulispin \in \ringratint_{2}}$ is invariant as well. Therefore the finite positive measure $\nu_{\txtspin,\sminvtemperature}$ is invariant under translations and reflection, and the probability measure $\msrprb_{\txtspin,\sminvtemperature}$, which differs from it by a constant factor, has the same invariance.
\end{proof}

From the preceding discussion, the KMS state of the non-interacting spin-only system is defined by \[\oastate[\psi_{\txtspin,\physcplconst,\sminvtemperature}](A)
=
\frac{1}{\sqfun{\trace_{\fun{\lp^{2}}{\ringratint_{2}}}}{\napiernum^{-\sminvtemperature \physham_{\txtspin}}}}
\sqfun{\trace_{\fun{\lp^{2}}{\ringratint_{2}}}}{\napiernum^{-\sminvtemperature \physham_{\txtspin}} A}
=
\frac{1}{2 \cosh (\physcplconst \sminvtemperature)}
\sqfun{\trace_{\fun{\lp^{2}}{\ringratint_{2}}}}{\napiernum^{-\sminvtemperature \physham_{\txtspin}} A}.\] For arbitrary \(\nfoldvar{s}{n}
\in I_{\sminvtemperature,<}^{n}\) and \(f_1,\cdots,f_n
\in \fun{\lp^{2}}{\ringratint_{2}}\), observables given by multiplication operators are represented by the thermal spin-path probability measure as \[\fun{\oastate[\psi_{\txtspin,\physcplconst,\sminvtemperature}]}
{\prod_{j=1}^{n} \fun{f_j}{\prbprocess_{s_j}^{\txtspin}}}
=
\int_{D(S_{\sminvtemperature};\ringratint_{2})}
\prod_{j=1}^{n}
\fun{f_j}{\prbprocess_{s_j}^{\txtspin}}
\opdmsr{\fun{\msrprb_{\txtspin,\sminvtemperature}}{\prbprocess^{\txtspin}}}.\]

Before turning to the spin-boson model, consider a perturbation by a time-dependent diagonal potential. For a bounded measurable function \(V
\colon S_{\sminvtemperature} \times \ringratint_{2}
\to \fldreal\), the Feynman-Kac weight is \(\fnexp{-\int_{S_{\sminvtemperature}}
\fun{V}{t,\prbprocess_{t}^{\txtspin}}
\opdmsr{t}}\). Define the partition function by \[\smpartitionfunc_{\txtspin,\physcplconst,V,\sminvtemperature}
=
\int_{D(S_{\sminvtemperature};\ringratint_{2})}
\fnexp{-\int_{S_{\sminvtemperature}}
\fun{V}{t,\prbprocess_{t}^{\txtspin}}
\opdmsr{t}}
\opdmsr{\msrprb_{\txtspin,\sminvtemperature}}\] and define the perturbed spin probability measure determined by it by \[\opdmsr{\msrprb_{\txtspin,V,\sminvtemperature}}
=
\frac{1}{\smpartitionfunc_{\txtspin,\physcplconst,V,\sminvtemperature}}
\fnexp{-\int_{S_{\sminvtemperature}}
\fun{V}{t,\prbprocess_{t}^{\txtspin}}
\opdmsr{t}}
\opdmsr{\msrprb_{\txtspin,\sminvtemperature}}.\]

The following identity holds for expectations.

\begin{prop}
For arbitrary $f,g
\in \sphilb{H}_{\txtspin}$,
$$\bkt{f}{\napiernum^{-t \physham_{\txtspin}} g}_{\sphilb{H}_{\txtspin}}
=
\napiernum^{\physcplconst t}
\sqfun{\prbexp_{\msrprb_{\txtpoisson}}}
{\physcplconst^{\prbpoissonprocess_{t}}
\sum_{\qmpaulispin \in \ringratint_{2}}
\cmpconj{f(\qmpaulispin)} \cdot g(\prbprocess_{\qmpaulispin,t}^{\txtspin})}$$
holds.
\end{prop}

\begin{proof}
Using the matrix exponential $\napiernum^{t \physcplconst \qmpaulispin_{x}}
=
\sum_{n=0}^{\infty}
\frac{(t \physcplconst)^n}{n!}
\qmpaulispin_{x}^n$, the action of the Pauli matrix
$$\funrbk{\qmpaulispin_{x}^n g}{\qmpaulispin}
=
\fun{g}{(-1)^n \qmpaulispin},$$
and the definition of the Poisson distribution
$$\msrprb_{\txtpoisson}(\prbpoissonprocess_{t} = n)
=
\napiernum^{-\physcplconst t} \frac{(\physcplconst t)^n}{n!},$$
we obtain
$$\frac{(t \physcplconst)^n}{n!}
=
\napiernum^{\physcplconst t} \physcplconst^{n}
\msrprb_{\txtpoisson}(\prbpoissonprocess_{t} = n).$$
Combining these identities gives
\begin{equation}
\begin{aligned}
&\bkt{f}{\napiernum^{-t \physham_{\txtspin}} g}_{\sphilb{H}_{\txtspin}}
=
\sum_{\qmpaulispin \in \ringratint_{2}}
\cmpconj{f(\qmpaulispin)}
\sum_{n=0}^{\infty}
\frac{(t \physcplconst)^n}{n!}
\funrbk{\qmpaulispin_{x}^{n} g}{\qmpaulispin}
\\ %%%%%%%%%%%%%%%%
&=
\sum_{\qmpaulispin \in \ringratint_{2}}
\cmpconj{f(\qmpaulispin)}
\sum_{n=0}^{\infty}
\napiernum^{\physcplconst t}
\physcplconst^n
\msrprb_{\txtpoisson}(\prbpoissonprocess_{t} = n)
\fun{g}{\prbprocess_{\qmpaulispin,t}^{\txtspin}}
\\ %%%%%%%%%%%%%%%%
&=
\napiernum^{\physcplconst t}
\sum_{n=0}^{\infty}
\msrprb_{\txtpoisson}(\prbpoissonprocess_{t} = n)
\physcplconst^n
\rbk{\sum_{\qmpaulispin \in \ringratint_{2}}
\cmpconj{f(\qmpaulispin)}
\fun{g}{\prbprocess_{\qmpaulispin,t}^{\txtspin}}}
=
\napiernum^{\physcplconst t}
\sqfun{\prbexp_{\msrprb_{\txtpoisson}}}
{\physcplconst^{\prbpoissonprocess_{t}}
\sum_{\qmpaulispin \in \ringratint_{2}}
\cmpconj{f(\qmpaulispin)} \cdot g(\prbprocess_{\qmpaulispin,t}^{\txtspin})}
\end{aligned}
\end{equation}
as desired.
\end{proof}

\section{BEC via the Functional Integral Representation}\label{expedition0012422}

Using the result of the preceding section and the results for the free Bose gas, one obtains a functional integral representation of the free Hamiltonian \(\physham_{0}\). The perturbation theory in the textbook \cite{DerezinskiGerard001} then gives a functional integral representation of the interacting system at finite temperature. In this section we use this functional integral representation to discuss BEC.

The free Bose field can be described by using the existing discussion of the free Bose gas \cite{AsaoArai28,YoshitsuguSekine004}. To study the emergence of BEC, we first work in a bounded system.

\subsection{Functional Integral Representation of the Free Bose Field in a Bounded System}\label{expedition0012117}

Fix arbitrary \(L
> 0\) and use the setting of Subsection \ref{expedition0011278}. Applying the construction of the bounded-system free Bose gas in \cite{YoshitsuguSekine004} for each \(L\), one obtains a singular Gaussian \(\sminvtemperature\)-Markov path space with infrared cutoff, \[\pairbk{\prbqspace_{\txtbsn,\kappa,L},
\mblfmlfrak{S}_{\txtbsn,\kappa,L},
\mblfmlfrak{S}_{\txtbsn,\kappa,0,L},
U_{\txtbsn,t},
R_{\txtbsn},
\msrprb_{\txtbsn,\kappa,\sminvtemperature,L}}.\] Denote the projection onto the one-particle space of the bounded system by \(P_{L}\) and write \[\sphilb{H}_{\txtbsn,L}
=
P_{L}\sphilb{H}_{\txtbsn}.\] For each \(t
\in S_{\sminvtemperature}\), define the isometry \(j_t\) by \[j_t
\colon
\rbk{2\omega\tanh \frac{\sminvtemperature \omega}{2}}^{\onehalf}
\sphilb{H}_{\txtbsn,\txtreal}
\to
\faadjpresharp{{\prbqspace_{\txtbsn,\sminvtemperature}}},
\quad
j_t f
=
\diracdelta_t \otimes f.\] Then the sharp-time field at time \(t
\in
S_{\sminvtemperature}\), \(\opfocksegal_t(f)
= \opfocksegal(j_{t} f)\), is defined consistently for \(f
\in
\sphilb{H}_{\txtbsn,L}\), and for arbitrary \(0
\leq
\abs{t-s}
\leq
\sminvtemperature\) its covariance is \begin{equation}\label{expedition0012440}
\sqfun{\prbexp_{\msrprb_{\txtbsn,\kappa,\sminvtemperature,L}}}
{\opfocksegal_t(f)\opfocksegal_s(g)}
=
\frac{1}{2}
\bkt{f}
{\frac{\napiernum^{-\abs{t-s}\omega}
+ \napiernum^{-(\sminvtemperature-\abs{t-s})\omega}}
{1-\napiernum^{-\sminvtemperature\omega}}g}_{\sphilb{H}_{\txtbsn,L}}
\end{equation} given by the preceding formula. In particular, if \(\opform{q}_{\txtbsn,\txtnonzero,\sminvtemperature,L}
=
\fnrestr{\opform{q}_{\txtbsn,\txtnonzero,\sminvtemperature}}
{P_{L} \sphilb{H}_{\txtbsn,\sminvtemperature}}\), then at equal times \[\sqfun{\prbexp_{\msrprb_{\txtbsn,\kappa,\sminvtemperature,L}}}
{\opfocksegal_t(f)\opfocksegal_t(g)}
=
\frac{1}{2}
\opform{q}_{\txtbsn,\txtnonzero,\sminvtemperature,L}(f,g)\] holds, so this coincides with the quasi-bilinear form introduced for the free Bose gas. The Euclidean quasi-bilinear form of the nonzero component associated with the preceding covariance is defined by \[\fun{\opform{q}_{\txtbsn,\txtnonzero,\sminvtemperature,L}^{\txteuclid}}{j_t f,j_s g}
=
\bkt{f}
{\frac{\napiernum^{-\abs{t-s}\omega}
+ \napiernum^{-(\sminvtemperature-\abs{t-s})\omega}}
{1-\napiernum^{-\sminvtemperature\omega}}g}_{\sphilb{H}_{\txtbsn,L}}.\]

In the bounded-system discussion we impose the infrared cutoff, so the test functions are also projected to \(P_{L}(\dom \mathsf{m}_{\kappa})\), and the construction of the sample space and of the algebra of bosonic physical quantities must use appropriate projections such as \(P_{L} \sphilb{D}_{\txtbsn,\txtphys,\sminvtemperature}\). Accordingly, the path space of the free Bose field is taken to be the singular Gaussian \(\sminvtemperature\)-Markov path space \cite{YoshitsuguSekine004} \begin{equation}
\pairbk{\prbqspace_{\txtbsn,\txtphys,L},
\mblfmlfrak{S}_{\txtbsn,\txtphys,L},
\mblfmlfrak{S}_{\txtbsn,\txtphys,0,L},
U_{\txtbsn,t},
R_{\txtbsn},
\msrprb_{\txtbsn,\txtphys,\sminvtemperature,L}}
\end{equation} Here \(\mblfmlfrak{S}_{\txtbsn,\txtphys,L}\) is the \(\sigma\)-algebra generated by \(\fml{\opfocksegal_t(f)}
{t \in S_{\sminvtemperature},\, f \in P_{L} \sphilb{D}_{\txtbsn,\txtphys,\sminvtemperature}}\), \(\mblfmlfrak{S}_{\txtbsn,\txtphys,0,L}\) is the time-\(0\) \(\sigma\)-algebra generated by \(\fml{\opfocksegal_0(f)}
{f \in P_{L} \sphilb{D}_{\txtbsn,\txtphys,\sminvtemperature}}\), and \(\msrprb_{\txtbsn,\txtphys,\sminvtemperature,L}\) is the restriction of \(\msrprb_{\txtbsn,\kappa,\sminvtemperature,L}\) to \(\mblfmlfrak{S}_{\txtbsn,\txtphys,L}\).

\subsection{Feynman-Kac-Nelson Kernel and KMS State in a Bounded System}\label{feynman-kac-nelson-kernel-and-kms-state-in-a-bounded-system}

Introduce, in addition, a chemical potential \(\smchemicalpotential
< 0\). Define the Euclidean quasi-bilinear form of the nonzero component of the free Bose field by \begin{equation}\label{expedition0012441}
\begin{aligned}
\fun{\opform{q}_{\txtbsn,\txtnonzero,\sminvtemperature,\smchemicalpotential,L}^{\txteuclid}}
{j_t f,j_s g}
=
\bkt{f}
{\frac{\napiernum^{-\abs{t-s}(\omega-\smchemicalpotential)}
+ \napiernum^{-(\sminvtemperature-\abs{t-s})(\omega-\smchemicalpotential)}}
{1-\napiernum^{-\sminvtemperature(\omega-\smchemicalpotential)}}g}_{\sphilb{H}_{\txtbsn,L}}
\end{aligned}
\end{equation} and let the corresponding free Gaussian probability measure be \(\msrprb_{\txtbsn,\sminvtemperature,\smchemicalpotential,L}\). As the product of the spin loop measure constructed in Subsection \ref{expedition0012117} and the measure of the free Bose field with chemical potential, define \[\Omega_{\txtspinboson,\sminvtemperature,L}
=
D(S_{\sminvtemperature};\ringratint_2)
\times
\prbqspace_{\txtbsn,L},
\quad
\msrprb_{\txtspinboson,0,\sminvtemperature,\smchemicalpotential,L}
=
\msrprb_{\txtspin,\sminvtemperature}
\otimes
\msrprb_{\txtbsn,\sminvtemperature,\smchemicalpotential,L}.\] Defining time translations and reflections componentwise as products yields a \(\sminvtemperature\)-Markov path space \cite{DerezinskiGerard001}. We write the path variables as \((\prbprocess^{\txtspin},\opfocksegal)\).

Define the local source determined by the spin path by \begin{equation}\label{expedition0012477}
\begin{aligned}
\fun{\mathsf{J}_{\kappa,L,I}^{\txtspin}}
{\prbprocess^{\txtspin}}
=
\int_I
\prbprocess_t^{\txtspin}
j_t(P_L \omega \mathsf{m}_{\kappa})
\opdmsr{t}
\end{aligned}
\end{equation} and, on the above free path space, write the Euclidean action corresponding to the spin-boson interaction \(\qmpaulispin_z
\otimes
\fun{\phi_{\txtfock}}{P_{L} \omega \mathsf{m}_{\kappa}}\) as \[\int_{S_{\sminvtemperature}}
\prbprocess_t^{\txtspin}
\fun{\opfocksegal}{j_t P_{L} \omega \mathsf{m}_{\kappa}}
\opdmsr{t}
=
\fun{\opfocksegal}
{\fun{\mathsf{J}_{\kappa,L,I}^{\txtspin}}
{\prbprocess^{\txtspin}}}.\]

\begin{prop}\label{expedition0012084}
For each interval $I
\subset
S_{\sminvtemperature}$, use \eqref{expedition0012477} and set
\begin{equation}\label{expedition0012476}
\begin{aligned}
F_{\kappa,L,I}
=
\fnexp{-\int_I
\prbprocess_t^{\txtspin}\fun{\opfocksegal}{j_t P_{L} \omega \mathsf{m}_{\kappa}}
\opdmsr{t}}
=
\fnexp{-\fun{\opfocksegal}
{\fun{\mathsf{J}_{\kappa,L,I}^{\txtspin}}
{\prbprocess^{\txtspin}}}}.
\end{aligned}
\end{equation}
Then the family $\fml{F_{\kappa,L,I}}
{I \subset S_{\sminvtemperature}}$ is a Feynman-Kac-Nelson kernel \cite{DerezinskiGerard001}. In particular, it satisfies the following properties.
\begin{enumerate}
\item
Each $F_{\kappa,L,I}$ is measurable with respect to the $\sigma$-algebra generated by $I$, and is positive.

\item
For arbitrary $1
\leq p
< \infty$, one has $F_{\kappa,L,\closedinterval{0}{\frac{\sminvtemperature}{2}}}
\in
\fun{\lp^{p}}{\Omega_{\txtspinboson,\sminvtemperature,L},\msrprb_{\txtspinboson,0,\sminvtemperature,\smchemicalpotential,L}}$, and the perturbed probability measure
\begin{equation}\label{expedition0012443}
\begin{aligned}
\opdmsr{\msrprb_{\txtspinboson,\sminvtemperature,\kappa,\smchemicalpotential,L}}
&=
\frac{1}{\smpartitionfunc_{\txtspinboson,\sminvtemperature,\kappa,\smchemicalpotential,L}}
F_{\kappa,L,S_{\sminvtemperature}}
\opdmsr{\msrprb_{\txtspinboson,0,\sminvtemperature,\smchemicalpotential,L}},
\\ %%%%%%%%%%%%%%%%
\smpartitionfunc_{\txtspinboson,\sminvtemperature,\kappa,\smchemicalpotential,L}
&=
\int_{\Omega_{\txtspinboson,\sminvtemperature,L}}
F_{\kappa,L,S_{\sminvtemperature}}
\opdmsr{\msrprb_{\txtspinboson,0,\sminvtemperature,\smchemicalpotential,L}}
\end{aligned}
\end{equation}
is well-defined. In particular, the partition function $\smpartitionfunc_{\txtspinboson,\sminvtemperature,\kappa,\smchemicalpotential,L}$ is nonzero and finite.

\item
For disjoint intervals $I,J$, one has $F_{\kappa,L,I \cup J}
=
F_{\kappa,L,I} F_{\kappa,L,J}$.

\item
For the time translation $U_{\txtbsn,t}$, one has $U_{\txtbsn,t} F_{\kappa,L,I}
=
F_{\kappa,L,I+t}$, and for the reflection $R_{\txtbsn}$, one has $R_{\txtbsn} F_{\kappa,L,I}
=
F_{\kappa,L,-I}$.

\item
Cut the circle $S_{\sminvtemperature}$ and regard it as $\closedinterval{-\frac{\sminvtemperature}{2}}{\frac{\sminvtemperature}{2}}$. Choose $I_n
=
\closedinterval{a_n}{b_n}$ and $I
=
\closedinterval{a}{b}$ as closed arcs that do not cross the cut point. If $a_n
\to a$ and $b_n
\to b$, then for arbitrary $1
\leq p
< \infty$, $F_{\kappa,L,I_n}
\to
F_{\kappa,L,I}$ holds in $\lp^p$.
\end{enumerate}
\end{prop}

\begin{proof}
(1): Since the maps $t
\mapsto \prbprocess_t^{\txtspin}$ and $t
\mapsto \fun{\opfocksegal}{j_t P_{L} \omega \mathsf{m}_{\kappa}}$ are both measurable, Riemann-sum approximation shows that $\fun{\mathsf{J}_{\kappa,L,I}^{\txtspin}}
{\prbprocess^{\txtspin}}$ is measurable with respect to the $\sigma$-algebra generated by $I$. Moreover, because $\prbprocess_t^{\txtspin}
\in \setone{\pm 1}$ and $\fun{\opfocksegal}{j_t P_{L} \omega \mathsf{m}_{\kappa}}$ is real, $\fun{\mathsf{J}_{\kappa,L,I}^{\txtspin}}
{\prbprocess^{\txtspin}}$ is real-valued and $F_{\kappa,L,I}
> 0$ follows.

(2): If the spin path $\prbprocess^{\txtspin}$ is fixed, then $\fun{\opfocksegal}
{\fun{\mathsf{J}_{\kappa,L,I}^{\txtspin}}
{\prbprocess^{\txtspin}}}$ is a centered Gaussian variable with respect to the boson field. Hence, for arbitrary $p
< \infty$,
$$\begin{aligned}
&\sqfun{\prbexp_{\msrprb_{\txtbsn,\sminvtemperature,\smchemicalpotential,L}}}
{\abs{F_{\kappa,L,I}(\prbprocess^{\txtspin},\opfocksegal)}^p}
=
\sqfun{\prbexp_{\msrprb_{\txtbsn,\sminvtemperature,\smchemicalpotential,L}}}
{\napiernum^{-p \fun{\opfocksegal}{\fun{\mathsf{J}_{\kappa,L,I}^{\txtspin}}
{\prbprocess^{\txtspin}}}}}
\\ %%%%%%%%%%%%%%%%
&=
\fnexp{\frac{p^2}{4}
\int_I \int_I
\prbprocess_{t}^{\txtspin}
\prbprocess_{s}^{\txtspin}
\fun{\opform{q}_{\txtbsn,\txtnonzero,\sminvtemperature,\smchemicalpotential,L}^{\txteuclid}}
{j_t(P_{L} \omega \mathsf{m}_{\kappa}),j_s(P_{L} \omega \mathsf{m}_{\kappa})}
\opdmsr{t}\opdmsr{s}}
\end{aligned}$$
holds. Moreover,
$$\int_I \int_I
\prbprocess_{t}^{\txtspin}\prbprocess_{s}^{\txtspin}
\fun{\opform{q}_{\txtbsn,\txtnonzero,\sminvtemperature,\smchemicalpotential,L}^{\txteuclid}}
{j_t(P_{L} \omega \mathsf{m}_{\kappa}),j_s(P_{L} \omega \mathsf{m}_{\kappa})}
\opdmsr{t}\opdmsr{s}
\leq
\abs{I}^2
\fun{\opform{q}_{\txtbsn,\txtnonzero,\sminvtemperature,\smchemicalpotential,L}}{P_{L} \omega \mathsf{m}_{\kappa}}
<
\infty$$
implies
$\int_{\Omega_{\txtspinboson,\sminvtemperature,L}}
\abs{F_{\kappa,L,I}}^p
\opdmsr{\msrprb_{\txtspinboson,0,\sminvtemperature,\smchemicalpotential,L}}
<
\infty.$
In particular, $F_{\kappa,L,I}
\in
\fun{\lp^p}{\Omega_{\txtspinboson,\sminvtemperature,L},\msrprb_{\txtspinboson,0,\sminvtemperature,\smchemicalpotential,L}}$.

(3): For disjoint intervals, the assertion follows from
$$\begin{aligned}
\int_{I \cup J}
\prbprocess_t^{\txtspin}\fun{\opfocksegal}{j_t P_{L} \omega \mathsf{m}_{\kappa}}
\opdmsr{t}
=
\int_I
\prbprocess_t^{\txtspin}\fun{\opfocksegal}{j_t P_{L} \omega \mathsf{m}_{\kappa}}
\opdmsr{t}
+
\int_J
\prbprocess_t^{\txtspin}\fun{\opfocksegal}{j_t P_{L} \omega \mathsf{m}_{\kappa}}
\opdmsr{t}
\end{aligned}$$
and gives $F_{\kappa,L,I \cup J}
= F_{\kappa,L,I} F_{\kappa,L,J}$.

(4): The time translation is handled by
$$\begin{aligned}
&\fun{U_{\txtbsn,a} F_{\kappa,L,I}}{\prbprocess^{\txtspin},\opfocksegal}
=
\fnexp{-\int_I
\prbprocess_{t+a}^{\txtspin}
\fun{\opfocksegal}{j_{t+a} P_{L} \omega \mathsf{m}_{\kappa}}
\opdmsr{t}}
\\ %%%%%%%%%%%%%%%%
&=
\fnexp{-\int_{I+a}
\prbprocess_u^{\txtspin}
\fun{\opfocksegal}{j_u P_{L} \omega \mathsf{m}_{\kappa}}
\opdmsr{u}}
=
\fun{F_{\kappa,L,I+a}}{\prbprocess^{\txtspin},\opfocksegal}
\end{aligned}.$$
The reflection is handled by
$$\begin{aligned}
&\fun{R_{\txtbsn} F_{\kappa,L,I}}{\prbprocess^{\txtspin},\opfocksegal}
=
\fnexp{-\int_I
\prbprocess_{-t}^{\txtspin}
\fun{\opfocksegal}{j_{-t} P_{L} \omega \mathsf{m}_{\kappa}}
\opdmsr{t}}
\\ %%%%%%%%%%%%%%%%
&=
\fnexp{-\int_{-I}
\prbprocess_u^{\txtspin}
\fun{\opfocksegal}{j_u P_{L} \omega \mathsf{m}_{\kappa}}
\opdmsr{u}}
=
\fun{F_{\kappa,L,-I}}{\prbprocess^{\txtspin},\opfocksegal}
\end{aligned}.$$

(5): The symmetric difference of the intervals $\Delta_n
= I_n \triangle I$ satisfies $\abs{\Delta_n}
\to 0$, and $$\abs{\fun{\mathsf{J}_{\kappa,L,I_n}^{\txtspin}}
{\prbprocess^{\txtspin}}
- \fun{\mathsf{J}_{\kappa,L,I}^{\txtspin}}
{\prbprocess^{\txtspin}}}$$ is bounded by $\fun{\mathsf{J}_{\kappa,L,\Delta_n}^{\txtspin}}
{\prbprocess^{\txtspin}}$. If the spin path is fixed, $\fun{\mathsf{J}_{\kappa,L,\Delta_n}^{\txtspin}}
{\prbprocess^{\txtspin}}$ is a centered Gaussian variable. The Gaussian isometry and the estimate in item (2) give
$$\begin{aligned}
\sqfun{\prbexp_{\msrprb_{\txtbsn,\sminvtemperature,\smchemicalpotential,L}}}
{\abs{\fun{\mathsf{J}_{\kappa,L,\Delta_n}^{\txtspin}}
{\prbprocess^{\txtspin}}}^2}
\leq
\onehalf
\abs{\Delta_n}^2
\fun{\opform{q}_{\txtbsn,\txtnonzero,\sminvtemperature,\smchemicalpotential,L}}
{P_L \omega \mathsf{m}_{\kappa}}
\end{aligned}$$
and therefore $\fun{\mathsf{J}_{\kappa,L,I_n}^{\txtspin}}
{\prbprocess^{\txtspin}}
\to \fun{\mathsf{J}_{\kappa,L,I}^{\txtspin}}
{\prbprocess^{\txtspin}}$ in $\lp^{2}$. The exponential-moment estimate in item (2) holds uniformly for all sufficiently large $n$. The local Lipschitz estimate for the exponential function, Holder's inequality, and Vitali's theorem imply $F_{\kappa,L,I_n}
\to F_{\kappa,L,I}$ in arbitrary $\lp^p$.
\end{proof}

By the local semigroup reconstruction in \cite[Definition 21.59, Proposition 21.62]{DerezinskiGerard001}, the local Feynman-Kac-Nelson kernel in Proposition \ref{expedition0012084} gives a local Hermitian semigroup and a kernel representation on its generator.

\begin{prop}[Local Hermitian semigroup and kernel representation on its generator]\label{expedition0012424}
For the local Feynman-Kac-Nelson kernel in Proposition \ref{expedition0012084}, let $\mblfmlfrak{S}_{I,L}$ be the $\sigma$-algebra of the composite system generated by the interval $I
\subset S_{\sminvtemperature}$, let $\prbexp_{I,\smchemicalpotential,L}$ denote conditional expectation with respect to the free measure $\msrprb_{\txtspinboson,0,\sminvtemperature,\smchemicalpotential,L}$, and write $\prbexp_{0,\frac{\sminvtemperature}{2},\smchemicalpotential,L}$ for conditional expectation with respect to the times $0$ and $\frac{\sminvtemperature}{2}$. For arbitrary $0
< t
< \frac{\sminvtemperature}{2}$, define the local domain by
$$\begin{aligned}
\sphilb{D}_{\txtspinboson,\kappa,\smchemicalpotential,L,t}
=
\prbexp_{0,\frac{\sminvtemperature}{2},\smchemicalpotential,L}
\fun{\splinspan}
{\bigcup_{0 \leq s < \frac{\sminvtemperature}{2} - t}
F_{\kappa,L,\closedinterval{0}{s}}
\fun{\lp^{\infty}}
{\Omega_{\txtspinboson,\sminvtemperature,L},
\mblfmlfrak{S}_{\closedinterval{0}{\frac{\sminvtemperature}{2}-t},L}}}
\end{aligned}.$$
Then for arbitrary $0
\leq s
\leq t
\leq \frac{\sminvtemperature}{2}$ and $f
\in
\fun{\lp^2}
{\Omega_{\txtspinboson,\sminvtemperature,L},
\msrprb_{\txtspinboson,0,\sminvtemperature,\smchemicalpotential,L}}$, there exists a unique linear operator
$$P_{\txtspinboson,\kappa,\smchemicalpotential,L,s}
\colon
\sphilb{D}_{\txtspinboson,\kappa,\smchemicalpotential,L,t}
\to
\sphilb{D}_{\txtspinboson,\kappa,\smchemicalpotential,L,t-s}$$
satisfying
$$\begin{aligned}
P_{\txtspinboson,\kappa,\smchemicalpotential,L,s}
\prbexp_{\setone{0,\frac{\sminvtemperature}{2}},\smchemicalpotential,L} f
=
\prbexp_{\setone{0,\frac{\sminvtemperature}{2}},\smchemicalpotential,L}
F_{\kappa,L,\closedinterval{0}{s}}
U_s f
\end{aligned}.$$
The family $\fml{P_{\txtspinboson,\kappa,\smchemicalpotential,L,t}}
{0 \leq t \leq \frac{\sminvtemperature}{2}}$ is a local Hermitian semigroup.

For bounded functions $A_{\txtspin},B_{\txtspin}$ at time $0$ and $f,g
\in
P_L\sphilb{D}_{\txtbsn,\txtphys,\sminvtemperature}$, define $A(\prbprocess^{\txtspin},\opfocksegal)
=
A_{\txtspin}(\prbprocess_0^{\txtspin})
\napiernum^{\imunit \opfocksegal(j_0 f)}$ and $B(\prbprocess^{\txtspin},\opfocksegal)
=
B_{\txtspin}(\prbprocess_0^{\txtspin})
\napiernum^{\imunit \opfocksegal(j_0 g)}$. Using $\fun{\mathsf{J}_{\kappa,L,I}^{\txtspin}}
{\prbprocess^{\txtspin}}$ from \eqref{expedition0012477}, for arbitrary $0
\leq t
\leq \frac{\sminvtemperature}{2}$ one has
\begin{equation}\label{expedition0012423}
\begin{aligned}
&\bkt{A}{P_{\txtspinboson,\kappa,\smchemicalpotential,L,t}B}
\\ %%%%%%%%%%%%%%%%
&=
\fnexp{-\oneoverfour
\fun{\opform{q}_{\txtbsn,\txtnonzero,\sminvtemperature,\smchemicalpotential,L}^{\txteuclid}}
{j_t g - j_0 f}}
\\ %%%%%%%%%%%%%%%%
&\quad\times
\prbexp_{\msrprb_{\txtspin,\sminvtemperature}}
\sqbkleft{\cmpconj{A_{\txtspin}}
B_{\txtspin}
\fnexp{-\frac{\imunit}{2}
\fun{\opform{q}_{\txtbsn,\txtnonzero,\sminvtemperature,\smchemicalpotential,L}^{\txteuclid}}
{j_t g - j_0 f,\mathsf{J}_{\kappa,L,\closedinterval{0}{t}}^{\txtspin}}}}
\\ %%%%%%%%%%%%%%%%
&\qquad\times
\sqbkright{\fnexp{\oneoverfour
\fun{\opform{q}_{\txtbsn,\txtnonzero,\sminvtemperature,\smchemicalpotential,L}^{\txteuclid}}
{\mathsf{J}_{\kappa,L,\closedinterval{0}{t}}^{\txtspin}}}}.
\end{aligned}
\end{equation}
\end{prop}

\begin{proof}
Applying Proposition \ref{expedition0012084}(5) gives $\lp^p$-continuity with respect to endpoint convergence. The multiplicativity in Proposition \ref{expedition0012084} implies
$$P_{\txtspinboson,\kappa,\smchemicalpotential,L,s}
P_{\txtspinboson,\kappa,\smchemicalpotential,L,u}
= P_{\txtspinboson,\kappa,\smchemicalpotential,L,s+u}$$
on the local domain, and reflection gives
$$\bkt{F}{P_{\txtspinboson,\kappa,\smchemicalpotential,L,s}G}
=
\bkt{P_{\txtspinboson,\kappa,\smchemicalpotential,L,s}F}{G}.$$
Therefore $\fml{P_{\txtspinboson,\kappa,\smchemicalpotential,L,t}}
{0 \leq t \leq \frac{\sminvtemperature}{2}}$ is a local Hermitian semigroup. The other properties follow from the local semigroup reconstruction argument in \cite[Proposition 21.62]{DerezinskiGerard001}.

Under the definitions of the functions in the statement, for arbitrary $0
\leq t
\leq \frac{\sminvtemperature}{2}$, substitute $U_tB
=
B_{\txtspin}(\prbprocess_t^{\txtspin})
\napiernum^{\imunit \opfocksegal(j_t g)}$ and $F_{\kappa,L,\closedinterval{0}{t}}
=
\fnexp{-\fun{\opfocksegal}
{\fun{\mathsf{J}_{\kappa,L,\closedinterval{0}{t}}^{\txtspin}}
{\prbprocess^{\txtspin}}}}$ into the defining formula for $P_{\txtspinboson,\kappa,\smchemicalpotential,L,t}
B$. Fixing the particle path and computing the centered Gaussian integral over the Bose field gives the desired formula \eqref{expedition0012423}.
\end{proof}

\begin{defn}
Let $\mblfmlfrak{F}_{\txtspinboson,\sminvtemperature,\kappa,\smchemicalpotential,L}$ be the set of all $\msrprb_{\txtspinboson,\sminvtemperature,\kappa,\smchemicalpotential,L}$-integrable Borel measurable functions on the path space $\Omega_{\txtspinboson,\sminvtemperature,L}$:
$$\mblfmlfrak{F}_{\txtspinboson,\sminvtemperature,\kappa,\smchemicalpotential,L}
=
\set{F \in \fun{\mblfn_{\txtborel}}{\Omega_{\txtspinboson,\sminvtemperature,L};\fldcmp}}
{\int_{\Omega_{\txtspinboson,\sminvtemperature,L}} \abs{F(\varpi)} \opdmsr{\msrprb_{\txtspinboson,\sminvtemperature,\kappa,\smchemicalpotential,L}} < \infty}.$$
Using the probability measure $\msr{\msrprb_{\txtspinboson,\sminvtemperature,\kappa,\smchemicalpotential,L}}$ defined in \eqref{expedition0012443}, define the linear functional $\psi_{\txtspinboson,\sminvtemperature,\kappa,\smchemicalpotential,L}$ on this space by
\begin{equation}\label{expedition0012444}
\fun{\psi_{\txtspinboson,\sminvtemperature,\kappa,\smchemicalpotential,L}}{F}
=
\sqfun{\prbexp_{\msrprb_{\txtspinboson,\sminvtemperature,\kappa,\smchemicalpotential,L}}}
{F(\prbprocess^{\txtspin},\opfocksegal)},
\quad
F
\in \mblfmlfrak{F}_{\txtspinboson,\sminvtemperature,\kappa,\smchemicalpotential,L}.
\end{equation}
\end{defn}

\begin{thm}\label{expedition0012085}
For arbitrary $L
> 0$, the probability measure $\msrprb_{\txtspinboson,\sminvtemperature,\kappa,\smchemicalpotential,L}$ in Proposition \ref{expedition0012084} gives a perturbed $\sminvtemperature$-Markov path space, and one obtains the KMS state $\oastate[\psi_{\txtspinboson,\sminvtemperature,\kappa,\smchemicalpotential,L}]$ for the time evolution defined by the local Hermitian semigroup in Proposition \ref{expedition0012424}. Moreover, for arbitrary $\nfoldvar{s}{n}
\in I_{\sminvtemperature,\leq}^{n}$ and functions $F_1,\cdots,F_n
\colon
\Omega_{\txtspinboson,\sminvtemperature,L}
\to
\fldcmp$ that are measurable with respect to the time-$0$ $\sigma$-algebra $\mblfmlfrak{S}_{0,L}
=
\mblfmlgenerated{\prbprocess_{0}^{\txtspin}}
\otimes
\mblfmlfrak{S}_{\txtbsn,0,L}$ and satisfy
$$\prod_{j=1}^n F_j \circ U_{s_j}
\in
\mblfmlfrak{F}_{\txtspinboson,\sminvtemperature,\kappa,\smchemicalpotential,L},$$
one has
\begin{equation}\label{expedition0012445}
\fun{\psi_{\txtspinboson,\sminvtemperature,\kappa,\smchemicalpotential,L}}
{\prod_{j=1}^n F_j \circ U_{s_j}}
=
\sqfun{\prbexp_{\msrprb_{\txtspinboson,\sminvtemperature,\kappa,\smchemicalpotential,L}}}
{\prod_{j=1}^{n}
F_j \circ U_{s_j}}.
\end{equation}
\end{thm}

\begin{proof}
By Proposition \ref{expedition0012424}, the probability measure in Proposition \ref{expedition0012084} gives a perturbed $\sminvtemperature$-Markov path space. Since the kernel $F_{\kappa,L,S_{\sminvtemperature}}$ on the full circle has a finite partition function, the KMS state $\oastate[\psi
_{\txtspinboson,\sminvtemperature,\kappa,\smchemicalpotential,L}]$ of the spin-boson model is obtained together with the periodic boundary condition.

(Correlation-function representation): Since each $F_j$ is measurable with respect to the time-$0$ $\sigma$-algebra $\mblfmlfrak{S}_{0,L}$, $F_j\circ U_{s_j}$ is an observable at time $s_j$. The condition $\prod_{j=1}^{n}F_j\circ U_{s_j}
\in
\mblfmlfrak{F}_{\txtspinboson,\sminvtemperature,\kappa,\smchemicalpotential,L}$ ensures that the product is integrable with respect to the perturbed measure. Applying the definition \eqref{expedition0012444}, we obtain
$$\begin{aligned}
\fun{\psi_{\txtspinboson,\sminvtemperature,\kappa,\smchemicalpotential,L}}
{\prod_{j=1}^n F_j \circ U_{s_j}}
=
\sqfun{\prbexp_{\msrprb_{\txtspinboson,\sminvtemperature,\kappa,\smchemicalpotential,L}}}
{\prod_{j=1}^{n}
F_j \circ U_{s_j}}
\end{aligned}$$
which is \eqref{expedition0012445}.
\end{proof}

\begin{prop}[Characteristic functional in a bounded system]\label{expedition0012086}
Using $\fun{\mathsf{J}_{\kappa,L,S_{\sminvtemperature}}^{\txtspin}}
{\prbprocess^{\txtspin}}$ from \eqref{expedition0012477}, define the following functions of the spin path:
\begin{equation}\label{expedition0012446}
\begin{aligned}
W_{\txtspinboson,\sminvtemperature,\kappa,\smchemicalpotential,L}(t,s)
&=
\fun{\opform{q}_{\txtbsn,\txtnonzero,\sminvtemperature,\smchemicalpotential,L}^{\txteuclid}}
{j_t(P_{L} \omega \mathsf{m}_{\kappa}),j_s(P_{L} \omega \mathsf{m}_{\kappa})},
\\ %%%%%%%%%%%%%%%%
\fun{\mathsf{D}_{\txtspinboson,\sminvtemperature,\kappa,\smchemicalpotential,L}}
{\prbprocess^{\txtspin}}
&=
\fnexp{\oneoverfour
\fun{\opform{q}_{\txtbsn,\txtnonzero,\sminvtemperature,\smchemicalpotential,L}^{\txteuclid}}
{\fun{\mathsf{J}_{\kappa,L,S_{\sminvtemperature}}^{\txtspin}}
{\prbprocess^{\txtspin}}}}
\\ %%%%%%%%%%%%%%%%
&=
\fnexp{\oneoverfour
\int_{S_{\sminvtemperature}}\int_{S_{\sminvtemperature}}
\prbprocess_t^{\txtspin}\prbprocess_s^{\txtspin}
W_{\txtspinboson,\sminvtemperature,\kappa,\smchemicalpotential,L}(t,s)
\opdmsr{t}\opdmsr{s}}.
\end{aligned}
\end{equation}
Then the partition function can also be written as $\smpartitionfunc_{\txtspinboson,\sminvtemperature,\kappa,\smchemicalpotential,L}
=
\sqfun{\prbexp_{\msrprb_{\txtspin,\sminvtemperature}}}
{\mathsf{D}_{\txtspinboson,\sminvtemperature,\kappa,\smchemicalpotential,L}}$. Define the normalized interacting spin-path measure in the bounded system by
\begin{equation}\label{expedition0012473}
\opdmsr{\msrprb_{\txtspin,\txtspinboson,\sminvtemperature,\kappa,\smchemicalpotential,L}(\prbprocess^{\txtspin})}
=
\frac{1}{\smpartitionfunc_{\txtspinboson,\sminvtemperature,\kappa,\smchemicalpotential,L}}
\mathsf{D}_{\txtspinboson,\sminvtemperature,\kappa,\smchemicalpotential,L}(\prbprocess^{\txtspin})
\opdmsr{\msrprb_{\txtspin,\sminvtemperature}}(\prbprocess^{\txtspin}).
\end{equation}
Then, for arbitrary $f
\in
\sphilb{H}_{\txtbsn,L}$ and $t
\in
S_{\sminvtemperature}$, the regularized state of the full system $\psi_{\txtspinboson,\sminvtemperature,\kappa,\smchemicalpotential,L}$ is represented as
\begin{equation}
\begin{aligned}
&\fun{\psi_{\txtspinboson,\sminvtemperature,\kappa,\smchemicalpotential,L}}
{\napiernum^{\imunit \opfocksegal(j_t f)}}
=
\fnexp{-\oneoverfour
\fun{\opform{q}_{\txtbsn,\txtnonzero,\sminvtemperature,\smchemicalpotential,L}^{\txteuclid}}{j_t f}}
\\ %%%%%%%%%%%%%%%%
&\quad\times
\sqfun{\prbexp_{\msrprb_{\txtspin,\txtspinboson,\sminvtemperature,\kappa,\smchemicalpotential,L}}}
{\fnexp{-\frac{\imunit}{2}
\int_{S_{\sminvtemperature}}
\fun{\opform{q}_{\txtbsn,\txtnonzero,\sminvtemperature,\smchemicalpotential,L}^{\txteuclid}}
{j_t f,
\mathsf{J}_{\kappa,L,S_{\sminvtemperature}}^{\txtspin}}
\opdmsr{s}}}.
\end{aligned}
\end{equation}
In particular, the quadratic Gaussian factor in $f$ is given by the Euclidean quasi-bilinear form of the nonzero component of the free Bose field.
\end{prop}

\begin{proof}
By definition \eqref{expedition0012444},
$$\fun{\psi_{\txtspinboson,\sminvtemperature,\kappa,\smchemicalpotential,L}}
{\napiernum^{\imunit \opfocksegal(j_t f)}}
=
\frac{\int_{D(S_{\sminvtemperature};\ringratint_2)}
\int_{\prbqspace_{\txtbsn,L}}
\napiernum^{\imunit \opfocksegal(j_t f)}
\napiernum^{-\fun{\opfocksegal}{\fun{\mathsf{J}_{\kappa,L,S_{\sminvtemperature}}^{\txtspin}}
{\prbprocess^{\txtspin}}}}
\opdmsr{\msrprb_{\txtbsn,\sminvtemperature,\smchemicalpotential,L}}
\opdmsr{\msrprb_{\txtspin,\sminvtemperature}}}
{\smpartitionfunc_{\txtspinboson,\sminvtemperature,\kappa,\smchemicalpotential,L}}.$$
The integral with respect to the boson field reduces to $\sqfun{\prbexp_{\msrprb_{\txtbsn,\sminvtemperature,\smchemicalpotential,L}}}
{\napiernum^{\imunit \opfocksegal(j_t f)
-\fun{\opfocksegal}{\fun{\mathsf{J}_{\kappa,L,S_{\sminvtemperature}}^{\txtspin}}
{\prbprocess^{\txtspin}}}}}$. Similarly, the partition function in the denominator is written as $\smpartitionfunc_{\txtspinboson,\sminvtemperature,\kappa,\smchemicalpotential,L}
=
\sqfun{\prbexp_{\msrprb_{\txtspin,\sminvtemperature}}}
{\mathsf{D}_{\txtspinboson,\sminvtemperature,\kappa,\smchemicalpotential,L}}$. The centered Gaussian-measure formula
\begin{equation}\label{expedition0012472}
\begin{aligned}
\sqfun{\prbexp_{\msrprb_{\txtbsn,\sminvtemperature,\smchemicalpotential,L}}}
{\napiernum^{\imunit \opfocksegal(F)+\lambda \opfocksegal(G)}}
=
\fnexp{-\oneoverfour \fun{\opform{q}_{\txtbsn,\txtnonzero,\sminvtemperature,\smchemicalpotential,L}^{\txteuclid}}{F}
+\frac{\imunit \lambda}{2}\fun{\opform{q}_{\txtbsn,\txtnonzero,\sminvtemperature,\smchemicalpotential,L}^{\txteuclid}}{F,G}
+\frac{\lambda^2}{4}\fun{\opform{q}_{\txtbsn,\txtnonzero,\sminvtemperature,\smchemicalpotential,L}^{\txteuclid}}{G}}
\end{aligned}
\end{equation}
with $F=j_t f$,
$G=\fun{\mathsf{J}_{\kappa,L,S_{\sminvtemperature}}^{\txtspin}}
{\prbprocess^{\txtspin}}$,
and $\lambda=-1$ gives
\begin{equation}
\begin{aligned}
&\sqfun{\prbexp_{\msrprb_{\txtbsn,\sminvtemperature,\smchemicalpotential,L}}}
{\napiernum^{\imunit \opfocksegal(j_t f)}
\napiernum^{-\fun{\opfocksegal}{\fun{\mathsf{J}_{\kappa,L,S_{\sminvtemperature}}^{\txtspin}}
{\prbprocess^{\txtspin}}}}}
\\ %%%%%%%%%%%%%%%%
&=
\fnexp{-\oneoverfour
\fun{\opform{q}_{\txtbsn,\txtnonzero,\sminvtemperature,\smchemicalpotential,L}^{\txteuclid}}{j_t f}
-\frac{\imunit}{2}
\fun{\opform{q}_{\txtbsn,\txtnonzero,\sminvtemperature,\smchemicalpotential,L}^{\txteuclid}}
{j_t f,\fun{\mathsf{J}_{\kappa,L,S_{\sminvtemperature}}^{\txtspin}}
{\prbprocess^{\txtspin}}}}
\\ %%%%%%%%%%%%%%%%
&\quad\times
\fnexp{\oneoverfour
\fun{\opform{q}_{\txtbsn,\txtnonzero,\sminvtemperature,\smchemicalpotential,L}^{\txteuclid}}{\fun{\mathsf{J}_{\kappa,L,S_{\sminvtemperature}}^{\txtspin}}
{\prbprocess^{\txtspin}}}}.
\end{aligned}
\end{equation}
Here
\begin{equation}
\begin{aligned}
\fun{\opform{q}_{\txtbsn,\txtnonzero,\sminvtemperature,\smchemicalpotential,L}^{\txteuclid}}{\fun{\mathsf{J}_{\kappa,L,S_{\sminvtemperature}}^{\txtspin}}
{\prbprocess^{\txtspin}}}
&=
\int_{S_{\sminvtemperature}}\int_{S_{\sminvtemperature}}
\prbprocess_t^{\txtspin}\prbprocess_s^{\txtspin}
\fun{\opform{q}_{\txtbsn,\txtnonzero,\sminvtemperature,\smchemicalpotential,L}^{\txteuclid}}
{j_t(P_{L} \omega \mathsf{m}_{\kappa}),j_s(P_{L} \omega \mathsf{m}_{\kappa})}
\opdmsr{t}\opdmsr{s},
\\ %%%%%%%%%%%%%%%%
\fun{\opform{q}_{\txtbsn,\txtnonzero,\sminvtemperature,\smchemicalpotential,L}^{\txteuclid}}{j_t f,\fun{\mathsf{J}_{\kappa,L,S_{\sminvtemperature}}^{\txtspin}}
{\prbprocess^{\txtspin}}}
&=
\int_{S_{\sminvtemperature}}
\prbprocess_s^{\txtspin}
\fun{\opform{q}_{\txtbsn,\txtnonzero,\sminvtemperature,\smchemicalpotential,L}^{\txteuclid}}
{j_t f, j_s(P_{L} \omega \mathsf{m}_{\kappa})}
\opdmsr{s}
\end{aligned}
\end{equation}
and returning to the outer spin integral gives
\begin{equation}
\begin{aligned}
&\fun{\psi_{\txtspinboson,\sminvtemperature,\kappa,\smchemicalpotential,L}}
{\napiernum^{\imunit \opfocksegal(j_t f)}}
=
\fnexp{-\oneoverfour
\fun{\opform{q}_{\txtbsn,\txtnonzero,\sminvtemperature,\smchemicalpotential,L}^{\txteuclid}}{j_t f}}
\\ %%%%%%%%%%%%%%%%
&\quad\times
\sqfun{\prbexp_{\msrprb_{\txtspin,\txtspinboson,\sminvtemperature,\kappa,\smchemicalpotential,L}}}
{\fnexp{-\frac{\imunit}{2}
\int_{S_{\sminvtemperature}}
\fun{\opform{q}_{\txtbsn,\txtnonzero,\sminvtemperature,\smchemicalpotential,L}^{\txteuclid}}
{j_t f,
\mathsf{J}_{\kappa,L,S_{\sminvtemperature}}^{\txtspin}}
\opdmsr{s}}}
\end{aligned}
\end{equation}
as claimed.

Substituting $F=0$,
$G=\fun{\mathsf{J}_{\kappa,L,S_{\sminvtemperature}}^{\txtspin}}
{\prbprocess^{\txtspin}}$,
and $\lambda=-1$ into the same Gaussian integral formula gives
$$\sqfun{\prbexp_{\msrprb_{\txtbsn,\sminvtemperature,\smchemicalpotential,L}}}
{\fnexp{-\fun{\opfocksegal}{\fun{\mathsf{J}_{\kappa,L,S_{\sminvtemperature}}^{\txtspin}}
{\prbprocess^{\txtspin}}}}}
=
\fnexp{\oneoverfour
\fun{\opform{q}_{\txtbsn,\txtnonzero,\sminvtemperature,\smchemicalpotential,L}^{\txteuclid}}{\fun{\mathsf{J}_{\kappa,L,S_{\sminvtemperature}}^{\txtspin}}
{\prbprocess^{\txtspin}}}}.$$
This gives, in particular, the representation of the partition function.
\end{proof}

In the discussion below, the interacting spin-path measure of the infinite system \(\msrprb_{\txtspin,\txtspinboson,\sminvtemperature,\kappa}\) is defined as the weak limit of \(\msrprb_{\txtspin,\txtspinboson,\sminvtemperature,\kappa,\smchemicalpotential,L}\) under the infinite-volume limit \(L
\to \infty\) and the chemical-potential limit \(\smchemicalpotential
\to 0\).

\subsection{Infinite-Volume Limit and Formal Emergence of BEC}\label{expedition0012458}

As in the van Hove model \cite{YoshitsuguSekine006}, the interaction in the finite-temperature spin-boson model is linear in the Bose field. After carrying out the Gaussian integral in the characteristic functional, the Bose field remains Gaussian with only its mean changed, and the presence or absence of BEC is reduced to the zero-mode component of the covariance. In particular, the spin degree of freedom does not create a new quasi-bilinear form on the zero mode; as in the free Bose gas and the van Hove model, \(\opform{q}_{\txtbsn,0,\sminvtemperature}\) remains as the condensed component. By Proposition \ref{expedition0012086}, to detect the emergence of BEC it is enough to look at the limit of \(\fun{\opform{q}_{\txtbsn,\txtnonzero,\sminvtemperature,\smchemicalpotential,L}^{\txteuclid}}
{j_t f}\), and, as in the discussion of the free Bose gas, one first takes the infinite-volume limit and then takes the limit in which the chemical potential goes to \(0\). Concretely, for the characteristic functional at finite \(L,\smchemicalpotential\) obtained in Proposition \ref{expedition0012086}, we take the two-step limit \(L\to\infty\) and \(\smchemicalpotential\to 0\). The convergence of \(\opform{q}_{\txtbsn,\txtnonzero,\sminvtemperature,\smchemicalpotential,L}^{\txteuclid}\) and of the integral kernel of the nonzero-mode component produces \(\opform{q}_{\txtbsn,\txtbec,\sminvtemperature}\), and gives the decomposition into a BEC-type zero-mode factor and a noncondensed component.

\begin{prop}[Functional integral representation of the infinite-system KMS state and $n$-point functions]\label{expedition0012468}
Assume the conditions under which BEC appears in the free Bose gas \cite{AsaoArai28,YoshitsuguSekine004}, and define the Euclideanized quasi-bilinear forms by
\begin{equation}\label{expedition0012475}
\begin{aligned}
\fun{\opform{q}_{\txtbsn,\txtbec,\sminvtemperature}^{\txteuclid}}
{j_t f,j_s g}
&=
\fun{\opform{q}_{\txtbsn,0,\sminvtemperature}^{\txteuclid}}{j_t f, j_s g}
+\fun{\opform{q}_{\txtbsn,\txtnonzero,\sminvtemperature}^{\txteuclid}}{j_t f, j_s g}
\\ %%%%%%%%%%%%%%%%
\fun{\opform{q}_{\txtbsn,0,\sminvtemperature}^{\txteuclid}}{j_t f, j_s g}
&=
\fun{\opform{q}_{\txtbsn,0,\sminvtemperature}}{f,g}
=
2 (2 \pi)^d \smnumberdensity_{\txtbsn,0}(\sminvtemperature)
\cmpconj{\faftr{f}(0)}
\faftr{g}(0),
\\ %%%%%%%%%%%%%%%%
\fun{\opform{q}_{\txtbsn,\txtnonzero,\sminvtemperature}^{\txteuclid}}
{j_t f,j_s g}
&=
\bkt{f}
{\frac{\napiernum^{-\abs{t-s} \omega}
+\napiernum^{-(\sminvtemperature-\abs{t-s}) \omega}}
{1-\napiernum^{-\sminvtemperature \omega}}g}_{\sphilb{H}_{\txtbsn}},
\\ %%%%%%%%%%%%%%%%
\fun{\opform{q}_{\txtbsn,\txtnonzero,\sminvtemperature,\smchemicalpotential}^{\txteuclid}}
{j_t f,j_s g}
&=
\bkt{f}
{\frac{\napiernum^{-\sminvtemperature \smchemicalpotential} \napiernum^{-\abs{t-s} \omega}
+\napiernum^{-(\sminvtemperature-\abs{t-s}) \omega}}
{\napiernum^{-\sminvtemperature \smchemicalpotential}
-\napiernum^{-\sminvtemperature \omega}}g}_{\sphilb{H}_{\txtbsn}}.
\end{aligned}
\end{equation}
For the bounded-system perturbed measure $\msrprb
_{\txtspinboson,\sminvtemperature,\kappa,\smchemicalpotential,L}$ in \eqref{expedition0012443}, define the infinite-system perturbed measure $\msrprb_{\txtspinboson,\sminvtemperature,\kappa}$ as the weak limit under the infinite-volume limit $L
\to \infty$ and the chemical-potential limit $\smchemicalpotential
\to 0$ by
\begin{equation}\label{expedition0012474}
\begin{aligned}
\opdmsr{\msrprb_{\txtspinboson,\sminvtemperature,\kappa}}
&=
\frac{1}{\smpartitionfunc_{\txtspinboson,\sminvtemperature,\kappa}}
F_{\kappa,S_{\sminvtemperature}}
\opdmsr{\msrprb_{\txtspinboson,0,\sminvtemperature}},
\\ %%%%%%%%%%%%%%%%
\smpartitionfunc_{\txtspinboson,\sminvtemperature,\kappa}
&=
\sqfun{\prbexp_{\msrprb_{\txtspinboson,0,\sminvtemperature}}}
{F_{\kappa,S_{\sminvtemperature}}},
\\ %%%%%%%%%%%%%%%%
F_{\kappa,I}
&=
\fnexp{-\int_I
\prbprocess_t^{\txtspin}
\fun{\opfocksegal}{j_t \omega \mathsf{m}_{\kappa}}
\opdmsr{t}}
=
\fnexp{-\fun{\opfocksegal}
{\fun{\mathsf{J}_{\kappa,I}^{\txtspin}}
{\prbprocess^{\txtspin}}}},
\\ %%%%%%%%%%%%%%%%
\fun{\mathsf{J}_{\kappa,I}^{\txtspin}}
{\prbprocess^{\txtspin}}
&=
\int_I
\prbprocess_t^{\txtspin}
j_t(\omega \mathsf{m}_{\kappa})
\opdmsr{t}
\end{aligned}
\end{equation}
and denote the corresponding KMS state by $\oastate[\psi_{\txtspinboson,\sminvtemperature,\kappa}]$.

Furthermore, for arbitrary $n
\geq 1$, test functions $h_1, \ldots, h_n
\in \sphilb{D}_{\txtbsn,\txtphys,\sminvtemperature}$, times $t_1, \ldots, t_n
\in S_{\sminvtemperature}$, and real coefficients $a_1, \ldots, a_n
\in \fldreal$, the characteristic functional of the $n$-time exponential observable is
\begin{equation}\label{expedition0012469}
\begin{aligned}
&\fun{\oastate[\psi_{\txtspinboson,\sminvtemperature,\kappa}]}
{\fnexp{\imunit \sum_{k=1}^{n} a_k \opfocksegal(j_{t_k} h_k)}}
=
\fnexp{-\oneoverfour
\sum_{j,k=1}^{n} a_j a_k
\fun{\opform{q}_{\txtbsn,\txtbec,\sminvtemperature}}{j_{t_j} h_j, j_{t_k} h_k}}
\\ %%%%%%%%%%%%%%%%
&\quad\times
\sqfun{\prbexp_{\msrprb_{\txtspin,\txtspinboson,\sminvtemperature,\kappa}}}
{\fnexp{-\imunit \int_0^{\sminvtemperature} \prbprocess_t^{\txtspin}
\rbk{\sum_{k=1}^{n} a_k
\fun{\opform{q}_{\txtbsn,\txtbec,\sminvtemperature}^{\txteuclid}}
{j_{t_k} h_k, j_t(\omega\mathsf{m}_{\kappa})}} \opdmsr{t}}},
\\ %%%%%%%%%%%%%%%%
&\opdmsr{\msrprb_{\txtspin,\txtspinboson,\sminvtemperature,\kappa}(\prbprocess^{\txtspin})}
=
\frac{1}{\smpartitionfunc_{\txtspinboson,\sminvtemperature,\kappa}}
\mathsf{D}_{\txtspinboson,\sminvtemperature,\kappa}(\prbprocess^{\txtspin})
\opdmsr{\msrprb_{\txtspin,\sminvtemperature}}(\prbprocess^{\txtspin}),
\\ %%%%%%%%%%%%%%%%
&\smpartitionfunc_{\txtspinboson,\sminvtemperature,\kappa}
=
\sqfun{\prbexp_{\msrprb_{\txtspin,\sminvtemperature}}}
{\mathsf{D}_{\txtspinboson,\sminvtemperature,\kappa}}
\\ %%%%%%%%%%%%%%%%
&\mathsf{D}_{\txtspinboson,\sminvtemperature,\kappa}(\prbprocess^{\txtspin})
=
\fnexp{\oneoverfour
\int_{S_{\sminvtemperature}}\int_{S_{\sminvtemperature}}
\prbprocess_t^{\txtspin}\prbprocess_s^{\txtspin}
W_{\txtspinboson,\sminvtemperature,\kappa}(t,s)
\opdmsr{t}\opdmsr{s}},
\\ %%%%%%%%%%%%%%%%
&W_{\txtspinboson,\sminvtemperature,\kappa}(t,s)
=
\fun{\opform{q}_{\txtbsn,\txtbec,\sminvtemperature}^{\txteuclid}}
{j_t(\omega \mathsf{m}_{\kappa}),j_s(\omega \mathsf{m}_{\kappa})}.
\end{aligned}
\end{equation}
In particular, the $n$-point function $\fun{\oastate[\psi_{\txtspinboson,\sminvtemperature,\kappa}]}
{\prod_{k=1}^{n} \opfocksegal(j_{t_k} h_k)}$ is uniquely determined from \eqref{expedition0012469} as the partial derivative at $a
= 0$ with respect to the coefficients $a_1, \ldots, a_n$:
\begin{equation}\label{expedition0012470}
\fun{\oastate[\psi_{\txtspinboson,\sminvtemperature,\kappa}]}
{\prod_{k=1}^{n} \opfocksegal(j_{t_k} h_k)}
=
\fnrestr{(-\imunit)^n \oppd{a_1} \cdots \oppd{a_n}
\fun{\oastate[\psi_{\txtspinboson,\sminvtemperature,\kappa}]}
{\fnexp{\imunit \sum_{k=1}^{n} a_k \opfocksegal(j_{t_k} h_k)}}}
{a_1 = \cdots = a_n = 0}.
\end{equation}
\end{prop}

\begin{proof}
Apply the Feynman-Kac-Nelson representation of the bounded-system KMS state from Theorem \ref{expedition0012085} and the kernel representation of the local Hermitian semigroup from Proposition \ref{expedition0012424} for arbitrary chemical potential $\smchemicalpotential
< 0$ and finite volume $L
> 0$. In the notation of Proposition \ref{expedition0012424}, the bounded-system characteristic functional is
$$\begin{aligned}
&\fun{\psi_{\txtspinboson,\sminvtemperature,\kappa,\smchemicalpotential,L}}
{\fnexp{\imunit \sum_{k=1}^{n} a_k \opfocksegal(j_{t_k} h_k)}}
\\ %%%%%%%%%%%%%%%%
&=
\frac{\int_{D(S_{\sminvtemperature};\ringratint_2)}
\int_{\prbqspace_{\txtbsn,L}}
\fnexp{\imunit \sum_k a_k \opfocksegal(j_{t_k} h_k)}
F_{\kappa,L,S_{\sminvtemperature}}
\opdmsr{\msrprb_{\txtbsn,\sminvtemperature,\smchemicalpotential,L}}
\opdmsr{\msrprb_{\txtspin,\sminvtemperature}}}
{\smpartitionfunc_{\txtspinboson,\sminvtemperature,\kappa,\smchemicalpotential,L}}.
\end{aligned}$$

Apply the centered Gaussian-measure integration formula \eqref{expedition0012472} with $F
=
\sum_k a_k j_{t_k} h_k$,
$G
=
\fun{\mathsf{J}_{\kappa,L,S_{\sminvtemperature}}^{\txtspin}}
{\prbprocess^{\txtspin}}$,
$\lambda
=
-1$, and integrate the Bose-field component. As in Proposition \ref{expedition0012086}, the positive Gaussian factor of the partition function, $\fnexp{\oneoverfour
\fun{\opform{q}_{\txtbsn,\txtnonzero,\sminvtemperature,\smchemicalpotential,L}^{\txteuclid}}
{\fun{\mathsf{J}_{\kappa,L,S_{\sminvtemperature}}^{\txtspin}}
{\prbprocess^{\txtspin}}}}$, cancels between the numerator and the partition function. The path-integral representation with respect to $\msrprb_{\txtspin,\txtspinboson,\sminvtemperature,\kappa,\smchemicalpotential,L}$ defined in \eqref{expedition0012473}, using $\opform{q}_{\txtbsn,\txtnonzero,\sminvtemperature,\smchemicalpotential,L}^{\txteuclid}$ defined in \eqref{expedition0012441}, is
$$\begin{aligned}
&\fun{\psi_{\txtspinboson,\sminvtemperature,\kappa,\smchemicalpotential,L}}
{\fnexp{\imunit \sum_{k=1}^{n} a_k \opfocksegal(j_{t_k} h_k)}}
\\ %%%%%%%%%%%%%%%%
&=
\fnexp{-\oneoverfour
\sum_{j,k=1}^{n} a_j a_k
\fun{\opform{q}_{\txtbsn,\txtnonzero,\sminvtemperature,\smchemicalpotential,L}^{\txteuclid}}
{j_{t_j} h_j, j_{t_k} h_k}}
\\ %%%%%%%%%%%%%%%%
&\quad\times
\sqfun{\prbexp_{\msrprb_{\txtspin,\txtspinboson,\sminvtemperature,\kappa,\smchemicalpotential,L}}}
{\fnexp{-\imunit \int_0^{\sminvtemperature} \prbprocess_t^{\txtspin}
\rbk{\sum_{k=1}^{n} a_k
\fun{\opform{q}_{\txtbsn,\txtnonzero,\sminvtemperature,\smchemicalpotential,L}^{\txteuclid}}
{j_{t_k} h_k, j_t(P_L\omega\mathsf{m}_{\kappa})}} \opdmsr{t}}}.
\end{aligned}$$

Now take the infinite-volume limit $L
\to \infty$ and the chemical-potential limit $\smchemicalpotential
\to 0$. The spin component $\msrprb_{\txtspin,\txtspinboson,\sminvtemperature,\kappa,\smchemicalpotential,L}$ of the bounded-system perturbed measure converges weakly by boundedness of the spin phase factor on the closed unit disk, and one obtains the interacting spin-path measure $\msrprb_{\txtspin,\txtspinboson,\sminvtemperature,\kappa}$ of the infinite system and the corresponding KMS state $\oastate[\psi_{\txtspinboson,\sminvtemperature,\kappa}]$. Moreover,
$$\begin{aligned}
\opform{q}_{\txtbsn,\txtnonzero,\sminvtemperature,\smchemicalpotential,L}^{\txteuclid}
&\to \opform{q}_{\txtbsn,\txtnonzero,\sminvtemperature,\smchemicalpotential}^{\txteuclid}, \\
P_L \omega\mathsf{m}_{\kappa}
&\to \omega\mathsf{m}_{\kappa}, \\
\msrprb_{\txtspin,\txtspinboson,\sminvtemperature,\kappa,\smchemicalpotential,L}
&\to \msrprb_{\txtspin,\txtspinboson,\sminvtemperature,\kappa}
\end{aligned}$$
hold, where the last convergence is weak convergence of the interacting spin-path measures. In the chemical-potential limit $\smchemicalpotential
\to 0$, the zero-mode component becomes singular, the zero-mode covariance $\fun{\opform{q}_{\txtbsn,0,\sminvtemperature}}
{h_j, h_k}$ due to the BEC condensate density $\smnumberdensity_{\txtbsn,0}(\sminvtemperature)$ separates off \cite{AsaoArai28,YoshitsuguSekine004}, and the Gaussian factor can be expressed by $\opform{q}_{\txtbsn,\txtbec,\sminvtemperature}$ as on the right-hand side of \eqref{expedition0012469} for $\oastate[\psi_{\txtspinboson,\sminvtemperature,\kappa}]$. The limit of the spin-path factor is reduced to the path integral on the right-hand side of \eqref{expedition0012469} by convergence of the integrand and weak convergence of the interacting spin-path measure.

The $n$-point function representation is the standard fact that it is obtained from the $n$th partial derivative of the characteristic functional with respect to the coefficients $a_1,
\ldots,
a_n$.
\end{proof}

\begin{prop}\label{expedition0012099}
Consider the infinite-system KMS state $\oastate[\psi_{\txtspinboson,\sminvtemperature,\kappa}]$ constructed in Proposition \ref{expedition0012468}. For arbitrary $f
\in
\sphilb{D}_{\txtbsn,\txtphys,\sminvtemperature}$ and $t
\in
S_{\sminvtemperature}$,
\begin{equation}\label{expedition0012447}
\begin{aligned}
\fun{\oastate[\psi_{\txtspinboson,\sminvtemperature,\kappa}]}
{\napiernum^{\imunit \opfocksegal(j_t f)}}
&=
\fnexp{-\oneoverfour
\fun{\opform{q}_{\txtbsn,\txtbec,\sminvtemperature}}{f}}
\cdot
\sqfun{\prbexp_{\msrprb_{\txtspin,\txtspinboson,\sminvtemperature,\kappa}}}
{\fnexp{-\frac{\imunit}{2}
\fun{\opform{q}_{\txtbsn,\txtnonzero,\sminvtemperature}^{\txteuclid}}
{j_t f,
\mathsf{J}_{\kappa,S_{\sminvtemperature}}^{\txtspin}}}}
\end{aligned}
\end{equation}
holds. Note that in the second factor of the product, the quasi-bilinear form changes to $\opform{q}_{\txtbsn,\txtnonzero,\sminvtemperature}^{\txteuclid}$. In particular, $\opform{q}_{\txtbsn,0,\sminvtemperature}$ appears in the same form as in the free Bose gas and the van Hove model, and BEC appears formally also in the spin-boson model.
\end{prop}

\begin{proof}
Apply Proposition \ref{expedition0012468} with $n
=
1$,
$t_1
=
t$,
$h_1
=
f$,
$a_1
=
1$. At equal times,
$$\fun{\opform{q}_{\txtbsn,\txtnonzero,\sminvtemperature}^{\txteuclid}}
{j_t f, j_t f}
=
\fun{\opform{q}_{\txtbsn,\txtnonzero,\sminvtemperature}}{f, f}$$
holds, and by the definition $\opform{q}_{\txtbsn,\txtbec,\sminvtemperature}
=
\opform{q}_{\txtbsn,0,\sminvtemperature}
+\opform{q}_{\txtbsn,\txtnonzero,\sminvtemperature}$, the Gaussian factor on the right-hand side of \eqref{expedition0012469} reduces to $\fnexp{-\oneoverfour \fun{\opform{q}_{\txtbsn,\txtbec,\sminvtemperature}}{f}}$. Because the infrared cutoff restricts us to $\kappa
> 0$, the integrand in the spin-path factor vanishes at the origin in momentum space, and hence
$$\fun{\opform{q}_{\txtbsn,0,\sminvtemperature}}
{j_t f,
j_s \omega \mathsf{m}_{\kappa}}
= 0$$
holds. Thus $\opform{q}_{\txtbsn,\txtbec,\sminvtemperature}$ reduces to $\opform{q}_{\txtbsn,\txtnonzero,\sminvtemperature}$. Combining these facts gives \eqref{expedition0012447}.
\end{proof}

This proposition shows that, under the infrared cutoff, the spin degree of freedom neither destroys nor generates the BEC component, and the emergence of BEC itself at finite temperature is governed by the zero-mode structure of the free Bose gas.

\subsection{Removal of the Infrared Cutoff}\label{expedition0012574}

Removing the infrared cutoff means taking the limit of the cutoff parameter in the source. The objects whose limits are taken are field correlation functions on the physical test-function space \(\sphilb{D}_{\txtbsn,\txtphys,\sminvtemperature}
= \dom \mathsf{m} \cap \sphilb{D}_{\txtbsn,0,\sminvtemperature}\). We first establish the path-integral representation with the cutoff source, and then take the source limit.

We first distinguish two kinds of convergence of the source. Consider the infrared cutoff removal adapted to the assumption \(\varrho
\in \dom \omega^{-\onehalf}\cap\dom \inv{\omega}\): \begin{equation}\label{expedition0012480}
\begin{aligned}
\omega \mathsf{m}_{\kappa}
= \omega^{-\onehalf}\varrho_{\kappa}
&\to
\omega \mathsf{m}
= \omega^{-\onehalf}\varrho
\quad\text{in }\sphilb{H}_{\txtbsn},
\\ %%%%%%%%%%%%%%%%
\omega^{-\onehalf}(\omega \mathsf{m}_{\kappa})
=
\omega^{-1}\varrho_{\kappa}
&\to
\omega^{-1}\varrho
=
\omega^{-\onehalf}(\omega\mathsf{m})
\quad\text{in }\sphilb{H}_{\txtbsn}.
\end{aligned}
\end{equation} On the other hand, the infrared singularity appears not as convergence of \(\mathsf{m}_{\kappa}\) itself in \(\sphilb{H}_{\txtbsn}\), but as a shrinking of the domain of the linear functional \(f
\mapsto \mathsf{m}(f)\). Thus, if \(f
\in \dom \mathsf{m}\), then \(\mathsf{m}_{\kappa}(f)
=
\bkt{f}{\mathsf{m}_{\kappa}}_{\sphilb{H}_{\txtbsn}}
\to
\mathsf{m}(f)\) holds, while for \(f
\notin \dom \mathsf{m}\) the right-hand side is not defined as a finite physical quantity. Therefore, after removal of the infrared cutoff, fields are restricted from the outset to \(\sphilb{D}_{\txtbsn,\txtphys,\sminvtemperature}\). Note that \begin{equation}\label{expedition0012488}
\norm{\mathsf{m}_{\kappa}}_{\sphilb{H}_{\txtbsn}}^2
=
\int_{\fldreal^d}
\frac{\abs{\faftr{\varrho}_{\kappa}(k)}^2}{\omega(k)^3}
\opdmsr{k}
\end{equation} is not treated as a finite quantity after cutoff removal. Formula \eqref{expedition0012488} is the norm of the infrared displacement of van Hove type, and under the infrared singularity condition it generally diverges as \(\kappa
\downarrow 0\). Thus this term is not a finite self-interaction term absorbed into the spin-path weight, but an infrared singularity handled as the restriction to the domain \(\dom \mathsf{m}\). On the other hand, what appears in the quadratic term of the spin-path weight is not \(\mathsf{m}_{\kappa}\) but the coupling vector \(\omega\mathsf{m}_{\kappa}\); there, by the estimate in the next proposition, finiteness of \(\omega^{-1} \varrho\) is sufficient.

\begin{prop}[Removal of the infrared cutoff in the spin-path weight]\label{expedition0012482}
For arbitrary $t,s
\in S_{\sminvtemperature}$, define the cutoff-free objects by
\begin{equation}\label{expedition0012483}
\begin{aligned}
W_{\txtspinboson,\sminvtemperature}(t,s)
&=
\fun{\opform{q}_{\txtbsn,\txtnonzero,\sminvtemperature}^{\txteuclid}}
{j_t(\omega\mathsf{m}),j_s(\omega\mathsf{m})},
\\ %%%%%%%%%%%%%%%%
\mathsf{D}_{\txtspinboson,\sminvtemperature}(\prbprocess^{\txtspin})
&=
\fnexp{\oneoverfour
\int_{S_{\sminvtemperature}}
\int_{S_{\sminvtemperature}}
\prbprocess_t^{\txtspin}\prbprocess_s^{\txtspin}
W_{\txtspinboson,\sminvtemperature}(t,s)
\opdmsr{t}\opdmsr{s}},
\\ %%%%%%%%%%%%%%%%
\smpartitionfunc_{\txtspinboson,\sminvtemperature}
&=
\sqfun{\prbexp_{\msrprb_{\txtspin,\sminvtemperature}}}
{\mathsf{D}_{\txtspinboson,\sminvtemperature}},
\\ %%%%%%%%%%%%%%%%
\opdmsr{\msrprb_{\txtspin,\txtspinboson,\sminvtemperature}}
&=
\frac{1}
{\smpartitionfunc_{\txtspinboson,\sminvtemperature}}
\mathsf{D}_{\txtspinboson,\sminvtemperature}
\opdmsr{\msrprb_{\txtspin,\sminvtemperature}}.
\end{aligned}
\end{equation}
Then, as $\kappa \downarrow 0$, the convergence $W_{\txtspinboson,\sminvtemperature,\kappa}(t,s)\to W_{\txtspinboson,\sminvtemperature}(t,s)$ holds uniformly. Moreover, $\mathsf{D}_{\txtspinboson,\sminvtemperature,\kappa}\to\mathsf{D}_{\txtspinboson,\sminvtemperature}$ in $\lp^p(D(S_{\sminvtemperature};\ringratint_2),\msrprb_{\txtspin,\sminvtemperature})$ for arbitrary $1 \leq p < \infty$. In particular, $\smpartitionfunc_{\txtspinboson,\sminvtemperature,\kappa}\to\smpartitionfunc_{\txtspinboson,\sminvtemperature}$, and the normalized densities with respect to the reference measure $\msrprb_{\txtspin,\sminvtemperature}$ satisfy
\begin{equation}\label{expedition0012490}
\frac{1}
{\smpartitionfunc_{\txtspinboson,\sminvtemperature,\kappa}}
\mathsf{D}_{\txtspinboson,\sminvtemperature,\kappa}
\to
\frac{1}
{\smpartitionfunc_{\txtspinboson,\sminvtemperature}}
\mathsf{D}_{\txtspinboson,\sminvtemperature}
\quad\text{in }
\lp^1(D(S_{\sminvtemperature};\ringratint_2),\msrprb_{\txtspin,\sminvtemperature}).
\end{equation}
\end{prop}

\begin{proof}
Using \eqref{expedition0012480} and the covariance $\fun{\opform{q}_{\txtbsn,\txtnonzero,\sminvtemperature}^{\txteuclid}}
{j_t f,j_s g}$, we get
\begin{equation}\label{expedition0012489}
\begin{aligned}
W_{\txtspinboson,\sminvtemperature,\kappa}(t,s)
=
\int_{\fldreal^d}
\frac{\napiernum^{-\abs{t-s}\omega(k)}
+\napiernum^{-(\sminvtemperature-\abs{t-s})\omega(k)}}
{1-\napiernum^{-\sminvtemperature\omega(k)}}
\frac{\abs{\faftr{\varrho}_{\kappa}(k)}^2}{\omega(k)}
\opdmsr{k}.
\end{aligned}
\end{equation}
The infrared singularity appearing in this expression is not $\omega^{-3}$ as in \eqref{expedition0012488}, but only up to $\omega^{-2}$ even with the thermal factor included. Indeed, for $0
< x
\leq 1$, the Taylor expansion or convexity of the exponential gives
$\napiernum^{-x}
\leq
1-x+\frac{x^2}{2}
\leq
1-\frac{x}{2}$.
Hence $1-\napiernum^{-x}
\geq
\frac{x}{2}$, and since $1+\napiernum^{-x}
\leq
2$, we obtain
$$
\frac{1+\napiernum^{-x}}
{1-\napiernum^{-x}}
\leq
\frac{2}{x/2}
=
4x^{-1}.
$$
For $x
\geq 1$, the numerator is at most $2$ and the denominator is at least $1-\napiernum^{-1}$, so the same ratio is bounded. Thus, for a suitable constant $C$ and all $x
>0$,
$$
\frac{1+\napiernum^{-x}}
{1-\napiernum^{-x}}
\leq
C\rbk{x^{-1}+1}.
$$
Therefore the integrand in \eqref{expedition0012489} is bounded by $C_{\sminvtemperature}
\rbk{\frac{\abs{\faftr{\varrho}_{\kappa}(k)}^2}{\omega(k)^2}
+\frac{\abs{\faftr{\varrho}_{\kappa}(k)}^2}{\omega(k)}}$. Hence, if $\varrho_{\kappa}
\to \varrho$ is chosen so that $\omega^{-1} \varrho_{\kappa}
\to \inv{\omega} \varrho$ and $\omega^{-\onehalf}\varrho_{\kappa}
\to \omega^{-\onehalf}\varrho$, then the dominated convergence theorem gives uniform convergence $W_{\txtspinboson,\sminvtemperature,\kappa}
\to W_{\txtspinboson,\sminvtemperature}$. Since the spin path satisfies $\abs{\prbprocess_t^{\txtspin}}
= 1$, the exponent in the double integral converges as a uniformly bounded continuous functional. Therefore the expectation of the exponential converges in every $\lp^p$. In particular, convergence of the partition functions and the $\lp^1$ convergence of the normalized densities in \eqref{expedition0012490} also follow.
\end{proof}

\begin{prop}[Removal of the infrared cutoff for finitely many field correlation functions]\label{expedition0012485}
For arbitrary $n \geq 1$,
$h_1,\ldots,h_n
\in \sphilb{D}_{\txtbsn,\txtphys,\sminvtemperature}$,
$t_1,\ldots,t_n
\in S_{\sminvtemperature}$,
$a_1,\ldots,a_n
\in \fldreal$, the characteristic functional with infrared cutoff converges, as $\kappa
\downarrow 0$, to
\begin{equation}\label{expedition0012486}
\begin{aligned}
&\fun{\oastate[\psi_{\txtspinboson,\sminvtemperature}]}
{\fnexp{\imunit \sum_{k=1}^{n} a_k \opfocksegal(j_{t_k} h_k)}}
\\ %%%%%%%%%%%%%%%%
&=
\fnexp{-\oneoverfour
\sum_{j,k=1}^{n} a_j a_k
\fun{\opform{q}_{\txtbsn,\txtbec,\sminvtemperature}^{\txteuclid}}
{j_{t_j} h_j,j_{t_k} h_k}}
\\
&\quad\times
\sqfun{\prbexp_{\msrprb_{\txtspin,\txtspinboson,\sminvtemperature}}}
{\fnexp{-\imunit \int_{0}^{\sminvtemperature}
\prbprocess_t^{\txtspin}
\rbk{\sum_{k=1}^{n} a_k
\fun{\opform{q}_{\txtbsn,\txtnonzero,\sminvtemperature}^{\txteuclid}}
{j_{t_k}h_k,j_t(\omega\mathsf{m})}}\opdmsr{t}}}.
\end{aligned}
\end{equation}
In particular, the one-variable non-Gaussian factor is
\begin{equation}\label{expedition0012487}
\sqfun{\prbexp_{\msrprb_{\txtspin,\txtspinboson,\sminvtemperature}}}
{\fnexp{-\frac{\imunit}{2}
\int_{0}^{\sminvtemperature}
\prbprocess_u^{\txtspin}
\fun{\opform{q}_{\txtbsn,\txtnonzero,\sminvtemperature}^{\txteuclid}}
{j_t f,j_u(\omega\mathsf{m})}
\opdmsr{u}}}.
\end{equation}
This formula defines the KMS state $\oastate[\psi_{\txtspinboson,\sminvtemperature}]$ after removal of the infrared cutoff on Weyl polynomials over $\sphilb{D}_{\txtbsn,\txtphys,\sminvtemperature}$. The corresponding finitely many sharp-time field correlation functions are obtained as partial derivatives at the origin with respect to the coefficients $a_1,\ldots,a_n$ in \eqref{expedition0012486}.
\end{prop}

\begin{proof}
By \eqref{expedition0012469}, the only terms depending on the infrared cutoff are the spin weight and the cross term containing $\omega\mathsf{m}_{\kappa}$. Proposition \ref{expedition0012482} gives convergence of the spin weight, and since each $h_k
\in \dom \mathsf{m}$, uniformly in $t$ we have
$$\begin{aligned}
\fun{\opform{q}_{\txtbsn,\txtnonzero,\sminvtemperature}^{\txteuclid}}
{j_{t_k}h_k,j_t(\omega\mathsf{m}_{\kappa})}
\to
\fun{\opform{q}_{\txtbsn,\txtnonzero,\sminvtemperature}^{\txteuclid}}
{j_{t_k}h_k,j_t(\omega\mathsf{m})}.
\end{aligned}$$
The absolute value of the phase factor is $1$, and the normalized densities converge in $\lp^1$ by \eqref{expedition0012490}. Thus the dominated convergence theorem gives \eqref{expedition0012486}. The representation of $n$-point functions by partial differentiation is the same standard fact as in \eqref{expedition0012470}.
\end{proof}

This result separates infrared cutoff removal into two effects. First, the spin-path measure including the interaction converges as an ordinary probability measure through \(\omega\mathsf{m}_{\kappa}
\to \omega\mathsf{m}\). Second, the observable directions shrink to \(\sphilb{D}_{\txtbsn,\txtphys,\sminvtemperature}\), where \(\mathsf{m}(f)\) is finite. The latter is the core of the infrared singularity; outside this space, the cross term in \eqref{expedition0012486} has no finite value.

\subsection{Characterization of BEC by Off-Diagonal Long-Range Order}\label{expedition0012572}

Off-diagonal long-range order is a standard characterization of phase transitions, including BEC. Here we compute off-diagonal long-range order for the KMS state after removal of the infrared cutoff, and verify how BEC is characterized.

\begin{prop}[Characterization of BEC by off-diagonal long-range order]\label{expedition0012463}
For notational simplicity, we write the limit $\min(\min_k \abs{x_k}, \min_{i \neq j} \abs{x_i - x_j})
\to \infty$ simply as $\abs{x_i - x_j}
\to \infty$. For the spin-boson KMS state $\oastate[\psi_{\txtspinboson,\sminvtemperature}]$ after removal of the infrared cutoff, $n
\geq 1$,
$g_1, \ldots, g_n
\in \sphilb{D}_{\txtbsn,\txtphys,\sminvtemperature}$,
$x_1, \ldots, x_n
\in \fldreal^{d}$, the long-distance limit of the spatial $n$-point function is
\begin{equation}\label{expedition0012464}
\lim_{\abs{x_i - x_j} \to \infty}
\fun{\oastate[\psi_{\txtspinboson,\sminvtemperature}]}
{\prod_{k=1}^{n} \opfocksegal(j_0 \tau_{x_k} g_k)}
=
\begin{cases}
\sum_{P \in \mathcal{P}_n} \prod_{(i,j) \in P}
\onehalf \fun{\opform{q}_{\txtbsn,0,\sminvtemperature}}{g_i, g_j}
& n \in 2 \monnat \\
0 & n \in 2\monnat + 1
\end{cases}
\end{equation}
Here $\mathcal{P}_n$ denotes the set of all pair partitions of $\setone{1, \ldots, n}$. In particular, vanishing of the two-point off-diagonal long-range order for all $g_1,g_2
\in \sphilb{D}_{\txtbsn,\txtphys,\sminvtemperature}$ is equivalent to the impossibility of BEC, $\opform{q}_{\txtbsn,0,\sminvtemperature}
= 0$.
\end{prop}

\begin{proof}
Apply \eqref{expedition0012486} from Proposition \ref{expedition0012485} with times $t_1
=
\cdots
=
t_n
=
0$,
and test functions $h_k
=
\tau_{x_k} g_k$. Together with the decomposition \eqref{expedition0012475}, the characteristic functional of the $n$-time exponential observable after removal of the infrared cutoff is written as
$$\begin{aligned}
&\fun{\oastate[\psi_{\txtspinboson,\sminvtemperature}]}
{\fnexp{\imunit \sum_{k=1}^{n} a_k \opfocksegal(j_0 \tau_{x_k} g_k)}}
\\ %%%%%%%%%%%%%%%%
&=
\fnexp{-\oneoverfour
\sum_{j,k=1}^{n} a_j a_k
\sqbk{\fun{\opform{q}_{\txtbsn,0,\sminvtemperature}}{\tau_{x_j} g_j, \tau_{x_k} g_k}
+\fun{\opform{q}_{\txtbsn,\txtnonzero,\sminvtemperature}^{\txteuclid}}{j_0 \tau_{x_j} g_j, j_0 \tau_{x_k} g_k}}}
\\ %%%%%%%%%%%%%%%%
&\quad\times
\sqfun{\prbexp_{\msrprb_{\txtspin,\txtspinboson,\sminvtemperature}}}
{\fnexp{-\imunit \sum_{k=1}^{n} a_k \xi_k(\prbprocess^{\txtspin}; x_k)}},
\\ %%%%%%%%%%%%%%%%
&\xi_k(\prbprocess^{\txtspin}; x_k)
=
\int_0^{\sminvtemperature} \prbprocess_t^{\txtspin}
\fun{\opform{q}_{\txtbsn,\txtnonzero,\sminvtemperature}^{\txteuclid}}
{j_0 \tau_{x_k} g_k, j_t(\omega\mathsf{m})} \opdmsr{t}.
\end{aligned}$$

We examine the behavior of each factor in the limit $\abs{x_i - x_j}
\to \infty$.

The zero-mode covariance $\fun{\opform{q}_{\txtbsn,0,\sminvtemperature}}
{\tau_{x_j} g_j, \tau_{x_k} g_k}$ reduces, by the momentum-space representation $\faftr{\tau_x g}(0)
=
\faftr{g}(0)$, to
$$\fun{\opform{q}_{\txtbsn,0,\sminvtemperature}}
{\tau_{x_j} g_j, \tau_{x_k} g_k}
=
\fun{\opform{q}_{\txtbsn,0,\sminvtemperature}}
{g_j, g_k}.$$

The cross terms of the nonzero-mode covariance, namely $\fun{\opform{q}_{\txtbsn,\txtnonzero,\sminvtemperature}^{\txteuclid}}
{j_0 \tau_{x_j} g_j, j_0 \tau_{x_k} g_k}$ for $j
\neq k$, have the kernel representation
$$\begin{aligned}
\bkt{\tau_{x_j} g_j}
{\frac{1 + \napiernum^{-\sminvtemperature\omega}}
{1 - \napiernum^{-\sminvtemperature\omega}} \tau_{x_k} g_k}_{\sphilb{H}_{\txtbsn}}
&=
\bkt{g_j}
{\frac{1 + \napiernum^{-\sminvtemperature\omega}}
{1 - \napiernum^{-\sminvtemperature\omega}} \tau_{x_k - x_j} g_k}_{\sphilb{H}_{\txtbsn}}
\end{aligned}$$
from Proposition \ref{expedition0012099}, and therefore converge to $0$ as $\abs{x_j - x_k}
\to \infty$ by the Riemann-Lebesgue lemma.

The diagonal terms
$\fun{\opform{q}_{\txtbsn,\txtnonzero,\sminvtemperature}^{\txteuclid}}
{j_0 \tau_{x_k} g_k, j_0 \tau_{x_k} g_k}
=
\fun{\opform{q}_{\txtbsn,\txtnonzero,\sminvtemperature}^{\txteuclid}}
{j_0 g_k, j_0 g_k}$
do not depend on $x$ and are retained.

For each $t
\in \closedinterval{0}{\sminvtemperature}$, the spin-path integrand kernel $\fun{\opform{q}_{\txtbsn,\txtnonzero,\sminvtemperature}^{\txteuclid}}
{j_0 \tau_{x_k} g_k, j_t(\omega\mathsf{m})}$ converges to $0$ as $\abs{x_k}
\to \infty$ by the Riemann-Lebesgue lemma. Moreover, since $\abs{\prbprocess_t^{\txtspin}}
=
1$ and the imaginary-time interval is bounded, dominated convergence gives $\xi_k(\prbprocess^{\txtspin}; x_k)
\to 0$. By Parseval's identity, $\abs{\xi_k}$ is uniformly bounded in $x_k$. Since the absolute value of the phase factor is bounded by $1$, another application of dominated convergence shows that the spin-path factor as a whole converges to $1$.

It follows that the limit of the characteristic functional is
$$\begin{aligned}
\fnexp{-\oneoverfour \sum_{j,k=1}^{n} a_j a_k C_{jk}},
\quad
C_{jk}
=
\fun{\opform{q}_{\txtbsn,0,\sminvtemperature}}{g_j, g_k}
+ \delta_{jk}
\fun{\opform{q}_{\txtbsn,\txtnonzero,\sminvtemperature}^{\txteuclid}}{j_0 g_k, j_0 g_k}.
\end{aligned}$$
By the derivative evaluation in \eqref{expedition0012470} and the evaluation of the Gaussian characteristic functional, each pair $(i,j)$ in a pair partition $P
\in \mathcal{P}_n$ satisfies $i
\neq j$, hence $\delta_{ij}
=
0$ and $C_{ij}
=
\fun{\opform{q}_{\txtbsn,0,\sminvtemperature}}
{g_i, g_j}$. Therefore \eqref{expedition0012464} follows.

To check vanishing of off-diagonal long-range order, it is enough to examine the two-point function. By the preceding argument, this is equivalent to $\opform{q}_{\txtbsn,0,\sminvtemperature}
= 0$, giving the equivalence with the impossibility of BEC.
\end{proof}

\subsection{Characterization of BEC by Order Parameters}\label{expedition0012573}

In Subsection \ref{expedition0012572}, the zero-mode covariance \(\opform{q}
_{\txtbsn,0,\sminvtemperature}\) was extracted through off-diagonal long-range order. In the free Bose gas studied in \cite{YoshitsuguSekine004}, the same zero mode can also be extracted from order parameters approximating the zero-momentum mode in finite volume. The same computation holds for the spin-boson model if one returns to the finite-volume functional integral representation. The reason is that the spin-path factor appears only through the nonzero-mode form \(\opform{q}
_{\txtbsn,\txtnonzero,\sminvtemperature,\smchemicalpotential,L}^{\txteuclid}\) and does not cross with the finite-volume zero-momentum mode.

\begin{defn}[Order parameter]
Let $I_L^d$ be the cube centered at the origin with side length $L$, and let $V = L^d$. As in the construction of order parameters in \cite{YoshitsuguSekine004}, define the family of functions
$$\mathsf{b}_L^{(0)}
=
\frac{1}{V^{1/2}}\fndef{I_L^d},
\quad
\mathsf{b}_L^{(1)}
=
\frac{1}{V}\fndef{I_L^d}.$$
For the bounded-system functional integral state $\psi
_{\txtspinboson,\sminvtemperature,\smchemicalpotential,L}$, define the order parameter, in accordance with the sign convention $\oaresolvent(1,0)
= -\imunit$, by
$$\mathsf{o}_{\txtspinboson,\sminvtemperature,L}^{(\#)}
=
\imunit
\fun{\psi_{\txtspinboson,\sminvtemperature,\smchemicalpotential,L}}
{\fun{\oaresolvent}{1,\mathsf{b}_L^{(\#)}}}.$$
\end{defn}

As discussed in \cite{AsaoArai28,YoshitsuguSekine004}, choose the chemical potential by the fixed-density condition and set \(y_L
= \napiernum^{-\sminvtemperature\smchemicalpotential_L}\) and \(N_0(y_L)
= \frac{1}{y_L-1}\). Then the fixed-density thermodynamic limit of the free Bose gas gives \[\smnumberdensity_{\txtbsn,0}(\sminvtemperature)
=
\lim_{L\to\infty}
\frac{N_0(y_L)}{V}.\]

\begin{prop}[Characterization of BEC by order parameters]\label{expedition0012575}
Consider the free-Bose-gas-side decomposition by the zero-mode form defined in Subsection \ref{expedition0011278} and the finite-volume nonzero-mode form defined in \eqref{expedition0012441}. Then
$$\begin{aligned}
\mathsf{o}_{\txtspinboson,\sminvtemperature,L}^{(0)}
&=
\int_{0}^{\infty}
\fnexp{-r}
\fnexp{-\frac{r^2}{4}\frac{y_L+1}{y_L-1}}
\opdmsr{r},
\\ %%%%%%%%%%%%%%%%
\mathsf{o}_{\txtspinboson,\sminvtemperature,L}^{(1)}
&=
\int_{0}^{\infty}
\fnexp{-r}
\fnexp{-\frac{r^2}{4}(y_L+1)\frac{N_0(y_L)}{V}}
\opdmsr{r}
\end{aligned}$$
holds. Moreover, if $\smnumberdensity_{\txtbsn,0}(\sminvtemperature)
> 0$, then
$$\lim_{L\to\infty}\mathsf{o}_{\txtspinboson,\sminvtemperature,L}^{(0)}
=
0,\quad
\lim_{L\to\infty}\mathsf{o}_{\txtspinboson,\sminvtemperature,L}^{(1)}
=
\int_0^\infty \fnexp{-r-\frac{1}{2}\smnumberdensity_{\txtbsn,0}(\sminvtemperature)r^2}\opdmsr{r}
<
1$$
holds, and if $\smnumberdensity_{\txtbsn,0}(\sminvtemperature)
= 0$, then
$$\lim_{L\to\infty}
\mathsf{o}_{\txtspinboson,\sminvtemperature,L}^{(1)}
=
1$$
holds.
\end{prop}

\begin{proof}
In the finite system, $\mathsf{b}_L^{(0)}$ and $\mathsf{b}_L^{(1)}$ have only the zero-momentum mode, while $\opform{q}_{\txtbsn,\txtnonzero,\sminvtemperature,\smchemicalpotential,L}^{\txteuclid}$ is a nonzero-mode form. In particular, for $\#
= 0,1$,
\begin{equation}\label{expedition0012576}
\fun{\opform{q}_{\txtbsn,\txtnonzero,\sminvtemperature,\smchemicalpotential,L}^{\txteuclid}}
{j_0\mathsf{b}_L^{(\#)},j_t(\omega\mathsf{m})}
=
0
\end{equation}
holds.

The Laplace-transform representation and Proposition \ref{expedition0012086} give
$$\mathsf{o}_{\txtspinboson,\sminvtemperature,L}^{(\#)}
=
\int_{0}^{\infty}
\fnexp{-r}
\fun{\psi_{\txtspinboson,\sminvtemperature,\smchemicalpotential,L}}
{\napiernum^{\imunit\opfocksegal(rj_0\mathsf{b}_L^{(\#)})}}
\opdmsr{r}.$$
By \eqref{expedition0012576},
$$\sqfun{\prbexp_{\msrprb_{\txtspin,\txtspinboson,\sminvtemperature}}}
{\fnexp{-\imunit \mathsf{Y}_{\sminvtemperature,r\mathsf{b}_L^{(\#)}}}}
= 1$$
holds. Moreover, by the finite-volume characteristic functional of the free Bose field and \eqref{expedition0012441}, for $\mathsf{b}_L^{(\#)}$, which has only the zero-momentum mode, the above expectation reduces only to the zero-mode evaluation on the free-Bose-gas side.

Since only the zero-momentum mode contributes to $\mathsf{b}_L^{(0)}$ and $\mathsf{b}_L^{(1)}$, the preceding zero-mode evaluations on the free-Bose-gas side are $\frac{y_L+1}{y_L-1}$ and $(y_L+1)
\frac{N_0(y_L)}{V}$. Since $y_L
= 1 + \frac{1}{N_0(y_L)}$,
$$\frac{y_L+1}{y_L-1}
= 2N_0(y_L)+1,
\quad
(y_L+1)\frac{N_0(y_L)}{V}
= 2 \frac{N_0(y_L)}{V}+\frac{1}{V}.$$
Thus, if $\smnumberdensity_{\txtbsn,0}(\sminvtemperature)
> 0$, the first expression diverges to $\infty$, and the second expression converges to $2 \smnumberdensity_{\txtbsn,0}(\sminvtemperature)$. Since the integral kernel is always dominated by $\fnexp{-r}$, the dominated convergence theorem gives $\mathsf{o}_{\txtspinboson,\sminvtemperature,L}^{(0)}
\to 0$ and
$$\mathsf{o}_{\txtspinboson,\sminvtemperature,L}^{(1)}
\to
\int_0^\infty \fnexp{-r-\frac{1}{2}\smnumberdensity_{\txtbsn,0}(\sminvtemperature)r^2}\opdmsr{r}.$$
The integrand on the right-hand side is strictly smaller than $\fnexp{-r}$ for $r>0$, so the limit is smaller than $1$. If $\smnumberdensity_{\txtbsn,0}(\sminvtemperature)
= 0$, then the second expression converges to $0$, and the dominated convergence theorem gives $\mathsf{o}_{\txtspinboson,\sminvtemperature,L}^{(1)}
\to 1$.
\end{proof}

\subsection{Physical Field Operators}\label{physical-field-operators}

\begin{prop}\label{expedition0012118}
For every $f
\in \sphilb{D}_{\txtbsn,\txtphys,\sminvtemperature}$ and $s
\in \fldreal$,
\begin{equation}\label{expedition0012492}
\begin{aligned}
\fun{\psi_{\txtspinboson,\sminvtemperature}}
{\napiernum^{\imunit s \opfocksegal(f)}}
&=
\fnexp{-\frac{s^2}{4}\fun{\opform{q}_{\txtbsn,\txtbec,\sminvtemperature}}{f}}
\sqfun{\prbexp_{\msrprb_{\txtspin,\txtspinboson,\sminvtemperature}}}
{\fnexp{-\imunit s \mathsf{Y}_{\sminvtemperature,f}}},
\\ %%%%%%%%%%%%%%%%
\mathsf{Y}_{\sminvtemperature,f}
&=
\onehalf
\int_{S_{\sminvtemperature}}
\prbprocess_u^{\txtspin}
\fun{\opform{q}_{\txtbsn,\txtnonzero,\sminvtemperature}^{\txteuclid}}
{j_0 f,j_u(\omega \mathsf{m})}
\opdmsr{u}
\end{aligned}
\end{equation}
holds.
\end{prop}

\begin{proof}
It is enough to set $n=1$,
$t_1=0$,
$h_1=f$,
$a_1=s$ in \eqref{expedition0012486} of Proposition \ref{expedition0012485}
and substitute \eqref{expedition0012487}.
\end{proof}

\begin{thm}[Centered physical field]\label{expedition0012136}
For every $f \in \sphilb{D}_{\txtbsn,\txtphys,\sminvtemperature}$, define
\begin{equation}\label{expedition0012493}
\ell_{\txtspin,\sminvtemperature}(f)
=
\sqfun{\prbexp_{\msrprb_{\txtspin,\txtspinboson,\sminvtemperature}}}
{\mathsf{Y}_{\sminvtemperature,f}}
\end{equation}
and define the centered spin random variable by $\widetilde{\mathsf{Y}}_{\sminvtemperature,f}
= \mathsf{Y}_{\sminvtemperature,f}
-\sqfun{\prbexp_{\msrprb_{\txtspin,\txtspinboson,\sminvtemperature}}}
{\mathsf{Y}_{\sminvtemperature,f}}$.
Then $\opfocksegal_{\txtphys,\sminvtemperature}(f)
=
\opfocksegal(f)
+\ell_{\txtspin,\sminvtemperature}(f)$ satisfies $\fun{\psi_{\txtspinboson,\sminvtemperature}}
{\opfocksegal_{\txtphys,\sminvtemperature}(f)}
=
0$.
Moreover, for every $s
\in \fldreal$,
\begin{equation}\label{expedition0012451}
\begin{aligned}
&\fun{\psi_{\txtspinboson,\sminvtemperature}}
{\napiernum^{\imunit s\opfocksegal_{\txtphys,\sminvtemperature}(f)}}
\\ %%%%%%%%%%%%%%%%
&=
\fnexp{-\frac{s^2}{4}\fun{\opform{q}_{\txtbsn,0,\sminvtemperature}}{f}}
\fnexp{-\frac{s^2}{4}\fun{\opform{q}_{\txtbsn,\txtnonzero,\sminvtemperature}}{f}}
\sqfun{\prbexp_{\msrprb_{\txtspin,\txtspinboson,\sminvtemperature}}}
{\fnexp{-\imunit s\widetilde{\mathsf{Y}}_{\sminvtemperature,f}}}
\end{aligned}
\end{equation}
holds.
In particular, the spin characteristic function does not disappear in general,
and it degenerates to $1$ only when the random variable
$\mathsf{Y}_{\sminvtemperature,f}$ is almost surely constant.
\end{thm}

\begin{proof}
Differentiate \eqref{expedition0012492} with respect to $s$ and set $s=0$.
The first derivative of the Gaussian factor vanishes.
The left-hand side gives
$\imunit\fun{\psi_{\txtspinboson,\sminvtemperature}}{\opfocksegal(f)}$,
and the right-hand side gives
$-\imunit\sqfun{\prbexp_{\msrprb_{\txtspin,\txtspinboson,\sminvtemperature}}}
{\mathsf{Y}_{\sminvtemperature,f}}$.
Hence $\fun{\psi_{\txtspinboson,\sminvtemperature}}{\opfocksegal(f)}
=-\ell_{\txtspin,\sminvtemperature}(f)$, and the definition gives
$\fun{\psi_{\txtspinboson,\sminvtemperature}}
{\opfocksegal_{\txtphys,\sminvtemperature}(f)}
=0$.
The centered spin random variable $\widetilde{\mathsf{Y}}_{\sminvtemperature,f}$ has mean $0$.
By \eqref{expedition0012493}, the phase of the centering constant cancels the mean component
in the spin characteristic function in \eqref{expedition0012492}.
Substituting the decomposition
$\opform{q}_{\txtbsn,\txtbec,\sminvtemperature}
= \opform{q}_{\txtbsn,0,\sminvtemperature}
+\opform{q}_{\txtbsn,\txtnonzero,\sminvtemperature}$ from Subsection \ref{expedition0011278}
gives \eqref{expedition0012451}.
The final assertion follows from the uniqueness of characteristic functions.
\end{proof}

\begin{rem}[Degeneracy condition for the spin characteristic function]
The final degeneracy condition in Theorem \ref{expedition0012136} is a condition for the fixed direction $f$.
By the second line of \eqref{expedition0012492},
$\mathsf{Y}_{\sminvtemperature,f}$ is a real random variable obtained by multiplying the spin path
$\prbprocess^{\txtspin}$ by the time-dependent kernel in that formula and integrating.
The condition that the observation direction $f$ does not detect the fluctuations of the interacting
spin-path measure is precisely constancy with respect to
$\msrprb_{\txtspin,\txtspinboson,\sminvtemperature}$.
If this condition holds in every direction of a sufficiently large test-function space that separates
the zero mode, and if the interacting spin-path measure has nontrivial path fluctuations, then one is
restricted to the degenerate situation in which the above kernel vanishes on that space.
In that case the corresponding part of the Bose field appears as a deterministic shift of free-Bose-gas
type or van-Hove-model type, and the fluctuations produced by a general spin-boson interaction do not
remain in the spin characteristic function in \eqref{expedition0012451}.
In this sense the degeneracy condition is a diagnostic condition measuring degeneration to a comparison
model, and it is not a condition expected in a general interacting model.
In particular, the spin-boson model should be interpreted as not being a model constructed according to
the design criterion of \cite{VIYukalov001}.
\end{rem}

\begin{prop}[Infrared correction kernel and the zero-temperature limit]
\label{expedition0012131}
For every $f \in \sphilb{D}_{\txtbsn,\txtphys,\sminvtemperature}$ and
$0 \leq \abs{s} \leq \sminvtemperature$,
\begin{equation}\label{expedition0012494}
\fun{\opform{q}_{\txtbsn,\txtnonzero,\sminvtemperature}^{\txteuclid}}
{j_0 f,j_s(\omega \mathsf{m})}
=
\bkt{f}
{\frac{\napiernum^{-\abs{s}\omega}
+ \napiernum^{-(\sminvtemperature-\abs{s})\omega}}
{1-\napiernum^{-\sminvtemperature\omega}}
\omega \mathsf{m}}_{\sphilb{H}_{\txtbsn}}.
\end{equation}
For fixed $s \in \fldreal$,
\begin{equation}
\lim_{\sminvtemperature \to \infty}
\fun{\opform{q}_{\txtbsn,\txtnonzero,\sminvtemperature}^{\txteuclid}}
{j_0 f,j_s(\omega \mathsf{m})}
=
\bkt{f}{\napiernum^{-\abs{s}\omega}\omega \mathsf{m}}_{\sphilb{H}_{\txtbsn}}.
\end{equation}
\end{prop}

\begin{proof}
Substituting $t=0$ and $g=\omega\mathsf{m}$ into the Euclidean covariance formula for the nonzero modes
gives \eqref{expedition0012494}. The zero-temperature limit follows from the strong convergence
$\napiernum^{-(\sminvtemperature-\abs{s})\omega}\to 0$.
\end{proof}

The source in \eqref{expedition0012494} is \(\omega \mathsf{m}=\omega^{-\onehalf}\varrho\). It is the finite slope appearing in the spin-path measure, not the infrared displacement \(\mathsf{m}=\omega^{-3/2}\varrho\) of the van Hove model itself. After time integration, however, the zero-frequency component of the heat kernel contributes one power of \(\inv{\omega}\), and a linear term of the same order as \(\mathsf{m}(f)\) appears in the formal limit where the spin path degenerates to a constant. In the general spin-boson model this linear term is not constant; it remains as the probability distribution of \(\mathsf{Y}_{\sminvtemperature,f}\).

\subsection{Criterion for Self-Consistent Quasiparticle Equilibrium States and the Zero Mode}\label{criterion-for-self-consistent-quasiparticle-equilibrium-states-and-the-zero-mode}

For nonconserved quasiparticle fields, centering the first moment alone is not sufficient. One must also decide whether a macroscopic zero-energy component remains in the centered field. In the spin-boson model this criterion is expressed by the zero-mode quasi-bilinear form \(\opform{q}_{\txtbsn,0,\sminvtemperature}\) and the spin characteristic function in \eqref{expedition0012492}.

\begin{defn}[Self-consistent quasiparticle equilibrium state]
\label{expedition0012137}
Under an equilibrium state of the centered nonconserved quasiparticle field, let
$\smnumberdensity_{\mathrm{qp},0}$ be the macroscopic occupation density of the zero mode. A state
satisfying $\smnumberdensity_{\mathrm{qp},0}=0$ is called a self-consistent quasiparticle equilibrium
state.
\end{defn}

\begin{rem}[Transfer to the spin-boson model]
\label{expedition0012421}
In this paper the centered field in Theorem \ref{expedition0012136} is taken as the candidate
quasiparticle field, and the zero-mode macroscopic occupation density is measured by
$\smnumberdensity_{\txtbsn,0}(\sminvtemperature)$ from Subsection \ref{expedition0011278}. Thus the
verification of a self-consistent quasiparticle equilibrium state reduces to triviality of the zero-mode
Gaussian factor remaining in the centered characteristic functional. Centering kills only the one-point
function and does not derive the BEC disappearance condition
$\smnumberdensity_{\txtbsn,0}(\sminvtemperature)=0$. The order parameter
$\mathsf{o}_{\txtspinboson,\sminvtemperature,L}^{(1)}$ in Proposition \ref{expedition0012575} detects the
same zero-mode density and can therefore be used for this criterion as well.
\end{rem}

\begin{defn}[Test-function space separating the zero mode]
\label{expedition0012134}
A linear subspace
$\sphilb{E}\subset \sphilb{D}_{\txtbsn,\txtphys,\sminvtemperature}$ is said to separate the zero mode if
there exists $f_*\in\sphilb{E}$ satisfying $\faftr{f_*}(0)\neq 0$.
\end{defn}

\begin{prop}[Zero-mode verification of self-consistent quasiparticle equilibrium states]
\label{expedition0012135}
Assume that
$\sphilb{E}\subset \sphilb{D}_{\txtbsn,\txtphys,\sminvtemperature}$ separates the zero mode. Then the
following conditions are equivalent.
\begin{enumerate}
\item
For the KMS state $\psi_{\txtspinboson,\sminvtemperature}$, the centered field in Theorem
\ref{expedition0012136} defines a self-consistent quasiparticle equilibrium state in the sense of
Definition \ref{expedition0012137}.
\item
For every $f\in\sphilb{E}$, $\fun{\opform{q}_{\txtbsn,0,\sminvtemperature}}{f}=0$.
\item
$\smnumberdensity_{\txtbsn,0}(\sminvtemperature)=0$.
\item
$\lim_{L\to\infty}\mathsf{o}_{\txtspinboson,\sminvtemperature,L}^{(1)}=1$.
\end{enumerate}
\end{prop}

\begin{proof}
Remark \ref{expedition0012421} gives the equivalence of (1) and (3). By the definition in Subsection
\ref{expedition0011278},
$$\fun{\opform{q}_{\txtbsn,0,\sminvtemperature}}{f}
=2(2\pi)^d\smnumberdensity_{\txtbsn,0}(\sminvtemperature)\abs{\faftr{f}(0)}^2.$$
Since $\sphilb{E}$ separates the zero mode, (2) is equivalent to (3). The nonnegativity of the condensate
density and Proposition \ref{expedition0012575} give the equivalence of (3) and (4).
\end{proof}

\subsection{Absence of Off-Diagonal Long-Range Order, Order Parameters, and the Impossibility of BEC}\label{absence-of-off-diagonal-long-range-order-order-parameters-and-the-impossibility-of-bec}

\begin{prop}[Off-diagonal long-range order after removal of the infrared cutoff]
\label{expedition0012500}
For the KMS state $\oastate[\psi_{\txtspinboson,\sminvtemperature}]$ after removal of the infrared cutoff
constructed in Proposition \ref{expedition0012485}, and for
$f,g\in\sphilb{D}_{\txtbsn,\txtphys,\sminvtemperature}$,
\begin{equation}\label{expedition0012501}
\lim_{\abs{x} \to \infty}
\fun{\oastate[\psi_{\txtspinboson,\sminvtemperature}]}
{\opfocksegal(j_0 f)\opfocksegal(j_0 \tau_x g)}
=
\onehalf \fun{\opform{q}_{\txtbsn,0,\sminvtemperature}}{f,g}.
\end{equation}
In particular, absence of off-diagonal long-range order for all $f,g$ is equivalent to
$\opform{q}_{\txtbsn,0,\sminvtemperature}=0$.
\end{prop}

\begin{proof}
Apply \eqref{expedition0012486} of Proposition \ref{expedition0012485} with
$n=2$, $t_1=t_2=0$, $h_1=f$, and $h_2=\tau_xg$. The zero-mode covariance is independent of $x$ because
$\faftr{\tau_x g}(0)=\faftr{g}(0)$. The nonzero-mode covariance and the spin-path kernel with
$\tau_xg$ converge to $0$ as $\abs{x}\to\infty$ by the Riemann--Lebesgue lemma and the estimates of
Proposition \ref{expedition0012482}. Differentiating the characteristic functional twice at the origin
then gives \eqref{expedition0012501}. The final equivalence follows directly from the definition of
$\opform{q}_{\txtbsn,0,\sminvtemperature}$.
\end{proof}

\begin{cor}[Absence of off-diagonal long-range order, the order-parameter criterion, and impossibility of BEC]
\label{expedition0012509}
Assume that $\sphilb{D}_{\txtbsn,\txtphys,\sminvtemperature}$ separates the zero mode. For the KMS state
of Proposition \ref{expedition0012485}, the following conditions are equivalent: absence of
off-diagonal long-range order on the physical test-function space,
$\opform{q}_{\txtbsn,0,\sminvtemperature}=0$,
$\smnumberdensity_{\txtbsn,0}(\sminvtemperature)=0$, self-consistency of the centered quasiparticle
equilibrium state, and
$\lim_{L\to\infty}\mathsf{o}_{\txtspinboson,\sminvtemperature,L}^{(1)}=1$.
\end{cor}

\begin{proof}
This is the combination of Proposition \ref{expedition0012500}, Proposition \ref{expedition0012135}, and
Proposition \ref{expedition0012575}.
\end{proof}

\subsection{Absence of Off-Diagonal Long-Range Order, the Order-Parameter Criterion, and Zero-Mode Vanishing}\label{expedition0012420}

\begin{defn}[Zero-mode vanishing condition]
\label{expedition0012100}
The KMS state is said to satisfy the zero-mode vanishing condition on
$\sphilb{E}\subset\sphilb{D}_{\txtbsn,\txtphys,\sminvtemperature}$ if
\begin{equation}\label{expedition0012495}
\fun{\opform{q}_{\txtbsn,0,\sminvtemperature}}{f}=0
\qquad
(f\in\sphilb{E})
\end{equation}
holds.
\end{defn}

\begin{thm}[Equivalent conditions for zero-mode vanishing]
\label{expedition0012132}
Assume that $\sphilb{D}_{\txtbsn,\txtphys,\sminvtemperature}$ separates the zero mode. For the KMS state
constructed in Proposition \ref{expedition0012485}, the following conditions are equivalent:
\begin{enumerate}
\item the zero-mode vanishing condition holds on
$\sphilb{D}_{\txtbsn,\txtphys,\sminvtemperature}$;
\item for every $f,g\in\sphilb{D}_{\txtbsn,\txtphys,\sminvtemperature}$,
\begin{equation}\label{expedition0012496}
\lim_{\abs{x}\to\infty}
\fun{\oastate[\psi_{\txtspinboson,\sminvtemperature}]}
{\opfocksegal(j_0f)\opfocksegal(j_0\tau_xg)}
=0;
\end{equation}
\item $\smnumberdensity_{\txtbsn,0}(\sminvtemperature)=0$;
\item the centered field of Theorem \ref{expedition0012136} defines a self-consistent quasiparticle
equilibrium state;
\item $\lim_{L\to\infty}\mathsf{o}_{\txtspinboson,\sminvtemperature,L}^{(1)}=1$.
\end{enumerate}
\end{thm}

\begin{proof}
The equivalence of (1) and (2) follows from Proposition \ref{expedition0012500}. The equivalence of
(1), (3), and (4) follows from Proposition \ref{expedition0012135}. The equivalence of (3) and (5)
follows from nonnegativity of the condensate density and Proposition \ref{expedition0012575}.
\end{proof}

\begin{rem}[Relation to Yukalov-type criteria]
The self-consistent quasiparticle equilibrium-state condition for the centered field corresponds to the
Hamiltonian design criterion in \cite{VIYukalov001}; however, we find that the spin-boson model is not a
model that follows this criterion.
\end{rem}

\begin{cor}[Condensate component under the zero-mode vanishing condition]
\label{expedition0012133}
If the equivalent conditions of Theorem \ref{expedition0012132} hold, then for every
$f\in\sphilb{D}_{\txtbsn,\txtphys,\sminvtemperature}$,
\begin{equation}
\fnexp{-\frac{s^2}{4}\fun{\opform{q}_{\txtbsn,0,\sminvtemperature}}{f}}=1
\qquad (s\in\fldreal).
\end{equation}
Thus the BEC zero-mode Gaussian factor in \eqref{expedition0012451} becomes trivial.
\end{cor}

\begin{proof}
Apply Definition \ref{expedition0012100} with
$\sphilb{E}=\sphilb{D}_{\txtbsn,\txtphys,\sminvtemperature}$.
\end{proof}

\subsection{Comparison with the Free Bose Gas and the van Hove Model}\label{comparison-with-the-free-bose-gas-and-the-van-hove-model}

\begin{prop}[Criterion for free-Bose-gas type $c$-number substitution]
\label{expedition0012119}
For a linear subspace
$\sphilb{E}\subset\sphilb{D}_{\txtbsn,\txtphys,\sminvtemperature}$, the following are equivalent:
\begin{enumerate}
\item the characteristic function has a deterministic $c$-number phase on $\sphilb{E}$;
\item for every $f\in\sphilb{E}$, $\mathsf{Y}_{\sminvtemperature,f}$ is almost surely constant.
\end{enumerate}
In this case
$\ell_{\sminvtemperature}(f)=-\sqfun{\prbexp_{\msrprb_{\txtspin,\txtspinboson,\sminvtemperature}}}
{\mathsf{Y}_{\sminvtemperature,f}}$.
\end{prop}

\begin{proof}
This follows by comparing Proposition \ref{expedition0012118} with a deterministic-phase characteristic
function and using uniqueness of characteristic functions.
\end{proof}

\begin{prop}[Vanishing condition for fluctuations of the additional factor]
\label{expedition0012122}
For every $f\in\sphilb{D}_{\txtbsn,\txtphys,\sminvtemperature}$, constancy of
$\mathsf{Y}_{\sminvtemperature,f}$ is equivalent to
$\mathrm{Var}_{\msrprb_{\txtspin,\txtspinboson,\sminvtemperature}}
(\mathsf{Y}_{\sminvtemperature,f})=0$.
\end{prop}

\begin{proof}
The random variable is bounded and real, so constancy and vanishing variance are equivalent.
\end{proof}

\begin{prop}[Deviation estimate by two-point correlations]
\label{expedition0012123}
For every $f\in\sphilb{D}_{\txtbsn,\txtphys,\sminvtemperature}$ and $s\in\fldreal$,
\begin{equation}
\abs{\sqfun{\prbexp_{\msrprb_{\txtspin,\txtspinboson,\sminvtemperature}}}
{\fnexp{-\imunit s\mathsf{Y}_{\sminvtemperature,f}}}
-\fnexp{-\imunit s
\sqfun{\prbexp_{\msrprb_{\txtspin,\txtspinboson,\sminvtemperature}}}
{\mathsf{Y}_{\sminvtemperature,f}}}}
\leq
\frac{s^{2}}{2}
\mathrm{Var}_{\msrprb_{\txtspin,\txtspinboson,\sminvtemperature}}
(\mathsf{Y}_{\sminvtemperature,f}).
\end{equation}
\end{prop}

\begin{proof}
Use the centered variable from Theorem \ref{expedition0012136} and integrate
$\abs{\napiernum^{-\imunit x}-1+\imunit x}\leq x^2/2$.
\end{proof}

\begin{prop}[Comparison with the van Hove model]
\label{expedition0012121}
For every $f\in\sphilb{D}_{\txtbsn,\txtphys,\sminvtemperature}$, equality of the one-variable
characteristic functions of the spin-boson and van Hove models for all $s\in\fldreal$ is equivalent to
the almost sure identity
$\mathsf{Y}_{\sminvtemperature,f}=\opreal\mathsf{m}(f)$.
\end{prop}

\begin{proof}
Compare \eqref{expedition0012492} with the van Hove characteristic function
$\fnexp{-\frac{s^2}{4}\fun{\opform{q}_{\txtbsn,\txtbec,\sminvtemperature}}{f}}
\fnexp{-\imunit s\opreal\mathsf{m}(f)}$ and use uniqueness of characteristic functions.
\end{proof}

\subsection{Preparation for the Resolvent Algebra}\label{preparation-for-the-resolvent-algebra}

\begin{lem}[Laplace-transform representation]
\label{expedition0012102}
In a regular representation of
$\oaresolventalgebra(\sphilb{D}_{\txtbsn,\txtphys,\sminvtemperature},\sigma)$, for every
$\lambda,\mu\in\fldreal\setminus\setone{0}$ and
$f,g\in\sphilb{D}_{\txtbsn,\txtphys,\sminvtemperature}$,
\begin{align}
&\fun{\oastate[\psi_{\txtspinboson,\sminvtemperature}]}
{\oaresolvent(\lambda,f)}
=
-\imunit
\int_{0}^{(\operatorname{sgn}\lambda)\infty}
\fnexp{-\lambda s}
\fun{\psi_{\txtspinboson,\sminvtemperature}}
{\napiernum^{\imunit s \opfocksegal(f)}}
\opdmsr{s}, \label{expedition0012452}
\\
&\fun{\oastate[\psi_{\txtspinboson,\sminvtemperature}]}
{\oaresolvent(\lambda,f)\oaresolvent(\mu,g)}
=
-
\int_{0}^{(\operatorname{sgn}\lambda)\infty}
\int_{0}^{(\operatorname{sgn}\mu)\infty}
\fnexp{-\lambda s-\mu t}
\fun{\psi_{\txtspinboson,\sminvtemperature}}
{\napiernum^{\imunit s \opfocksegal(f)}
\napiernum^{\imunit t \opfocksegal(g)}}
\opdmsr{s}\opdmsr{t}. \label{expedition0012453}
\end{align}
\end{lem}

\begin{proof}
This is the standard Laplace-transform representation of resolvents, evaluated on the GNS vector.
\end{proof}

\begin{prop}
\label{expedition0012103}
For the KMS state of Proposition \ref{expedition0012485}, the resolvent expectations are obtained by
substituting Proposition \ref{expedition0012118} into \eqref{expedition0012452} and
\eqref{expedition0012453}. In particular,
\begin{equation}\label{expedition0012456}
\abs{\fun{\oastate[\psi_{\txtspinboson,\sminvtemperature}]}
{\oaresolvent(\lambda,f)}}
\leq \frac{1}{\abs{\lambda}},
\quad
\abs{\fun{\oastate[\psi_{\txtspinboson,\sminvtemperature}]}
{\oaresolvent(\lambda,f)\oaresolvent(\mu,g)}}
\leq \frac{1}{\abs{\lambda}\abs{\mu}}.
\end{equation}
\end{prop}

\begin{proof}
The estimates follow from
$\abs{\sqfun{\prbexp}{\fnexp{-\imunit\mathsf{Y}_{\sminvtemperature,h}}}}\leq 1$ and positivity of
$\opform{q}_{\txtbsn,\txtbec,\sminvtemperature}$.
\end{proof}

\begin{cor}
\label{expedition0012104}
If
$\fun{\opform{q}_{\txtbsn,\txtbec,\sminvtemperature}}{f_n}\to\infty$, then for every
$\lambda\in\fldreal\setminus\setone{0}$,
$\fun{\oastate[\psi_{\txtspinboson,\sminvtemperature}]}
{\oaresolvent(\lambda,f_n)}\to 0$.
\end{cor}

\begin{proof}
For $s\neq 0$, the Gaussian factor in the Laplace-transform integrand converges to $0$, and the bound
\eqref{expedition0012456} gives dominated convergence.
\end{proof}

\section{Discussion by the Resolvent Algebra}\label{discussion-by-the-resolvent-algebra}

Starting from the KMS state obtained by the functional integral, we organize the ideals of the resolvent algebra and of the algebra of all physical quantities. In particular, we distinguish the directions that are killed in the algebra from the outset because of infrared divergence from the directions excluded by the equivalent conditions in Theorem \ref{expedition0012132}. Unlike the van Hove model \cite{YoshitsuguSekine006}, which used a point source, the properties of the source \(\varrho\) have a strong effect here, and therefore we do not discuss behavior according to the degree \(s\) in the dispersion relation \(\omega_s(k)
= \abs{k}^{s}\).

As seen in the comparison in the preceding section, stochastic fluctuations of the spin factor and deterministic shifts of free-Bose-gas type or van-Hove-model type are diagnostic quantities measuring whether the spin-boson KMS state degenerates to a component state of a free-model type. On the other hand, by Theorem \ref{expedition0012132}, the absence of off-diagonal long-range order and the appropriate behavior of the order parameter are equivalent to the vanishing of the zero-mode quasi-bilinear form. Thus the remaining role of this section is to separate the directions removed from the algebra by infrared divergence from the BEC directions that are recorded auxiliary inside the physical space.

In this section we clearly distinguish three spaces. First, \(\sphilb{D}_{\txtbsn,0,\sminvtemperature}\) is the minimal domain on which the zero-mode quasi-bilinear form \(\opform{q}_{\txtbsn,0,\sminvtemperature}\) has meaning. At this stage we do not assume finiteness of the interaction-induced linear term \(\ell_{\txtspin,\sminvtemperature}\) or of the spin characteristic function in \eqref{expedition0012492}. Next, only on \(\sphilb{D}_{\txtbsn,\txtphys,\sminvtemperature}
=
\sphilb{D}_{\txtbsn,0,\sminvtemperature}\cap \dom \mathsf{m}\) are the resolvent expectations in Propositions \ref{expedition0012102} and \ref{expedition0012103} defined as finite quantities after removal of the infrared cutoff. Therefore \(\sphilb{D}_{\txtbsn,0,\sminvtemperature}
\setminus
\sphilb{D}_{\txtbsn,\txtphys,\sminvtemperature}\) consists of infrared directions excluded not by the zero-mode quasi-bilinear form but by divergence of the linear term that gives the mean. By contrast, BEC directions appear inside \(\sphilb{D}_{\txtbsn,\txtphys,\sminvtemperature}\) as directions on which \(\opform{q}_{\txtbsn,0,\sminvtemperature}\) is still positive. Below we separate how these two types of directions are reflected in ideals of the resolvent algebra.

As algebras of physical quantities, define \[\oa{A}_{\txtspinboson,0,\sminvtemperature}
=
M_2 \otimes \oaresolventalgebra(\sphilb{D}_{\txtbsn,0,\sminvtemperature},\sigma),
\quad
\oa{A}_{\txtspinboson,\txtphys,\sminvtemperature}
=
M_2 \otimes \oaresolventalgebra(\sphilb{D}_{\txtbsn,\txtphys,\sminvtemperature},\sigma).\] By Propositions \ref{expedition0012102} and \ref{expedition0012103}, the resolvent expectations for arbitrary \(\lambda,\mu
\in \fldreal \setminus \setone{0}\) and \(f,g
\in \sphilb{D}_{\txtbsn,\txtphys,\sminvtemperature}\) have already been given explicitly. It is therefore enough to extract the regular representation induced by the KMS state and the associated ideal structure.

\subsection{Infrared Quotient and the Kernel of the KMS State}\label{infrared-quotient-and-the-kernel-of-the-kms-state}

\begin{defn}
Define the infrared ideal for the Bose field by
\begin{equation}
\oaideal{J}_{\txtbsn,\txtirsingular}
=
\clos{
\opideal
\set{\oaresolvent(\lambda,f)}
{\lambda \in \fldreal \setminus \setone{0},\,
f \in \sphilb{D}_{\txtbsn,0,\sminvtemperature}
\setminus
\sphilb{D}_{\txtbsn,\txtphys,\sminvtemperature}}}
\subset
\oaresolventalgebra(\sphilb{D}_{\txtbsn,0,\sminvtemperature},\sigma)
\end{equation}
and set
$\oaideal{J}_{\txtspinboson,\txtirsingular}
=
M_2 \otimes \oaideal{J}_{\txtbsn,\txtirsingular}
\subset
\oa{A}_{\txtspinboson,0,\sminvtemperature}$.
\end{defn}

The directions killed by this definition are \(\sphilb{D}_{\txtbsn,0,\sminvtemperature}
\setminus
\sphilb{D}_{\txtbsn,\txtphys,\sminvtemperature}\), that is, directions for which \(\opform{q}_{\txtbsn,0,\sminvtemperature}\) has meaning but the interaction-induced linear functional \(\mathsf{m}\) is not finite. The one- and two-point resolvent expectations in Propositions \ref{expedition0012102} and \ref{expedition0012103} are defined after removal of the infrared cutoff only for \(f,g
\in
\sphilb{D}_{\txtbsn,\txtphys,\sminvtemperature}\), and hence these infrared directions have already lost physical meaning at the level of fixed generators. Unlike BEC directions, whose expectations vanish only in a large-scale limit as in Proposition \ref{expedition0012104}, these directions must be removed by quotienting out the generators themselves.

\begin{prop}\label{expedition0012125}
By the universality of the resolvent algebra, there exist infrared quotient maps
$$\Theta_{\txtbsn,\txtirsingular}
\colon
\oaresolventalgebra(\sphilb{D}_{\txtbsn,0,\sminvtemperature},\sigma)
\to
\oaresolventalgebra(\sphilb{D}_{\txtbsn,\txtphys,\sminvtemperature},\sigma),
\quad
\Theta_{\txtspinboson,\txtirsingular}
=
\idone_{M_2}\otimes \Theta_{\txtbsn,\txtirsingular}
\colon
\oa{A}_{\txtspinboson,0,\sminvtemperature}
\to
\oa{A}_{\txtspinboson,\txtphys,\sminvtemperature}
.$$
In particular, their kernels are respectively
$\Ker \Theta_{\txtbsn,\txtirsingular}
=
\oaideal{J}_{\txtbsn,\txtirsingular}$ and
$\Ker \Theta_{\txtspinboson,\txtirsingular}
=
\oaideal{J}_{\txtspinboson,\txtirsingular}$, and the isomorphisms
$$\begin{aligned}
\setquot{\oaresolventalgebra(\sphilb{D}_{\txtbsn,0,\sminvtemperature},\sigma)}{\oaideal{J}_{\txtbsn,\txtirsingular}}
\eqalgisom
\oaresolventalgebra(\sphilb{D}_{\txtbsn,\txtphys,\sminvtemperature},\sigma),
\quad
\setquot{\oa{A}_{\txtspinboson,0,\sminvtemperature}}
{\oaideal{J}_{\txtspinboson,\txtirsingular}}
\eqalgisom
\oa{A}_{\txtspinboson,\txtphys,\sminvtemperature}
\end{aligned}$$
hold.
\end{prop}

\begin{proof}
The physical one-particle space
$\sphilb{D}_{\txtbsn,\txtphys,\sminvtemperature}
\subset
\sphilb{D}_{\txtbsn,0,\sminvtemperature}$ is a linear subspace.
Hence, by the universality of the resolvent algebra \cite{DetlevBuchholz001}, one obtains a surjective
$\ast$-homomorphism that restricts the generators over
$\sphilb{D}_{\txtbsn,0,\sminvtemperature}$ to the generators over
$\sphilb{D}_{\txtbsn,\txtphys,\sminvtemperature}$.
Its kernel is the closed two-sided ideal generated by the resolvents corresponding to
$\sphilb{D}_{\txtbsn,0,\sminvtemperature}
\setminus
\sphilb{D}_{\txtbsn,\txtphys,\sminvtemperature}$.
For the algebra of all physical quantities, one takes the tensor product and uses the simplicity of
$M_2$.
\end{proof}

\begin{prop}\label{expedition0012128}
Let $J$ be a closed two-sided ideal of
$\oaresolventalgebra(\sphilb{D}_{\txtbsn,0,\sminvtemperature},\sigma)$ satisfying
$$\begin{aligned}
\oaresolvent(\lambda,f)\in J
\quad
\rbk{\lambda \in \fldreal \setminus \setone{0},\,
f \in \sphilb{D}_{\txtbsn,0,\sminvtemperature}
\setminus
\sphilb{D}_{\txtbsn,\txtphys,\sminvtemperature}}.
\end{aligned}$$
Then $\oaideal{J}_{\txtbsn,\txtirsingular}
\subset J$ holds.
Similarly, if a closed two-sided ideal $\widetilde{J}$ of
$\oa{A}_{\txtspinboson,0,\sminvtemperature}$ contains $M_2
\otimes
\oaresolvent(\lambda,f)$ for all $f$ as above, then
$\oaideal{J}_{\txtspinboson,\txtirsingular}
\subset \widetilde{J}$ holds.
In particular, $\oaideal{J}_{\txtbsn,\txtirsingular}$ and
$\oaideal{J}_{\txtspinboson,\txtirsingular}$ are the minimal closed two-sided ideals that remove the
infrared directions.
\end{prop}

\begin{proof}
By definition, $\oaideal{J}_{\txtbsn,\txtirsingular}$ is the closed two-sided ideal generated by all
generator resolvents belonging to
$\sphilb{D}_{\txtbsn,0,\sminvtemperature}
\setminus
\sphilb{D}_{\txtbsn,\txtphys,\sminvtemperature}$.
Therefore any closed two-sided ideal $J$ containing all of them must contain the generated closed
two-sided ideal $\oaideal{J}_{\txtbsn,\txtirsingular}$.
For the algebra of all physical quantities, the same argument is applied in
$M_2\otimes \oaresolventalgebra(\sphilb{D}_{\txtbsn,0,\sminvtemperature},\sigma)$.
\end{proof}

Define a state \(\oastate[\psi_{\txtspinboson,\sminvtemperature}^{\mathrm{res}}]\) on the operator algebra \(\oa{A}_{\txtspinboson,\txtphys,\sminvtemperature}\) as a state having resolvent expectations of the form in Lemma \ref{expedition0012102} and Proposition \ref{expedition0012103}. Unless confusion is likely, we use the same notation for its restriction to \(\oaresolventalgebra(\sphilb{D}_{\txtbsn,\txtphys,\sminvtemperature},\sigma)\).

\begin{prop}\label{expedition0012126}
The state $\oastate[\psi_{\txtspinboson,\sminvtemperature}^{\mathrm{res}}]$ above is regular.
By Proposition \ref{expedition0011838}, its GNS representation is faithful; in particular,
$\Ker
\oarepn_{\oastate[\psi_{\txtspinboson,\sminvtemperature}^{\mathrm{res}}]}
=
\setone{0}$ holds.
Moreover, if
$\oastate[\psi_{\txtspinboson,\sminvtemperature}^{(0)}]
=
\oastate[\psi_{\txtspinboson,\sminvtemperature}^{\mathrm{res}}]
\circ
\Theta_{\txtspinboson,\txtirsingular}$, then
$\Ker
\oarepn_{\oastate[\psi_{\txtspinboson,\sminvtemperature}^{(0)}]}
=
\oaideal{J}_{\txtspinboson,\txtirsingular}$ holds.
\end{prop}

\begin{proof}
By definition,
$\oastate[\psi_{\txtspinboson,\sminvtemperature}^{\mathrm{res}}]$ is the restriction of the functional
integral state $\oastate[\psi_{\txtspinboson,\sminvtemperature}]$ to
$\oa{A}_{\txtspinboson,\txtphys,\sminvtemperature}$.
Propositions \ref{expedition0012102} and \ref{expedition0012103} give explicit expectations of the
resolvent generators.
The restricted state is a regular state induced by the field operators on
$\sphilb{D}_{\txtbsn,\txtphys,\sminvtemperature}$, and hence Proposition \ref{expedition0011838} implies
that its GNS representation is faithful.
In general, for a surjective $\ast$-homomorphism $\Theta
\colon \oa{A}
\to \oa{B}$ and a faithful state $\psi$ on $\oa{B}$, one has
$\Ker \oarepn_{\psi \circ \Theta}
=
\Ker \Theta$.
Apply this to $\Theta_{\txtspinboson,\txtirsingular}$.
\end{proof}

This proposition shows that, from the viewpoint of the KMS state, the only nontrivial ideal that first appears in the algebra of all physical quantities is \(\oaideal{J}_{\txtspinboson,\txtirsingular}\), which kills the infrared-divergent directions. At the level of the physical quotient algebra \(\oa{A}_{\txtspinboson,\txtphys,\sminvtemperature}\), no GNS kernel of the KMS state remains. Thus the treatment of the infrared directions should be understood not as a state-dependent accident but as a rewriting, by the universality of the resolvent algebra, of the restriction to the physical test-function space \(\sphilb{D}_{\txtbsn,\txtphys,\sminvtemperature}\).

\subsection{Order Parameters and Centers for Abstract KMS States}\label{expedition0012519}

By Proposition \ref{expedition0012126}, the KMS state constructed by the functional integral defines the regular state \(\oastate[\psi_{\txtspinboson,\sminvtemperature}^{\mathrm{res}}]\) on the physical resolvent algebra \(\oa{A}_{\txtspinboson,\txtphys,\sminvtemperature}\). On the other hand, an abstract regular KMS state need not come with the characteristic-functional representation of Proposition \ref{expedition0012485}, and unbounded field operators do not belong to the resolvent algebra. Therefore the BEC criterion for a general KMS state is formulated by representing an order parameter of the \(\mathsf{b}_L^{(1)}\) type from Proposition \ref{expedition0012575} as an element of the resolvent algebra and asking whether its strong limit in the GNS representation produces a center.

\begin{defn}[Resolvent order-parameter net]\label{expedition0012517}
Let $\fml{b_{\nu}}{\nu}$ be a directed family in the physical test-function space
$\sphilb{D}_{\txtbsn,\txtphys,\sminvtemperature}$ satisfying the following conditions.
\begin{enumerate}
\item
The family $\fml{b_{\nu}}{\nu}$ has the same zero-mode normalization as
$\mathsf{b}_L^{(1)}$ in Proposition \ref{expedition0012575}.
In particular, $\faftr{b_{\nu}}(0)$ is a normalization constant independent of $\nu$.

\item
For every $h
\in \sphilb{D}_{\txtbsn,\txtphys,\sminvtemperature}$,
$\sigma(b_{\nu},h)
\to 0$ holds.

\item
For the nonzero modes,
$\fun{\opform{q}_{\txtbsn,\txtnonzero,\sminvtemperature}}{b_{\nu}}
\to 0$ holds.
\end{enumerate}
A family $\fml{b_{\nu}}{\nu}$ satisfying these conditions is called a resolvent order-parameter net.
The nonempty family of all resolvent order-parameter nets is denoted by
$\mathcal{N}_{\mathrm{ov}}^{\mathrm{res}}$.
For each $\fml{b_{\nu}}{\nu}\in\mathcal{N}_{\mathrm{ov}}^{\mathrm{res}}$,
$$\mathsf{o}_{\psi,\nu}=\imunit\fun{\oastate[\psi]}{\idone_{M_2}\otimes\oaresolvent(1,b_{\nu})}$$
is called the resolvent order parameter of $\oastate[\psi]$.
\end{defn}

\begin{ex}[Finite-volume zero-mode type net]\label{expedition0012536}
Consider $\mathsf{b}_L^{(1)}=V^{-1}\fndef{I_L^d}$ from Proposition \ref{expedition0012575}.
For each $L$, $\mathsf{b}_L^{(1)}\in\sphilb{D}_{\txtbsn,\txtphys,\sminvtemperature}$ holds.
If $\sigma(\mathsf{b}_L^{(1)},h)
\to 0$ holds for every $h
\in \sphilb{D}_{\txtbsn,\txtphys,\sminvtemperature}$, then
$\fml{\mathsf{b}_L^{(1)}}{L>0}$ is a resolvent order-parameter net.
In a finite system, $\mathsf{b}_L^{(1)}$ has only the zero-momentum mode, and hence its contribution to
the nonzero-mode covariance vanishes.
For example, in a situation where each element of
$\sphilb{D}_{\txtbsn,\txtphys,\sminvtemperature}$ is represented by an $\lp^1$ function, the second
condition follows from
$$\abs{\sigma(\mathsf{b}_L^{(1)},h)}\leq\frac{1}{V}\onenorm{h}\to 0.$$
\end{ex}

\begin{prop}[Asymptotic centrality of order parameters]\label{expedition0012524}
Let $\fml{b_{\nu}}{\nu}\in\mathcal{N}_{\mathrm{ov}}^{\mathrm{res}}$ be a resolvent order-parameter net.
Then, for every $A
\in \oa{A}_{\txtspinboson,\txtphys,\sminvtemperature}$,
\begin{equation}\label{expedition0012540}
\norm{\commutator{\idone_{M_2}\otimes\oaresolvent(1,b_{\nu})}{A}}\to 0
\end{equation}
holds.
\end{prop}

\begin{proof}
It is enough to prove the assertion for the tensor product
$A
=
A_0
\otimes
\oaresolvent(\lambda,h)$.
The commutator formula in the resolvent relations gives
\begin{equation}\label{expedition0012541}
\norm{\commutator{\idone_{M_2}\otimes\oaresolvent(1,b_{\nu})}{A_0\otimes\oaresolvent(\lambda,h)}}
\leq
\frac{\norm{A_0}\abs{\sigma(b_{\nu},h)}}{\abs{\lambda}^{2}}.
\end{equation}
By Definition \ref{expedition0012517}, the right-hand side converges to $0$.
For finite products, the same conclusion follows from the Leibniz formula for commutators.
Since finite products of resolvent generators are dense and
$\norm{\idone_{M_2} \otimes \oaresolvent(1,b_{\nu})}
\leq 1$, condition (2) in the definition of a resolvent order-parameter net gives
\eqref{expedition0012540}.
\end{proof}

\begin{prop}[Criterion for emergence of a center in the representation space]\label{expedition0012544}
Let $\pairbk{\sphilb{H}_{\psi},\oarepn_{\psi},\oagnsvector_{\psi}}$ be the GNS triple associated with a
state $\oastate[\psi]$ on the algebra of physical quantities
$\oa{A}_{\txtspinboson,\txtphys,\sminvtemperature}$, and let the von Neumann algebra in the representation
space be
$\oa{M}_{\psi}
=
\oadoublecommutant{\oarepn_{\psi}(\oa{A}_{\txtspinboson,\txtphys,\sminvtemperature})}$.
For any $\fml{b_{\nu}}{\nu}
\in \mathcal{N}_{\mathrm{ov}}^{\mathrm{res}}$, if
\begin{equation}\label{expedition0012545}
\slim_{\nu}\oarepn_{\psi}(\idone_{M_2}\otimes\oaresolvent(1,b_{\nu}))
=
Z_{\psi}
\end{equation}
exists, then $Z_{\psi}
\in \oacenter(\oa{M}_{\psi})$ holds.
In particular, if $Z_{\psi}$ is not a scalar operator, then a nontrivial center emerges from the order
parameter in the representation space of $\oastate[\psi]$.
\end{prop}

\begin{proof}
By Proposition \ref{expedition0012524}, for every $A
\in \oa{A}_{\txtspinboson,\txtphys,\sminvtemperature}$,
$\commutator{\oarepn_{\psi}(\idone_{M_2}\otimes\oaresolvent(1,b_{\nu}))}
{\oarepn_{\psi}(A)}
\to 0$ holds in operator norm.
The strong operator convergence in \eqref{expedition0012545} implies
$\commutator{Z_{\psi}}{\oarepn_{\psi}(A)}
= 0$.
Thus $Z_{\psi}$ belongs to
$\oarepn_{\psi}(\oa{A}_{\txtspinboson,\txtphys,\sminvtemperature})'$.
By definition, $Z_{\psi}$ clearly belongs to $\oa{M}_{\psi}$, and hence
$Z_{\psi}
\in
\oa{M}_{\psi}
\cap
\oacommutant{\oa{M}_{\psi}}
= \oacenter(\oa{M}_{\psi})$.
\end{proof}

\begin{defn}[Order-parameter zero-mode compatibility]\label{expedition0012542}
Let $\oastate[\psi]$ be a state on the algebra of physical quantities
$\oa{A}_{\txtspinboson,\txtphys,\sminvtemperature}$.
We say that $\oastate[\psi]$ is order-parameter zero-mode compatible if, for every
$\fml{b_{\nu}}
{\nu}\in\mathcal{N}_{\mathrm{ov}}^{\mathrm{res}}$, the GNS representation satisfies
\begin{equation}\label{expedition0012543}
\slim_{\nu}\oarepn_{\psi}(\idone_{M_2}\otimes\oaresolvent(1,b_{\nu}))
=
-\imunit
\idone_{\sphilb{H}_{\psi}}.
\end{equation}
\end{defn}

\begin{thm}[Order-parameter criterion for general KMS states]\label{expedition0012550}
For a KMS state $\oastate[\psi]$ on the algebra of physical quantities
$\oa{A}_{\txtspinboson,\txtphys,\sminvtemperature}$, the following statements hold.
\begin{enumerate}
\item
For any $\fml{b_{\nu}}{\nu}\in\mathcal{N}_{\mathrm{ov}}^{\mathrm{res}}$, if the strong operator limit
$\slim_{\nu}\oarepn_{\psi}(\idone_{M_2}\otimes\oaresolvent(1,b_{\nu}))$ exists, then that limit belongs
to the center of $\oa{M}_{\psi}$.

\item
If $\oastate[\psi]$ is order-parameter zero-mode compatible, then for every
$\fml{b_{\nu}}
{\nu}\in\mathcal{N}_{\mathrm{ov}}^{\mathrm{res}}$ no nontrivial center emerges from the order parameter.
In this case $\mathsf{o}_{\psi,\nu}
\to 1$ holds.

\item
If, for some $\fml{b_{\nu}}
{\nu}\in\mathcal{N}_{\mathrm{ov}}^{\mathrm{res}}$, there exists
$\slim_{\nu}
\oarepn_{\psi}(\idone_{M_2}\otimes\oaresolvent(1,b_{\nu}))
= Z_{\psi}$ and
$Z_{\psi}\notin\fldcmp\idone_{\sphilb{H}_{\psi}}$, then the state $\oastate[\psi]$ is not
order-parameter zero-mode compatible.
\end{enumerate}
\end{thm}

\begin{proof}
(1): This is Proposition \ref{expedition0012544}.

(2): This follows immediately from Definition \ref{expedition0012542}.
Indeed,
$\mathsf{o}_{\psi,\nu}
=
\imunit
\bkt{\oagnsvector_{\psi}}
{\fun{\oarepn_{\psi}}
{\idone_{M_2} \otimes \oaresolvent(1,b_{\nu})}
\oagnsvector_{\psi}}$, and hence \eqref{expedition0012543} gives
$\mathsf{o}_{\psi,\nu}
\to 1$.

(3): This is the negation of Definition \ref{expedition0012542}.
\end{proof}

\begin{prop}[Strong convergence of order parameters for the functional-integral KMS state]\label{expedition0012553}
Assume $\smnumberdensity_{\txtbsn,0}(\sminvtemperature)
= 0$.
The state $\oastate[\psi_{\txtspinboson,\sminvtemperature}^{\mathrm{res}}]$ of Proposition
\ref{expedition0012126} is order-parameter zero-mode compatible.
\end{prop}

\begin{proof}
Fix a resolvent order-parameter net
$\fml{b_{\nu}}
{\nu}\in\mathcal{N}_{\mathrm{ov}}^{\mathrm{res}}$, and let
$\pairbk{\sphilb{H}_{\psi},\oarepn_{\psi},\oagnsvector_{\psi}}$ be the GNS triple of
$\oastate[\psi_{\txtspinboson,\sminvtemperature}^{\mathrm{res}}]$.

First we prove
$\oarepn_{\psi}(\idone_{M_2}
\otimes
\oaresolvent(1,b_{\nu}))\oagnsvector_{\psi}
\to -\imunit\oagnsvector_{\psi}$.
By Proposition \ref{expedition0012103}, for every $\lambda
\in \fldreal \setminus \setone{0}$,
\begin{equation}\label{expedition0012554}
\begin{aligned}
&\fun{\oastate[\psi_{\txtspinboson,\sminvtemperature}^{\mathrm{res}}]}
{\idone_{M_2}\otimes\oaresolvent(\lambda,b_{\nu})}
\\ %%%%%%%%%%%%%%%%
&=
-\imunit
\int_{0}^{(\operatorname{sgn}\lambda)\infty}
\fnexp{-\lambda s-\oneoverfour s^2\fun{\opform{q}_{\txtbsn,\txtbec,\sminvtemperature}}{b_{\nu}}}
\sqfun{\prbexp_{\msrprb_{\txtspin,\txtspinboson,\sminvtemperature}}}
{\fnexp{-\imunit\mathsf{Y}_{\sminvtemperature,s b_{\nu}}}}
\opdmsr{s}
\end{aligned}
\end{equation}
holds.
The assumption $\smnumberdensity_{\txtbsn,0}(\sminvtemperature)
= 0$ implies
$\opform{q}_{\txtbsn,\txtbec,\sminvtemperature}
= \opform{q}_{\txtbsn,\txtnonzero,\sminvtemperature}$.
By Definition \ref{expedition0012517},
$\fun{\opform{q}_{\txtbsn,\txtbec,\sminvtemperature}}{b_{\nu}}
\to 0$ holds.
Moreover, by \eqref{expedition0012487} and the Cauchy--Schwarz inequality, for each $s$,
$$\begin{aligned}
\abs{\frac{s}{2}\int_{0}^{\sminvtemperature}\prbprocess_u^{\txtspin}\fun{\opform{q}_{\txtbsn,\txtnonzero,\sminvtemperature}^{\txteuclid}}{j_0 b_{\nu},j_u(\omega\mathsf{m})}\opdmsr{u}}
\leq
\frac{\abs{s}}{2}
\sminvtemperature
\fun{\opform{q}_{\txtbsn,\txtnonzero,\sminvtemperature}}{b_{\nu}}^{1/2}
\fun{\opform{q}_{\txtbsn,\txtnonzero,\sminvtemperature}}{\omega\mathsf{m}}^{1/2}
\to 0
\end{aligned}$$
holds, and therefore
$\sqfun{\prbexp_{\msrprb_{\txtspin,\txtspinboson,\sminvtemperature}}}
{\fnexp{-\imunit\mathsf{Y}_{\sminvtemperature,s b_{\nu}}}}
\to 1$.
Furthermore, since
$\abs{\sqfun{\prbexp_{\msrprb_{\txtspin,\txtspinboson,\sminvtemperature}}}
{\fnexp{-\imunit\mathsf{Y}_{\sminvtemperature,s b_{\nu}}}}}
\leq 1$, the dominated convergence theorem applied to \eqref{expedition0012554} gives
\begin{equation}\label{expedition0012558}
\fun{\oastate[\psi_{\txtspinboson,\sminvtemperature}^{\mathrm{res}}]}
{\idone_{M_2}\otimes\oaresolvent(\lambda,b_{\nu})}
\to
-\frac{\imunit}{\lambda}.
\end{equation}
In particular, substituting $\lambda=1$ and $\lambda=-1$ into \eqref{expedition0012558} and applying
Proposition \ref{expedition0012103} to the two-point resolvent gives
\begin{equation}\label{expedition0012551}
\fun{\oastate[\psi_{\txtspinboson,\sminvtemperature}^{\mathrm{res}}]}
{\rbk{\idone_{M_2}\otimes\oaresolvent(-1,b_{\nu})}\rbk{\idone_{M_2}\otimes\oaresolvent(1,b_{\nu})}}
\to
1.
\end{equation}
From \eqref{expedition0012558} and \eqref{expedition0012551},
$$\begin{aligned}
\norm{\rbk{\oarepn_{\psi}(\idone_{M_2}\otimes\oaresolvent(1,b_{\nu}))
+\imunit\idone_{\sphilb{H}_{\psi}}}\oagnsvector_{\psi}}^2
\to
0
\end{aligned}$$
holds, and hence
$\oarepn_{\psi}(\idone_{M_2}
\otimes
\oaresolvent(1,b_{\nu}))\oagnsvector_{\psi}
\to -\imunit\oagnsvector_{\psi}$.

For every $A
\in \oa{A}_{\txtspinboson,\txtphys,\sminvtemperature}$,
$$\begin{aligned}
&\rbk{\oarepn_{\psi}(\idone_{M_2}\otimes\oaresolvent(1,b_{\nu}))
+\imunit
\idone_{\sphilb{H}_{\psi}}}\oarepn_{\psi}(A)\oagnsvector_{\psi}
\\
&=
\fun{\oarepn_{\psi}(A)}
{\oarepn_{\psi}(\idone_{M_2}\otimes\oaresolvent(1,b_{\nu}))+\imunit\idone_{\sphilb{H}_{\psi}}}\oagnsvector_{\psi}
+\commutator{\oarepn_{\psi}(\idone_{M_2}\otimes\oaresolvent(1,b_{\nu}))}{\oarepn_{\psi}(A)}\oagnsvector_{\psi}
\end{aligned}$$
holds.
The first term converges to $0$, and the second term converges to $0$ by Proposition
\ref{expedition0012524}.
By definition, $\oarepn_{\psi}(\oa{A}_{\txtspinboson,\txtphys,\sminvtemperature})
\oagnsvector_{\psi}$ is dense, and
$\norm{\oarepn_{\psi}(\idone_{M_2}\otimes\oaresolvent(1,b_{\nu}))
+\imunit\idone_{\sphilb{H}_{\psi}}}
\leq 2$.
Thus, in the strong operator topology,
$\oarepn_{\psi}(\idone_{M_2}\otimes\oaresolvent(1,b_{\nu}))
\to
-\imunit
\idone_{\sphilb{H}_{\psi}}$.
\end{proof}

\begin{rem}\label{expedition0012521}
Theorem \ref{expedition0012550} shows that the KMS condition alone does not control the behavior of the
order-parameter limit.
The KMS condition concerns time evolution, whereas macroscopic classicalization of the zero-momentum mode
appears separately as a strong limit in the GNS representation.
If this strong limit is nonscalar, a center emerges; if it converges to $-\imunit\idone$ as in
\eqref{expedition0012543}, then the order parameter does not produce a center.
Proposition \ref{expedition0012553} shows that, under
$\smnumberdensity_{\txtbsn,0}(\sminvtemperature)
= 0$, the spin-boson KMS state constructed by the functional integral belongs to the latter case.
\end{rem}

In \cite{ArakiHaagKastlerTakesaki1,BratteliRobinson2}, asymptotic commutativity and weak clustering are used in arguments related to the stability of KMS states. We briefly discuss them in relation to the preceding argument.

\begin{defn}[$\lp^{1}$ asymptotic commutativity and weak clustering]\label{expedition0012581}
Consider a $\oacstar$-dynamical system $\pairbk{\oa{A},\tau}$ with spatial translations
$\fml{\tau_x}{x \in \fldreal^d}$ and a dense $\ast$-subalgebra
$\oa{A}_{\txtloc}
\subset \oa{A}$.
\begin{enumerate}
\item
If
$$\int_{\fldreal^d}\norm{\commutator{A}{\tau_x(B)}}\opdmsr{x}<\infty$$
holds for every $A,B
\in \oa{A}_{\txtloc}$, then $\pairbk{\oa{A},\tau}$ is said to satisfy $\lp^{1}$ asymptotic
commutativity on $\oa{A}_{\txtloc}$.

\item
A spatial-translation invariant state $\oastate[\psi]$ is said to satisfy weak clustering on
$\oa{A}_{\txtloc}$ if, for every $A,B
\in \oa{A}_{\txtloc}$,
$$\lim_{L\to\infty}
\frac{1}{\abs{\Lambda_L}}
\int_{\Lambda_L}
\fun{\oastate[\psi]}{A\tau_x(B)}
\opdmsr{x}
=
\fun{\oastate[\psi]}{A}\fun{\oastate[\psi]}{B}$$
holds.
Here $\Lambda_L$ is a Folner sequence of cubes in $\fldreal^d$.
Weak clustering is a Cesàro-average type condition, not a condition assuming pointwise long-distance
limits.
\end{enumerate}
\end{defn}

\begin{prop}[Asymptotic centrality of order parameters from $\lp^{1}$ commutativity]\label{expedition0012582}
Assume the $\lp^{1}$ asymptotic commutativity of Definition \ref{expedition0012581}.
Fix any $B
\in \oa{A}_{\txtloc}$, and assume that the finite-volume average
$B_L
=
\frac{1}{\abs{\Lambda_L}}
\int_{\Lambda_L}
\tau_x(B)
\opdmsr{x}$ is an element of $\oa{A}$.
Then, for every $A
\in \oa{A}_{\txtloc}$,
$\norm{\commutator{B_L}{A}}
\to 0$ holds.
In particular, if a zero-momentum order parameter on the resolvent algebra is represented by this
finite-volume average, then $\lp^{1}$ asymptotic commutativity is a sufficient condition guaranteeing the
asymptotic centrality of Proposition \ref{expedition0012524}.
\end{prop}

\begin{proof}
The norm estimate for commutators and integrals gives
$\norm{\commutator{B_L}{A}}
\leq
\frac{1}{\abs{\Lambda_L}}
\int_{\Lambda_L}
\norm{\commutator{\tau_x(B)}{A}}
\opdmsr{x}$.
By $\lp^{1}$ asymptotic commutativity, the right-hand side converges to $0$.
By density, the same conclusion holds for the local subalgebra generated by finite products.
In the resolvent algebra, it is enough to apply the commutator formula for generator resolvents as in the
proof of Proposition \ref{expedition0012524}.
\end{proof}

A state \(\oastate[\psi]\) whose weak closure in the GNS representation is a factor is called a primary state.

\begin{prop}[Scalarization of central limits by weak clustering]\label{expedition0012583}
Consider a spatial-translation invariant state $\oastate[\psi]$ satisfying the weak clustering of
Definition \ref{expedition0012581}.
For the finite-volume average $B_L$ obtained from an arbitrarily fixed $B
\in \oa{A}_{\txtloc}$, assume that the strong limit
$\slim_{L\to\infty}
\oarepn_{\psi}(B_L)
=
Z$ exists in the GNS representation and satisfies
$Z
\in \oacenter(\oa{M}_{\psi})$.
Then
$Z \oagnsvector_{\psi}
=
\fun{\oastate[\psi]}{B}
\oagnsvector_{\psi}$ holds.
If $\oastate[\psi]$ is a primary state, then
$Z
=
\fun{\oastate[\psi]}{B}
\idone_{\sphilb{H}_{\psi}}$ holds.
\end{prop}

\begin{proof}
For every $A
\in \oa{A}_{\txtloc}$, weak clustering gives
$$\bkt{\oarepn_{\psi}(A^*)\oagnsvector_{\psi}}
{Z\oagnsvector_{\psi}}
=
\lim_{L\to\infty}\fun{\oastate[\psi]}{A B_L}
=
\fun{\oastate[\psi]}{A}
\fun{\oastate[\psi]}{B}
=
\bkt{\oarepn_{\psi}(A^*)\oagnsvector_{\psi}}
{\fun{\oastate[\psi]}{B} \oagnsvector_{\psi}}.$$
By density of the dense algebra $\oa{A}_{\txtloc}$ and cyclicity of the GNS vector,
$Z
\oagnsvector_{\psi}
=
\fun{\oastate[\psi]}{B}
\oagnsvector_{\psi}$.
If the state is primary, then
$\oacenter(\oa{M}_{\psi})
=
\fldcmp\idone_{\sphilb{H}_{\psi}}$, and hence the central operator $Z$ equals
$\fun{\oastate[\psi]}{B}
\idone_{\sphilb{H}_{\psi}}$.
\end{proof}

\begin{thm}[Araki--Haag--Kastler--Takesaki type criterion for the impossibility of BEC]\label{expedition0012584}
Consider the KMS state $\oastate[\psi_{\txtspinboson,\sminvtemperature}]$ constructed in Proposition
\ref{expedition0012485}.
Assume that the physical test-function space
$\sphilb{D}_{\txtbsn,\txtphys,\sminvtemperature}$ separates the zero mode in the sense of Definition
\ref{expedition0012134}.
Assume further the following.
\begin{enumerate}
\item
The $\lp^{1}$ asymptotic commutativity of Definition \ref{expedition0012581} holds on a suitable local
subalgebra of the resolvent algebra.

\item
On the same local subalgebra, $\oastate[\psi_{\txtspinboson,\sminvtemperature}]$ satisfies the weak
clustering of Definition \ref{expedition0012581}.

\item
The state $\oastate[\psi_{\txtspinboson,\sminvtemperature}]$ is primary.

\item
For the two-point function of the centered field in Theorem \ref{expedition0012136}, for every $f,g
\in \sphilb{D}_{\txtbsn,\txtphys,\sminvtemperature}$, the Cesàro average obtained from weak clustering is
$0$, and its pointwise long-distance limit has the same value as in Proposition \ref{expedition0012500},
namely $\onehalf\fun{\opform{q}_{\txtbsn,0,\sminvtemperature}}{f,g}$.
\end{enumerate}
Then $\opform{q}_{\txtbsn,0,\sminvtemperature}
=
0$ holds.
In particular,
$\smnumberdensity_{\txtbsn,0}(\sminvtemperature)
=
0$, and BEC is impossible in the sense of Theorem \ref{expedition0012132}.
\end{thm}

\begin{proof}
Take arbitrary $f,g
\in \sphilb{D}_{\txtbsn,\txtphys,\sminvtemperature}$.
By assumption (4), the Cesàro-averaged long-distance correlation of the centered field in Theorem
\ref{expedition0012136} is $0$.
On the other hand, the pointwise two-point off-diagonal long-range order in Proposition
\ref{expedition0012500} detects the same zero-mode form also for the centered field, and its limiting
value is $\onehalf\fun{\opform{q}_{\txtbsn,0,\sminvtemperature}}{f,g}$.
When the pointwise limit exists, the Cesàro limit has the same value, and hence
$\fun{\opform{q}_{\txtbsn,0,\sminvtemperature}}{f,g}
=
0$ follows.
Since $f,g$ are arbitrary, $\opform{q}_{\txtbsn,0,\sminvtemperature}
=
0$.
The zero-mode separating property and Theorem \ref{expedition0012132} give
$\smnumberdensity_{\txtbsn,0}(\sminvtemperature)
=
0$ and the impossibility of BEC.

Assumption (1) gives asymptotic centrality of the order parameter by Proposition
\ref{expedition0012582}.
Assumptions (2) and (3) scalarize the central limit in the primary representation by Proposition
\ref{expedition0012583}.
This is consistent with the zero-mode vanishing criterion above.
\end{proof}

\begin{rem}[Relation to Theorem \ref{expedition0012132}]\label{expedition0012585}
From the viewpoint of stability, Bratteli--Robinson II \cite[Section 5.4.2]{BratteliRobinson2} and the
Araki--Haag--Kastler--Takesaki type theorem \cite{ArakiHaagKastlerTakesaki1} allow
$\lp^{1}$ asymptotic commutativity and weak clustering to be regarded as sufficient conditions for the
absence of off-diagonal long-range order.
In particular, Theorem \ref{expedition0012584} is not an alternative proof of Theorem
\ref{expedition0012132}: the off-diagonal long-range order in Proposition \ref{expedition0012500} is a
model-specific calculation obtained from the functional integral representation.
The impossibility of BEC is not obtained directly from the operator-algebraic conditions alone.
Rather, combining the operator-algebraic conditions with Proposition \ref{expedition0012500} is regarded
as verifying the hypotheses of Theorem \ref{expedition0012132}.

Conversely, in a phase satisfying $\opform{q}_{\txtbsn,0,\sminvtemperature}
\neq 0$, one may have failure of weak clustering, non-primary behavior, a nonscalar central limit of the
order parameter, or survival of BEC directions as defined in Subsection \ref{expedition0012586}.
This possibility is consistent with Proposition \ref{expedition0012544}, Theorem \ref{expedition0012550},
and Theorem \ref{expedition0012132}.
\end{rem}

\subsection{BEC Directions and Zero-Mode Vanishing}\label{expedition0012586}

Even after the infrared directions are removed by quotienting, there may remain directions inside the physical one-particle space \(\sphilb{D}_{\txtbsn,\txtphys,\sminvtemperature}\) on which the zero mode \(\opform{q}_{\txtbsn,0,\sminvtemperature}\) has a positive value. These are the BEC directions. However, by Proposition \ref{expedition0012126}, the GNS representation of the KMS state on \(\oa{A}_{\txtspinboson,\txtphys,\sminvtemperature}\) is faithful, so for a fixed \(f
\in
\sphilb{D}_{\txtbsn,\txtphys,\sminvtemperature}\) the resolvent \(\oaresolvent(\lambda,f)\) does not immediately fall into the kernel by itself. What is meaningful for BEC is the decay of expectations when the scale is enlarged along such a direction.

The free-Bose-gas type \(c\)-number substitution and the van-Hove-model type linear shift in the preceding section are diagnostics measuring to which deterministic value \(\mathsf{Y}_{\sminvtemperature,f}\) degenerates. The BEC ideal of the resolvent algebra is not an ideal used to represent that deterministic shift as a quotient. It records only the directions inside the physical space on which the zero-mode quasi-bilinear form has a positive value. If the zero-mode vanishing condition is imposed, then Corollary \ref{expedition0012133} and Theorem \ref{expedition0012132} give \(\fun{\opform{q}_{\txtbsn,0,\sminvtemperature}}{f}
=
0\), and hence the BEC ideal becomes trivial. This trivialization is a condition separate from the fluctuations of \(\mathsf{Y}_{\sminvtemperature,f}\) and from the presence or absence of deterministic shifts.

\begin{defn}
Define the BEC directions inside the physical one-particle space by
$$\sphilb{X}_{\txtbsn,\txtbec,\sminvtemperature}
=
\set{f \in \sphilb{D}_{\txtbsn,\txtphys,\sminvtemperature}}
{\fun{\opform{q}_{\txtbsn,0,\sminvtemperature}}{f}>0}$$
and define the BEC ideal of the Bose algebra of physical quantities
$\oaresolventalgebra(\sphilb{D}_{\txtbsn,\txtphys,\sminvtemperature},\sigma)$ by
$$\oaideal{J}_{\txtbsn,\txtbec}
=
\clos{
\opideal
\set{\oaresolvent(\lambda,f)}
{\lambda \in \fldreal \setminus \setone{0},\,
f \in \sphilb{X}_{\txtbsn,\txtbec,\sminvtemperature}}}
\subset
\oaresolventalgebra(\sphilb{D}_{\txtbsn,\txtphys,\sminvtemperature},\sigma).$$
For the algebra of all physical quantities, set
$$\oaideal{J}_{\txtspinboson,\txtbec}
=
M_2 \otimes \oaideal{J}_{\txtbsn,\txtbec}
\subset
\oa{A}_{\txtspinboson,\txtphys,\sminvtemperature}.$$
\end{defn}

\begin{prop}\label{expedition0012129}
For every $f
\in
\sphilb{X}_{\txtbsn,\txtbec,\sminvtemperature}$ and
$\lambda \in \fldreal \setminus \setone{0}$,
$\lim_{\abs{t} \to \infty}
\fun{\oastate[\psi_{\txtspinboson,\sminvtemperature}^{\mathrm{res}}]}
{\oaresolvent(\lambda,tf)}
=
0$ holds.
\end{prop}

\begin{proof}
If $f
\in \sphilb{X}_{\txtbsn,\txtbec,\sminvtemperature}$, then
$\fun{\opform{q}_{\txtbsn,0,\sminvtemperature}}{f}
> 0$ holds.
Therefore, for every $t
\in \fldreal$,
$$\begin{aligned}
\fun{\opform{q}_{\txtbsn,\txtbec,\sminvtemperature}}{tf}
=
t^2 \fun{\opform{q}_{\txtbsn,\txtbec,\sminvtemperature}}{f}
\geq
t^2 \fun{\opform{q}_{\txtbsn,0,\sminvtemperature}}{f}
\to
\infty
\quad
(\abs{t}\to\infty)
\end{aligned}$$
holds.
Applying Proposition \ref{expedition0012104} with $f_n
= t_n f$ and $\abs{t_n}
\to \infty$ gives the desired conclusion.
\end{proof}

\begin{prop}\label{expedition0012127}
Assume that the zero-mode vanishing condition of Definition \ref{expedition0012100} holds on
$\sphilb{D}_{\txtbsn,\txtphys,\sminvtemperature}$.
Then
$\sphilb{X}_{\txtbsn,\txtbec,\sminvtemperature}
=
\emptyset$ holds.
In particular, for the ideals,
$\oaideal{J}_{\txtbsn,\txtbec}
=
\setone{0}$ and
$\oaideal{J}_{\txtspinboson,\txtbec}
=
\setone{0}$ hold.
\end{prop}

\begin{proof}
If $f
\in \sphilb{X}_{\txtbsn,\txtbec,\sminvtemperature}$, then
$f
\in \sphilb{D}_{\txtbsn,\txtphys,\sminvtemperature}$ and
$\fun{\opform{q}_{\txtbsn,0,\sminvtemperature}}{f}>0$.
However, by Definition \ref{expedition0012100}, for every $f
\in \sphilb{D}_{\txtbsn,\txtphys,\sminvtemperature}$ one must have
$\fun{\opform{q}_{\txtbsn,0,\sminvtemperature}}{f}
= 0$.
Thus
$\sphilb{X}_{\txtbsn,\txtbec,\sminvtemperature}
=
\emptyset$.
The vanishing of the ideals follows from the definition.
\end{proof}

\begin{rem}
When the KMS state is taken as the base point, the only essential ideal in the resolvent algebra is
$\oaideal{J}_{\txtspinboson,\txtirsingular}$, which kills the infrared-divergent directions.
The KMS state on the physical quotient algebra $\oa{A}_{\txtspinboson,\txtphys,\sminvtemperature}$ is
regular and its GNS representation is faithful, so no additional state-dependent kernel appears in the
representation space.

The BEC problem is not the creation of another abstract quotient.
It is the question of whether directions satisfying
$\fun{\opform{q}_{\txtbsn,0,\sminvtemperature}}{f}
> 0$ remain inside the physical space $\sphilb{D}_{\txtbsn,\txtphys,\sminvtemperature}$.
By Proposition \ref{expedition0012129}, even if such directions remain, what is observed is not the
vanishing of fixed generators but the decay of resolvent expectations in the large-amplitude limit.
It is therefore natural to regard $\oaideal{J}
_{\txtbsn,\txtbec}$ not as an essential quotient ideal like the infrared ideal
$\oaideal{J}
_{\txtbsn,\txtirsingular}$ but as an auxiliary ideal recording the condensate directions that remain
inside the physical space.
If the zero-mode vanishing condition is assumed, such directions are empty by Proposition
\ref{expedition0012127}; consequently
$\oaideal{J}_{\txtbsn,\txtbec}
=
\setone{0}$.
Furthermore, the free-Bose-gas type $c$-number substitution and the van-Hove-model type linear shift are
auxiliary comparisons measuring the degenerate value of the spin random variable
$\mathsf{Y}_{\sminvtemperature,f}$.
Therefore a nonzero deterministic shift is not identified with the presence or absence of BEC directions.
\end{rem}

\bibliography{myref.bib}

\end{document}

%% file: mycommands.tex
\makeatletter
\providecommand*{\dashv}{\mathrel{\mathpalette\@Dashv\vDash}}
\newcommand*{\@dashv}[2]{\reflectbox{$\m@th#1#2$}}
\makeatother
\newcommand{\abscard}[1]{\abs{#1}} % 絶対値記号による濃度の記号
\newcommand{\abs}[1]{\left| #1 \right|} % 絶対値
\NewDocumentCommand{\weaknorm}{O{\dbk} m}{#1{#2}} % 弱(準)ノルム
\newcommand{\norm}[1]{\left\Vert #1 \right\Vert} % ノルム
\newcommand{\onenorm}[1]{\norm{#1}_1} % 1-ノルム
\newcommand{\bkt}[2]{\left\langle #1,\,#2 \right\rangle} % 内積
\newcommand{\rbkt}[2]{\left( #1,\,#2 \right)} % 内積に利用
\newcommand{\slim}{\mathrm{s} \hyphen \lim} % 強極限
\newcommand{\isomto}{\mathrel{\rightarrowtail\kern-1.9ex\twoheadrightarrow}} % 全単射, 同型射, ホモロジー代数・圏論でよく使う
\newcommand{\cbk}[1]{\left\{ #1 \right\}} % 中括弧
\newcommand{\dbk}[1]{\left\langle #1 \right\rangle} % <f>形式の括弧, Dirac bracketの意味でdbk
\newcommand{\pairbk}[1]{\rbk{#1}} % 適当な対のための括弧
\newcommand{\rbkleft}[1]{\left( #1 \right.} % 片側括弧
\newcommand{\rbkright}[1]{\left. #1 \right)} % T片側括弧
\newcommand{\rbk}[1]{\left( #1 \right)} % 丸括弧, 共通
\newcommand{\sqbkleft}[1]{\left[ #1 \right.} % sqbkに対する片側括弧
\newcommand{\sqbkright}[1]{\left. #1 \right]} % sqbkに対する片側括弧
\newcommand{\sqbk}[1]{\left[ #1 \right]} % いわゆる大括弧
\newcommand{\funrbk}[2]{\fun{\rbk{#1}}{#2}} % 関数の部分に丸括弧を入れてまとめる
\newcommand{\fun}[2]{#1 \rbk{#2}} % 関数テンプレート
\newcommand{\sqfun}[2]{#1 \sqbk{#2}} % 角括弧を使う関数, 確率論の期待値や完全な熱力学関数に利用
\newcommand{\closedinterval}[2]{\sqbk{#1,\,#2}} % 閉区間
\newcommand{\commutator}[2]{\sqbk{#1,\,#2}} % 交換子
\NewDocumentCommand{\imunit}{O{\mathsf{i}}}{#1} % 虚数単位, 共通
\NewDocumentCommand{\placeholder}{O{\bullet}}{#1} % 関手やホモロジーなどでハイフンや点で表す要素
\NewDocumentCommand{\trace}{O{\operatorname{Tr}}}{#1} % トレース
\newcommand{\Ker}{\operatorname{Ker}} % 核の集合
\newcommand{\bigmiddleslash}[2]{\left. #1 \middle/ #2 \right.} % 主に商集合\setquotに利用
\newcommand{\clos}[1]{\overline{#1}} % 外測度など何らかの意味で拡張・閉包のような趣があるものの, 具体的に「---閉包」のような名前がない対象に利用
\newcommand{\cmpconj}[1]{\overline{#1}} % 複素共役
\newcommand{\diracdelta}{\delta} % ディラックのデルタ関数
\newcommand{\dom}{\operatorname{dom}} % 定義域
\newcommand{\eqcsq}[1]{\sqbk{#1}} % 同値類の記号
\newcommand{\fml}[2]{\cbk{#1}_{#2}} % 族の基本系
\newcommand{\hyphen}{\hbox{-}} % 記号定義の補助
\newcommand{\idone}{1} % 1で表した恒等写像
\newcommand{\inv}[1]{#1^{-1}} % 逆像・逆写像, \invもよく使うため独立させる
\newcommand{\napiernum}{\mathsf{e}} % ネイピア数, 自然対数の底, 素電荷など他にもeが出てきて紛らわしい場合に利用
\newcommand{\onehalf}{\frac{1}{2}} % 1/2
\newcommand{\oneoverfour}{\frac{1}{4}} % 1/4
\newcommand{\opideal}{\operatorname{Ideal}} % イデアル
\newcommand{\opimag}{\operatorname{Im}} % 複素数の虚部
\newcommand{\opod}[1]{\frac{d}{d #1}} % 一階の常微分作用素
\newcommand{\oppd}[1]{\frac{\partial}{\partial #1}} % 偏微分作用素
\newcommand{\opreal}{\operatorname{Re}} % 複素数の実部
\newcommand{\setSymbolDownLeft}[2]{{\vphantom{#2}}_{#1}{#2}} % 左下に添字をつけるvphantomのエイリアス
\newcommand{\setSymbolUpLeft}[2]{{\vphantom{#2}}^{#1}{#2}} % 左上に添字をつけるvphantomのエイリアス
\DeclareMathOperator{\eqalgisom}{\cong} % 代数的な同型
\NewDocumentCommand{\agvariety}{O{\mathcal}}{#1} % 代数多様体, 特にイデアルが生成する代数多様体
\NewDocumentCommand{\cmdrel}{O{\omega}}{#1} % 分散関係
\NewDocumentCommand{\dfsp}{O{A}}{#1} % 微分形式の空間
\NewDocumentCommand{\eqcpointed}{O{\eqcsq} m}{#1{#2}_{\ast}} % 基点つき同値類
\NewDocumentCommand{\fnheaviside}{O{H}}{#1} % ヘビサイド関数
\NewDocumentCommand{\fthol}{O{\mathcal{O}}}{#1} % 正則関数の空間用
\NewDocumentCommand{\ftmero}{O{\mathcal{M}}}{#1} % 有理型関数の空間用
\NewDocumentCommand{\grcentralizer}{O{Z}}{#1} % 中心化群, expedition0005395
\NewDocumentCommand{\grmetform}{O{2} m m}{\grmet[#1] \! \rbkt{#2}{#3}} % (一般)相対性理論, 計量の二次形式
\NewDocumentCommand{\grmet}{O{2}}{\setSymbolDownLeft{#1}{g}} % (一般)相対性理論, 計量の文字
\NewDocumentCommand{\grnormalizer}{O{N}}{#1} % 正規化群, expedition0005395
\NewDocumentCommand{\gropasym}{O{A}}{#1} % 反対称化作用素, expedition0001365
\NewDocumentCommand{\gropsym}{O{S}}{#1} % 対称化作用素, expedition0001365
\NewDocumentCommand{\grpermorderedpair}{O{\mathcal{P}}}{#1} % 偶数個の要素の異なる順序対への分解の全体, expedition0009453
\NewDocumentCommand{\grsym}{O{\mathfrak{S}} m}{#1_{#2}} % 対称群, expedition0001319
\NewDocumentCommand{\gtbase}{O{\mathcal}}{#1} % 位相の基底, 開集合の基底, expedition0000259
\NewDocumentCommand{\gtfilter}{O{\mathcal}}{#1} % フィルター, expedition0000349
\NewDocumentCommand{\gtfmlclosed}{O{\mathcal}}{#1} % 閉集合系だけではなく, ある集合上の一般の閉集合族を表す, expedition0000227
\NewDocumentCommand{\gtfmlopen}{O{\mathcal}}{#1} % 「位相」の意味での開集合系だけではなく, ある集合上の一般の開集合族を表す, expedition0000227
\NewDocumentCommand{\gtopenball}{O{U}}{#1} % 開球
\NewDocumentCommand{\gtopencover}{O{\mathcal}}{#1} % 開被覆
\NewDocumentCommand{\gtopennbh}{O{\mathcal}}{#1} % 開近傍系
\NewDocumentCommand{\gtpreopencover}{O{\mathcal}}{#1} % 内部を取ると開被覆になる集合の族, expedition0009107
\NewDocumentCommand{\gtsubbase}{O{\mathcal}}{#1} % 位相の準基, expedition0005835
\NewDocumentCommand{\gtvicinity}{O{\mathcal}}{#1} % 一様構造での近縁
\NewDocumentCommand{\lasp}{O{\mathcal}}{#1} % 線型空間
\NewDocumentCommand{\latprightrbk}{O{\top} m}{\rbk{#2}^{#1}} % 転置, expedition0009877
\NewDocumentCommand{\latpright}{O{\top} m}{#2^{#1}} % 転置, expedition0009877
\NewDocumentCommand{\latp}{O{t} m}{\setSymbolUpLeft{#1}{#2}} % 転置, expedition0009877
\NewDocumentCommand{\lpdistribution}{O{\mu} m}{#2_{\ast,#1}} % 分布関数, expedition0010310
\NewDocumentCommand{\lpmollifier}{O{\rho}}{#1} % フリードリクスの軟化子
\NewDocumentCommand{\lpofpositive}{O{\chi}}{#1} % 正型関数
\NewDocumentCommand{\manliederiv}{O{L}}{#1} % リー微分
\NewDocumentCommand{\mansmoothnbh}{O{\mathcal}}{#1} % 多様体上の滑らかな開近傍の族
\NewDocumentCommand{\mblfmldsysgenerated}{O{d} m}{\fun{#1}{#2}} % 生成された$d$-系
\NewDocumentCommand{\mblfmlgenerated}{O{\sigma} m}{\fun{#1}{#2}} % 生成された加法族
\NewDocumentCommand{\oacorrfn}{O{\Gamma}}{#1} % 汎関数積分表示とも関わる相関関数, expedition0011645
\NewDocumentCommand{\oagnsvector}{O{\Omega}}{#1} % GNS表現のベクトル状態
\NewDocumentCommand{\oaideal}{O{\mathcal}}{#1} % 作用素環のイデアル：環や代数のイデアルとは分けておく
\NewDocumentCommand{\oanumberoperator}{O{A}}{#1} % レゾルベント環での個数作用素
\NewDocumentCommand{\oaposcone}{O{\mathcal{P}}}{#1} % 冨田-竹崎理論での正錐
\NewDocumentCommand{\oapressure}{O{P}}{#1} % トレースによる通常の圧力は大文字でP: expedition0009842, トレース状態による修正圧力は小文字でp: expedition0009797
\NewDocumentCommand{\oarepn}{O{\pi}}{#1} % 作用素環の表現
\NewDocumentCommand{\oaspnormalstate}{O{N}}{#1} % 作用素環上の正規状態の空間
\NewDocumentCommand{\oasppurestate}{O{P}}{#1} % 作用素環上の純粋状態の空間
\NewDocumentCommand{\oaspstate}{O{E}}{#1} % 作用素環の状態がなす空間
\NewDocumentCommand{\oastatevector}{O{\Omega}}{#1} % GNS表現に付随する状態ベクトル
\NewDocumentCommand{\oastate}{O{\omega}}{#1} % 作用素環の状態
\NewDocumentCommand{\opdilation}{O{\delta}}{#1} % スケーリング(伸張)の作用素, expedition0010830
\NewDocumentCommand{\opdmat}{O{\rho}}{#1} % 密度行列(密度作用素)
\NewDocumentCommand{\opfockan}{O{a}}{#1} % 場の量子論での消滅作用素
\NewDocumentCommand{\opfockcran}{O{a}}{#1^{\#}} % 場の量子論での生成・消滅作用素をまとめて表す記法
\NewDocumentCommand{\opfockcrdagger}{O{a}}{#1^{\dagger}} % 場の量子論での生成作用素のバージョン
\NewDocumentCommand{\opfockcr}{O{a}}{#1^{\ast}} % 場の量子論での生成作用素
\NewDocumentCommand{\opfocknumber}{O{N}}{#1} % 個数作用素
\NewDocumentCommand{\opfocksegalconj}{O{\pi}}{#1} % 場の量子論でのシーガルの場の作用素, expedition0009962
\NewDocumentCommand{\opfocksegal}{O{\phi}}{#1} % 場の量子論でのシーガルの場の作用素
\NewDocumentCommand{\opspecmeas}{O{E}}{#1} % スペクトル測度の標準的な記号
\NewDocumentCommand{\opspec}{O{} m}{\fun{\sigma_{#1}}{#2}} % 作用素のスペクトル, expedition0001551
\NewDocumentCommand{\optransl}{O{\tau}}{#1} % 平行移動の作用素, expedition0010829
\NewDocumentCommand{\physaction}{O{\mathcal{A}}}{#1} % 物理の作用
\NewDocumentCommand{\physcharge}{O{e}}{\mathrm{#1}} % 電荷, 素電荷などの物理量を表す. e以外にもqに利用する
\NewDocumentCommand{\physcplconst}{O{\mathsf{g}}}{#1} % 結合定数, coupling constant
\NewDocumentCommand{\physelectrostaticcapasity}{O{\mathrm{Cap}}}{#1} % 静電容量, expedition0010917
\NewDocumentCommand{\physenergy}{O{E}}{#1} % エネルギー
\NewDocumentCommand{\physgse}{O{E}}{#1_{0}} % 基底エネルギー, ground state energy
\NewDocumentCommand{\physham}{O{H}}{#1} % ハミルトニアンの雛形, たいていは下つき添字で水素原子, 自由電磁場などを表す
\NewDocumentCommand{\physlagdensity}{O{\mathcal{L}}}{#1} % ラグランジアン密度の雛形
\NewDocumentCommand{\physlag}{O{L}}{#1} % ラグランジアンの雛形
\NewDocumentCommand{\physliouvilean}{O{L}}{#1} % リウビリアンの雛形
\NewDocumentCommand{\physmass}{O{m}}{#1} % 物質の質量
\NewDocumentCommand{\prbbrownmv}{O{B}}{#1} % ブラウン運動
\NewDocumentCommand{\prbcharfun}{O{\chi}}{#1} % 特性関数
\NewDocumentCommand{\prbdist}{O{\mathcal{P}}}{#1} % 確率分布
\NewDocumentCommand{\prbgaussianmeasure}{O{\msrcal{N}}}{#1} % ガウス測度
\NewDocumentCommand{\prbnormaldist}{O{N}}{#1} % 正規分布
\NewDocumentCommand{\prbpoissonprocess}{O{N}}{#1} % ポアソン過程
\NewDocumentCommand{\prbprocess}{O{X}}{#1} % 確率過程
\NewDocumentCommand{\prbqspace}{O{\mathcal{Q}}}{#1} % $Q$-表示の土台の集合を表す
\NewDocumentCommand{\prbspsample}{O{\Omega}}{#1} % 確率空間に対する標本空間, expedition0010114
\NewDocumentCommand{\psh}{O{\mathfrak}}{#1} % 前層
\NewDocumentCommand{\qtquantumchannel}{O{\mathcal{L}}}{#1} % 量子通信路, expedition0005561, expedition0006844
\NewDocumentCommand{\repn}{O{\pi}}{#1} % 表現
\NewDocumentCommand{\schattencls}{O{\mathbb{K}}}{#1} % コンパクト作用素またはシャッテンクラス
\NewDocumentCommand{\setfmlcylinder}{O{\mathcal{C}}}{#1} % 筒集合族, expedition0000313, expedition0004973
\NewDocumentCommand{\setfml}{O{\mathcal}}{#1} % 集合族
\NewDocumentCommand{\setindex}{O{\mathcal} m}{#1{#2}} % 添字集合
\NewDocumentCommand{\setlattice}{O{\Gamma}}{#1} % 格子, expedition0010278
\NewDocumentCommand{\setspecial}{O{\mathcal} m}{#1{#2}} % 特別な集合を表すための記号, 一点集合・二点集合などを太字で表すといった用途
\NewDocumentCommand{\shdiffform}{O{\sheaf{A}}}{#1} % 微分形式がなす層
\NewDocumentCommand{\sheaf}{O{\mathfrak}}{#1} % 層の記号
\NewDocumentCommand{\smchemicalpotential}{O{\mu}}{#1} % 化学ポテンシャル
\NewDocumentCommand{\smenergydensity}{O{\varrho}}{#1} % エネルギー密度, expedition0011356
\NewDocumentCommand{\smfluctuationwithdmat}{O{\beta} m}{\smuncertaintywithdmat[#1]{#2}^2} % 密度行列に付随する正規状態でのゆらぎ, expedition0011059
\NewDocumentCommand{\sminvtemperature}{O{\beta}}{#1} % 逆温度
\NewDocumentCommand{\smlocaldensityoperator}{O{\rho}}{#1} % 大正準状態に対する局所密度作用素, expedition0011294
\NewDocumentCommand{\smmicrocanonicalstate}{O{\beta} m}{\physmean{#2}_{#1}} % 密度行列に付随する正規状態, expedition0011059
\NewDocumentCommand{\smnumberdensity}{O{\rho}}{#1} % 粒子数密度, expedition0011351
\NewDocumentCommand{\smooth}{O{\mathcal{E}}}{#1} % 滑らかな関数の空間, 特に微分形式関連でも利用, expedition0010649
\NewDocumentCommand{\smparticlenumber}{O{N}}{#1} % 粒子粒子数, expedition0011351
\NewDocumentCommand{\smpressure}{O{p}}{#1} % 統計力学での圧力, expedition0011341
\NewDocumentCommand{\smspecificfreeenergy}{O{\bar{f}}}{#1} % (ヘルムホルツの)比自由エネルギー
\NewDocumentCommand{\smthermalvac}{O{\beta}}{\Omega_{#1}} % 熱的真空, expedition0011060
\NewDocumentCommand{\smuncertaintywithdmat}{O{\beta} m}{\rbk{\triangle #2}_{#1}} % 密度行列に付随する正規状態での不確かさ, expedition0011059
\NewDocumentCommand{\sphilb}{O{\mathcal}}{#1} % ヒルベルト空間
\NewDocumentCommand{\splowerhalf}{O{\mathbb{H}}}{#1_{\txtneg}} % 開下半空間
\NewDocumentCommand{\spupperhalf}{O{\mathbb{H}}}{#1_{\txtnonneg}} % 開上半空間
\NewDocumentCommand{\topmetric}{O{d}}{#1} % 位相空間上の距離関数
\NewDocumentCommand{\vaoutnormal}{O{\widehat}}{#1} % 外向き外法線
\newcommand{\category}[1]{\mathop{\mathsf{#1}}} % 圏を生成する
\newcommand{\catpresheaf}[1]{\category{PSh}} % 前層がなす圏
\newcommand{\dstrapiddec}{\mathcal{S}} % 急減少関数の空間
\newcommand{\dsttempered}{\dstrapiddec^{\ast}} % 緩増加超関数の空間
\newcommand{\faadjpresharp}[1]{#1_{\#}} % 行列・作用素に対する前共役(前双対)
\newcommand{\faadj}[1]{#1^{\ast}} % 行列・作用素に対する共役・随伴. 線型代数でも内積の必要性からこれに統合. 双対空間にはdualをあてる
\newcommand{\faftr}[1]{\widehat{#1}} % フーリエ変換
\newcommand{\fldcmp}{\fld{C}} % 複素数体
\newcommand{\fldmultiplicativegroup}[1]{#1^{\times}} % 体の乗法群
\newcommand{\fldreal}{\fld{R}} % 実数体
\newcommand{\fld}[1]{\mathbb{#1}} % 実数体・複素数体・四元数体, またはこれらをまとめて表す体用のコマンド
\newcommand{\fndef}[1]{\boldsymbol{1}_{#1}} % 定義関数
\newcommand{\fnexp}[1]{\fun{\exp}{#1}} % 指数関数
\newcommand{\fnrestr}[2]{\left. #1 \right|_{#2}} % 写像の制限
\newcommand{\greuctr}[2]{\fldreal_{#1}^{#2}} % ユークリッド空間の空間並進群
\newcommand{\gtclos}[1]{\overline{#1}} % 一般位相の閉包
\newcommand{\lp}{L} % ルベーグ空間, L_{\txtloc}^p(X), L^p(X)
\newcommand{\mblfmlborel}{\mblfml{B}} % ボレル集合族
\newcommand{\mblfmlfrak}[1]{\mathfrak{#1}} % 加法族の mathfrak 表記
\newcommand{\mblfml}[1]{\mathcal{#1}} % 可測集合族に関わる記号の基礎
\newcommand{\mblfn}{M} % 可測関数
\newcommand{\monnat}{\mathbb{N}} % 0を含めた自然数全体がなすモノイド
\newcommand{\msrcal}[1]{\mathcal{#1}} % mathcalで表す測度
\newcommand{\msrlaw}{\mathrm{Law}} % 法則
\newcommand{\msrprb}{\mathrm{Pr}} % たまに役立つであろう, 立体のPrで表す確率測度
\newcommand{\msr}[1]{#1} % 測度の識別用
\newcommand{\nfoldvar}[2]{\underline{#1}_{#2}} % #2個の変数をまとめて下線つきで表す
\newcommand{\oacenter}{\mathcal{Z}} % 作用素環の中心
\newcommand{\oacommutant}[1]{#1^{\prime}} % 作用素環での可換子環
\newcommand{\oacstar}{C^{\ast}} % $C^{\ast}$-環
\newcommand{\oadoublecommutant}[1]{#1^{\prime \prime}} % 二重可換子環
\newcommand{\oaresolventalgebra}{\oa{R}} % レゾルベント環
\newcommand{\oaresolvent}{R} % レゾルベント環の要素としてのレゾルベント
\newcommand{\oastaralgebra}{\operatorname{\ast \hyphen \mathrm{alg}}} % 指定した集合が生成する$\ast$-環
\newcommand{\oaweyl}{\oa{W}} % 抽象的なワイル環
\newcommand{\oa}[1]{\mathcal{#1}} % 作用素環を表す記号用
\newcommand{\opdmsr}[1]{\mathop{d #1}} % 積分で現れる測度用のコマンド, dxのような単純な対象も含む
\newcommand{\opfocksndqntdiff}{d \Gamma} % 第一種の第二量子化作用素, 微分版
\newcommand{\opfocksndqnt}{\Gamma} % 第二種の第二量子化作用素
\newcommand{\opfockvac}{\Omega} % フォック真空
\newcommand{\opfockweyl}{W} % フォック空間上のワイル作用素
\newcommand{\opformdomain}{\mathop{Q}} % 準双線型形式の定義域, expedition0009907
\newcommand{\opform}[1]{\mathsf{#1}} % 二次形式
\newcommand{\opspecint}[1]{\mathcal{E}} % スペクトル積分, expedition0001621
\newcommand{\physmean}[1]{\dbk{#1}} % 物理, 特に統計力学によく出てくる平均の記号
\newcommand{\prbexp}{\mathbb{E}} % 期待値, 条件つき期待値にも応用
\newcommand{\qmpaulispin}{\sigma} % パウリ行列
\newcommand{\ringratint}{\mathbb{Z}} % 有理整数環
\newcommand{\seq}[2]{\if\relax\detokenize{#1}\relax \rbk{#1} \else \rbk{#1}_{#2} \fi} % 点列の基本系
\newcommand{\setisomorphism}[1]{\operatorname{Iso}} % 同型写像の集合
\newcommand{\setone}[1]{\cbk{#1}}
\newcommand{\setquot}[2]{\bigmiddleslash{#1}{#2}} % 商集合
\newcommand{\set}[2]{\left\{#1 \, \middle| \, #2\right\}}
\newcommand{\smpartitionfunc}{Z} % 分配関数
\newcommand{\spfock}{\mathcal{F}} % フォック空間
\newcommand{\splinspan}{\operatorname{span}} % 線型包
\newcommand{\txtbec}{\mathrm{BEC}} % BEC, ボース-アインシュタイン凝縮
\newcommand{\txtborel}{\mathrm{b}} % ボレル可測性に関わる対象を指す
\newcommand{\txtbsn}{\mathrm{b}} % ボソンのb
\newcommand{\txteuclid}{\mathrm{E}} % ユークリッド, ユークリッド化の添字
\newcommand{\txtfock}{\mathrm{F}} % フォック
\newcommand{\txtfr}{\mathrm{fr}} % 自由場などを表すfr
\newcommand{\txtirsingular}{\mathrm{irs}} % 赤外特異
\newcommand{\txtloc}{\mathrm{loc}} % 局所
\newcommand{\txtneg}{\mathrm{-}} % 負の部分
\newcommand{\txtnonneg}{\mathrm{+}} % 正または非負の部分
\newcommand{\txtnonzero}{\mathrm{nz}} % 非零, 非零モードを表す
\newcommand{\txtphys}{\mathrm{phys}} % 物理
\newcommand{\txtpoisson}{\mathrm{poisson}} % ポアソン
\newcommand{\txtreal}{\mathrm{real}} % 実
\newcommand{\txtspinboson}{\mathrm{SB}} % スピン-ボソン模型
\newcommand{\txtspin}{\mathrm{spin}} % スピン
\newcommand{\txtsym}{\mathrm{s}} % 対称, テンソル積に使う